\documentclass[12pt]{article}
\usepackage{booktabs}
\usepackage{multirow}
\usepackage{geometry}
\usepackage{amssymb,amsmath, amsfonts, amsthm}
\usepackage{graphicx}
\usepackage{caption}
\usepackage{subcaption}
\usepackage{enumerate}
\usepackage{algorithm}
\usepackage[authoryear]{natbib}
\usepackage{enumitem}
\usepackage[utf8]{inputenc}
\usepackage{url}
\usepackage[most]{tcolorbox}
\usepackage{bm}
\usepackage{xcolor}
\usepackage{array}
\usepackage[unicode=true,pdfusetitle,
bookmarks=true,bookmarksnumbered=false,bookmarksopen=false,
breaklinks=false,pdfborder={0 0 1},colorlinks=true]{hyperref}
\hypersetup{linkcolor=red, citecolor=blue}

\newcommand{\blind}{1}

\newcommand\tablefigurenote[1]{\captionsetup{font=small}\caption*{Notes: #1}}

\newtheorem{lemma}{Lemma}[section]
\newtheorem{assumption}{Assumption}[section]
\newtheorem{proposition}{Proposition}[section]
\newtheorem{definition}{Definition}

\definecolor{internationalkleinblue}{rgb}{0.0, 0.18, 0.65}
\colorlet{bleu}{internationalkleinblue}

\newcommand*{\Reels}{\mathbb{R}}
\newcommand{\E}{\operatorname{E}}
\newcommand{\T}{\operatorname{T}}
\newcommand{\B}{\operatorname{B}}
\newcommand{\A}{\operatorname{A}}
\newcommand{\I}{\operatorname{I}}
\renewcommand{\P}{\operatorname{P}}

\newcommand*{\eps}{\varepsilon}

\DeclareMathOperator*{\argmin}{arg\,min}
\newcommand*{\betahat}{\widehat{\beta}}
\newcommand*{\hhat}{\widehat{h}}
\newcommand*{\alphahat}{\widehat{\alpha}}
\newcommand*{\thetahat}{\widehat{\theta}}
\newcommand*{\lambdahat}{\widehat{\lambda}}

\begin{document}
	
	\def\spacingset#1{\renewcommand{\baselinestretch}%
		{#1}\small\normalsize} \spacingset{1}
	
	
	\if1\blind
	{
		\title{\bf Average Marginal Effects in One-Step Partially Linear Instrumental Regressions\thanks{Elia Lapenta acknowledges funding from the French National Research Agency (ANR) under grant  ANR-23-CE26-0008-01.}}
		\author{
			Lucas Girard\thanks{CREST, CNRS, Ecole polytechnique, Groupe ENSAE-ENSAI, ENSAE Paris, Institut Polytechnique de Paris, Palaiseau, France. lucas.girard[at]ensae.fr}\and
			Elia Lapenta\thanks{University of Exeter and CREST. e.lapenta[at]exeter.ac.uk}
		}
		\date{April 2026}
		\maketitle
	} \fi
	
	\if0\blind
	{
		\bigskip
		\bigskip
		\bigskip
		\begin{center}
			{\LARGE\bf One-step  nonparametric instrumental regression using
				smoothing splines}
		\end{center}
		\medskip
	} \fi

	\bigskip
	\begin{abstract}
		We propose a novel procedure for estimating and conducting inference on average marginal effects in partially linear instrumental regressions using Reproducing Kernel Hilbert Space methods. Our procedure relies on a single regularization parameter. We obtain the consistency and asymptotic normality of our estimator. Since the variance of the limiting distribution has a complex analytical form, we propose a Bayesian bootstrap method to conduct inference and establish its validity. Our procedure is easy to implement and exhibits good finite-sample performance in simulations. Three empirical applications illustrate its implementation on real data, showing that it yields economically meaningful results.
	\end{abstract}
	
	\noindent%
	
	\spacingset{1.3}
	
	\section{Introduction}
	We consider the partially linear instrumental variables (IV) model 
	\begin{equation}
		\label{eq: npiv regression}
		Y = h_0(Z) + X^T\beta_0 + \varepsilon 
		\quad \text{with} \;
		\E\{ \varepsilon|W,X\} = 0 \,,
	\end{equation}
	where \(Y \in \Reels\) is the outcome, \(Z \in \Reels\) is the endogenous continuous treatment possibly correlated with the error term \(\eps \in \Reels\), \(X \in \Reels^{p}\) is a vector of exogenous covariates, 
	and \(W \in \Reels\) is a continuous instrument. 
	$\beta_0\in\mathbb{R}^p$ is a vector of coefficients, and the treatment or regression function \(h_0\) is nonparametrically specified. 
	We develop an estimation and inference method for the average marginal effect (AME) of the treatment, which is defined as 
	\begin{equation}
		\label{eq: definition of ame}
		\theta_0 := \E\{h'_0(Z)\} \,.
	\end{equation}
	The AME is a key object in many empirical applications, such as estimating the returns to education or measuring demand elasticities in industrial organization.
	If the treatment function $h_0$ were linear, the AME of the treatment would equal the slope coefficient of this function, and could be easily estimated by a two-stage least squares regression. However, a linearity or parametric assumption can rarely be justified on economic grounds and entails the risk of misspecification.
	When linearity is violated, the AME is not consistently estimated by two-stage least squares, and the resulting causal inference or counterfactual analyses can be misleading.
	Assuming a semiparametric IV model allows to overcome the risk of misspecification and to learn the treatment effects flexibly from the data, while mitigating the curse of dimensionality,  see
	\cite{florens2003inverse}, \cite{florens2012instrumental},  and \cite{ai2003efficient}.  
	
	Although non- and semiparametric approaches allow reducing the misspecification risk, inference for a nonparametric IV function is substantially more difficult than inference on a scalar parameter. Moreover, an estimated nonparametric regression remains less easy to interpret than an estimated linear model, where the coefficients can directly be viewed as average marginal effects of the variables and the treatment. For applied work, this is important because researchers often must report results in a form that policymakers and practitioners can interpret.
	Focusing on the AME in \eqref{eq: definition of ame} is appealing for two reasons: (i)~it provides an easy-to-interpret index of the treatment effect while addressing misspecification issues; (ii)~it enables  inference on a scalar parameter rather than on a nonparametric function.
	
	In this paper, we contribute to the literature by providing an inference procedure for the AME in \eqref{eq: definition of ame} that has two distinctive features. 
	First, it is based on a single regularization parameter. 
	Second, it builds on a framework widely used in the machine learning and Support Vector Machines literature, namely the Reproducing Kernel Hilbert Space setting.
	
	The first distinctive feature of our procedure is that it is based on a single regularization parameter. 
	Most of the existing literature estimates functionals of $h_0$, such as the AME, by running multi-step regressions. In the first step, they estimate the conditional expectation operator $\E\{\cdot|W, X\}$; in the second step, they use this estimate to obtain an estimator of the function $h_0$, which is then employed to construct an AME estimator.  
	See, e.g., \cite{ai2003efficient} and \cite{ai2007estimation}. 
	Such a multi-step regression approach requires selecting multiple regularization parameters, one for each regression. 
	This complicates the practical implementation of these methods, as the fine-tuning of each of these parameters might be difficult in practice. Moreover, each estimation step introduces an estimation error that impacts the finite sample behavior of the final estimator. 
	Unlike existing methods, we provide an estimation and inference procedure for the AME based on a one-step regression.
	Its main advantage is that it requires a single regularization parameter, which simplifies practical implementation.
	
	The second distinctive feature of our inference method is that it builds on a popular framework from the machine learning literature, the Reproducing Kernel Hilbert Space (RKHS) setting; see, e.g., \citet[Chapter 12]{wainwright2019high}, \cite{berlinet2011reproducing}, and \cite{steinwart2008support}. 
	Tools and algorithms from this framework have been extensively used in Support Vector Machine and classification methods.
	The advantages of constructing an inference procedure for the AME based on RKHSs are twofold.
	First, RKHS tools allow solving difficult computational problems in a tractable manner.
	In our context, they allow us to provide an estimator of the AME and a test statistic with simple and easy-to-compute expressions.
	Second, they have excellent finite-sample performances, as we show in our simulations, see Section~\ref{sec:small_sample_behavior_simulations}.
	
	We show that our AME estimator based on RKHSs is asymptotically normal, but the variance of the limiting distribution has a complex analytical form.  Thus, we develop an inference procedure relying on the Bayesian bootstrap and establish its validity. 
	Our method is simple to implement and we provide an \texttt{R} package that practitioners can readily use.
	In addition to its simplicity, we show in simulations that it yields good size control under the null hypothesis and exhibits good power under the alternative in small- and moderate-sample settings. 
	We also illustrate the potential of our method with three empirical applications based on \citet{angrist1999using}, \citet{frankel1999}, and \citet{sokullu2016semi}, with respective samples of 2,024, 150, and 117 observations. 
	These applications reveal the good performance of our estimation and inference procedure on real data, including small samples. 
	Thanks to its computational simplicity, its reliance on a single regularization parameter, and its good performance in finite samples, our method can be readily adopted by practitioners.
	
	\paragraph{Related literature.}
	Series-based methods for estimating AMEs and, more generally, functionals of IV regressions have been developed in \cite{ai2003efficient}, \cite{ai2007estimation}, and \cite{santos2011instrumental}. 
	In a recent paper, \cite{chen2023efficient} investigates the properties of several AME estimators based on Artificial Neural Networks, and study their finite sample performance. All these papers are based on a multi-step regression approach and require selecting multiple tuning parameters, one for each regression estimation, which may complicate the practical implementation.  
	\cite{breunig2016adaptive} considers adaptive estimation of \textit{known} functionals in nonparametric IV models, but does not study inference on them.
	\cite{chen2018optimal} and \cite{chen2025adaptive} provide inference procedures for \textit{known} functionals of nonparametric IV regressions. 
	Our context differs from theirs, as in our case the AME is an \textit{unknown} functional of the nonparametric IV regression, since it involves the expectation of \(h_0'\) with respect to the unknown distribution of~\(Z\).
	\cite{beyhum2023one} and \cite{zhang2023instrumental} propose one-step estimators of nonparametric IV models, but do not study inference on the AME. RKHS estimation of nonparametric IV models is also studied in \cite{singh2019kernel} and \cite{zhang2023instrumental}. 
	However, they neither consider a partially linear specification nor develop inference for the AME.
	
	\paragraph{Organization of the paper.}
	In Section~\ref{sec:estimator}, we introduce our AME estimator and derive its expression. 
	The asymptotic behavior of our estimator is established in Section~\ref{sec:asymptotic_analysis}.  
	Section~\ref{sec:bootstrap_test_for_theta0} presents the bootstrap test and obtains its validity. 
	The implementation of our inference procedure as well as the selection of our unique regularization parameter are discussed in Section~\ref{sec:implementation}. 
	In Section~\ref{sec:small_sample_behavior_simulations}, we provide simulation evidence for the small-sample behavior of our bootstrap test, demonstrating its excellent finite-sample performance.
	Section~\ref{sec:empirical_application} contains three empirical applications of our method. 
	Finally, Section~\ref{sec:conclusion} concludes.
	The Appendix collects the mathematical proofs of our results, auxiliary lemmas, and additional simulations. The  \texttt{R} package for implementing our inference procedure can be downloaded at 
	\url{https://github.com/lucasgirardh/rkhsiv}.
	
	\section{The estimator}
	\label{sec:estimator}
	\subsection{Estimation setup}
	We assume that \(h_0 \in \mathcal{H}\), with \(\mathcal{H}\) being a space of functions defined on the support of $Z$. Below we define precisely the space $\mathcal{H}$. For now, let us present the general features of the estimation method. From Equation \eqref{eq: npiv regression}, we have the moment condition $\operatorname{E}\{Y-h_0(Z) - X^T \beta_0  |W,X\}=0$.  We assume that identification holds on $\mathcal{H} \times \mathbb{R}^p$, that is, 
	\begin{equation}
		\label{eq:completeness_assumption}
		\E\{Y-h(Z)-X^T \beta|W,X\}=0\text{ with } h\in\mathcal{H} \text{ and } \beta\in\mathbb{R}^p
		\iff 
		h = h_0 
		\text{ and }
		\beta =\beta_0\, .
	\end{equation}
	Necessary and sufficient conditions for the identification of the partially linear IV model can be found in  \cite{florens2012instrumental}. For example, a sufficient condition for the identification of $(\beta_0, h_0)$ is that (i) $\E\{h(Z)|X,W\}=0$ implies $h=0$, (ii) $\E\{X X^T\}$ is full-rank, and (iii) if $\operatorname{E}\{h(Z)|X,W\}=X^T\beta$ for some $(h,\beta)\in\mathcal{H}\times \mathbb{R}^p$, then  $X^T\beta=0$. Part~(i) is the classical completeness condition in nonparametric IV models, see \cite{newey2003instrumental}, \cite{d2011completeness}, \cite{freyberger2017completeness}. Part~(ii) is the standard absence of perfect collinearity among control variables~\(X\). Part~(iii) expresses a similar idea: the only way for \(\E\{h(Z)|W, X\}\) to be a linear function of~\(X\) is to be the null function. See \cite{florens2012instrumental} for details.  For our purpose, we simply assume the identification condition in Equation \eqref{eq:completeness_assumption}.\\
	The logic underlying the construction of our estimator follows the ideas in \citet{beyhum2023one}. Building an estimator based on the conditional moment restriction in \eqref{eq:completeness_assumption} would require the preliminary estimation of the conditional expectation given $(X,W)$, as done by most of the literature. Differently, we consider an equivalent reformulation of \eqref{eq:completeness_assumption} that allows us to avoid such a preliminary estimation. From \citet[Theorem~2.2]{bierens2016econometric},  
	\begin{equation}
		\label{eq: Bierens reformulation of integral equation}
		\E\{Y-X^T\beta_0 - h_0(Z)|W,X\}=0
		\Leftrightarrow 
		\E\!\left\{
		\big[Y- X^T\beta_0 - h_0(Z)\big] \exp\!\big(\bm{i}(W, X^T) t\big)
		\right\} 
		= 0\; \forall\; t \in \mathbb{R}^{p+1} \,,  
	\end{equation}
	where $\bm{i}$ denotes the imaginary root.
	Let us define
	\begin{equation*}
		M(\beta,h)
		:= 
		\int 
		\big| 
		\E\!\left\{\left[Y- X^T\beta - h(Z)\right] 
		\exp\!\big(\bm{i}(W, X^T) t\big)
		\right\}
		\big|^2 
		d\mu(t)   \, ,
	\end{equation*}
	with $\mu$ a probability measure supported on $\mathbb{R}^{p+1}$ having a symmetric characteristic function. 
	Notice that $M(\beta,h)\geq 0$, and, by \eqref{eq: Bierens reformulation of integral equation} and the identification condition \eqref{eq:completeness_assumption}, $M(\beta,h)=0\Leftrightarrow (\beta,h)=(\beta_0,h_0)$. Thus, 
	\begin{equation*}
		(\beta_0,h_0)
		=
		\argmin_{\beta\in\mathbb{R}^{p}\,,\,h\in\mathcal{H}}
		M(\beta, h)\, .
	\end{equation*}
	We will now use the above equation to build an estimator of $(h_0,\beta_0)$.
	Given an independent and identically distributed (i.i.d.) sample \(\{Y_i, Z_i, X_i, W_i\}_{i=1, \ldots, n}\), we denote by \(\E_n\) the empirical-mean operator, that is, \(\E_n{f(Y_i, Z_i, X_i, W_i)}:= n^{-1} \sum_{i=1}^n f(Y_i, Z_i, X_i, W_i)\) for any function \(f\).
	To build the empirical counterpart of $M(\beta,h)$, we replace the population expectation \(\E\) with the empirical mean operator \(\E_n\):
	\begin{align}
		\label{eq:def_objective_function_Mn_beta_h}
		M_n(\beta,h):=&\int_{ }\big|\E_n\!\left\{\big[Y_i- X_i'\beta - h(Z_i)\big] 
		\exp\!\big(\bm{i}(W_i, X^T_i) t\big)
		\right\}
		\big|^2 
		d\mu(t)\nonumber \\
		=&\frac{1}{n^2}\sum_{1\leq i,j\leq n}[Y_i-h(Z_i)-X_i^T \beta]\;\mathcal{F}_\mu((W_i,X_i^T)-(W_j,X_j^T))\;[Y_j-h(Z_j)-X_j^T \beta]\,,
	\end{align}
	where $\mathcal{F}_\mu(\cdot)=\int_{ }\exp(\textbf{i}t^T \cdot)d\mu(t)$ denotes the symmetric characteristic function of $\mu$. \\
	Now, minimizing $M_n(\beta,h)$ with respect to $\beta$ and $h$ would lead to overfitting, as $h$ is a nonparametric function. Thus, we regularize the problem by minimizing a penalized version of $M_n$, that is,
	
	\begin{equation}
		\label{eq:empirical_program}
		(\betahat, \hhat)
		:=
		\argmin_{\beta\in\Reels^{p}\,,\,h\in\mathcal{H}}
		M_n(\beta, h) 
		+ 
		\lambda \|h\|^2_{\mathcal{H}} \,,
	\end{equation}
	where \(\lambda > 0\) is a penalty parameter, and \(\|\cdot\|_{\mathcal{H}}\) is the norm on the space $\mathcal{H}$.
	In the next section, we discuss the definition of $\mathcal{H}$ and show that the above minimization problem admits a unique solution. The practical choice of the penalty parameter~\(\lambda\) is discussed in Section~\ref{sec:implementation}. 
	Given $\widehat h$, we estimate the AME of the endogenous treatment as follows:
	\begin{equation}
		\label{eq: ame estimator}
		\widehat \theta:=\E_n\{\widehat h'(Z_i)\}\, .
	\end{equation}
	
	\subsection{Computation of the estimator}
	We choose the space $\mathcal{H}$ to be a Reproducing Kernel Hilbert Space (RKHS).
	This allows us (i)~to have enough flexibility to recover nonparametrically the treatment function $h_0$, and (ii)~to obtain an estimator with a simple closed form expression. 
	For a comprehensive treatment of RKHSs, we refer the reader to, e.g., \citet[Chapter 12]{wainwright2019high}, \cite{berlinet2011reproducing}, and \cite{steinwart2008support}. To clearly present our estimator, we briefly recall the notion of an RKHS. We denote by $\mathcal{Z}$ the support of $Z$, and let $K:\mathcal{Z}\times \mathcal{Z}\to \mathbb{R}$ be a symmetric function, that is, $K(z_1,z_2)=K(z_2,z_1)$ for any $z_1,z_2\in\mathcal{Z}$. $K$ is said to be positive semidefinite if, for any $m\in\mathbb{N}$ and any collection $z_1,\ldots,z_m\in\mathcal{Z}$, the $m\times m$ matrix with entries $\{K(z_i,z_j):i,j=1\ldots,m\}$ (the associated Gram matrix) is positive semidefinite. A symmetric and positive semidefinite function $K$ is called a kernel. The formal definition of an RKHS of functions is the following.  
	\begin{definition}\label{def: rkhs definition}
		Let $K:\mathcal{Z}\times \mathcal{Z}\to \mathbb{R} $ be a symmetric positive semidefinite function. Let $\mathcal{H}$ be a Hilbert space of functions defined on $\mathcal{Z}$ and valued in $\mathbb{R}$, endowed with the inner product $\left<\cdot,\cdot\right>_\mathcal{H}$. $\mathcal{H}$ is an RKHS with reproducing kernel $K$ if
		\begin{enumerate}
			\item $K(\cdot,z)\in\mathcal{H}$ for any $z\in\mathcal{Z}$;
			\item For any $h\in\mathcal{H}$ and $z\in\mathcal{Z}$,  $\left<h,K(\cdot,z)\right>_\mathcal{H}=h(z)$ . 
		\end{enumerate}
	\end{definition}
	From \citet[Theorems 4.20 and 4.21]{steinwart2008support}, every symmetric positive semidefinite kernel has a unique RKHS, and every RKHS has a unique reproducing kernel. For the sake of clarity, we provide two examples of RKHSs. For further examples, see \citet[Chapter 12]{wainwright2019high}, \cite{berlinet2011reproducing}, and \cite{steinwart2008support}. \\
	
	\noindent \textbf{Example 1 (Sobolev RKHS)}. Let $\mathcal{Z}=[0,1]$, and for any function $f:[0,1]\to \mathbb{R}$, denote with $f^{(\kappa)}$ its $\kappa$th derivative on $(0,1)$. Define   
	\begin{equation*}
		\mathcal{H}:=\Big\{f:[0,1]\to \mathbb{R}\text{ such that }f \text{ is $\kappa$ times differentiable and }\int_0^1|f^{(\kappa)}(z)|^2dz <\infty\Big\}\, . 
	\end{equation*}
	Then, $\mathcal{H}$ is an RKHS of functions with reproducing kernel
	\begin{equation*} 
		K(z,u)=\sum_{j=0}^{\kappa-1}\frac{z^{j}}{j!}\frac{u^{j}}{j!}+\int_{0}^1\frac{(z-t)_+^{\kappa-1}}{(\kappa-1)!} \frac{(u-t)_+^{\kappa-1}}{(\kappa-1)!} dt\, ,
	\end{equation*}
	where $(a)_+:=\max(a,0)$ for any $a\in\mathbb{R}$. See \citet[Chapter 12]{wainwright2019high}.
	\nopagebreak \hfill $\blacksquare$\\
	
	The RKHS of the following example is very popular in the statistics and machine learning literature.\\
	
	\noindent \textbf{Example 2 (Gaussian RKHS)}. Let $H$ be the space of functions $g:\mathbb{C}\to \mathbb{C}$ such that $g$ is analytic and $\int |g(u)|^2 \exp(|u-\overline{u}|^2/2)du<\infty$, where $\overline{u}$ denotes the complex conjugate of $u$ and $du$ stands for the complex Lebesgue measure on $\mathbb{C}$. Consider the space 
	\begin{equation*}
		\mathcal{H}=\Big\{f:\mathcal{Z}\to \mathbb{R}\text{ such that }\exists g\in H\text{ with } f(z)=\text{Real}(g(z))\;\forall\;z\in\mathcal{Z}\Big\}  \, ,
	\end{equation*}
	where $\text{Real}(g(z))$ is the real part of $g(z)$.
	Then, $\mathcal{H}$ is an RKHS with reproducing kernel $K(z,u)=\exp(-|z-u|^2/2)$. See \citet[Chapter 4]{steinwart2008support}.
	\nopagebreak \hfill $\blacksquare$\\
	
	Having set $\mathcal{H}$ to be an RKHS of functions, we now show that the empirical program~\eqref{eq:empirical_program} admits a unique solution, and give a closed form expression to our estimator $\widehat \theta$. To this end, let us introduce the following notation. Recall that the characteristic function of the measure $\mu$ is given by \(\mathcal{F}_\mu(\cdot):=\int_{ }\exp(\bm{i} t^T \cdot)d\mu(t)\). We define the following \(n \times n\) matrix via its generic \((i,j)-\)th entry, for \(i, j = 1, \ldots, n\), 
	\begin{equation*}
		\bm{F}:=
		[\mathcal{F}_\mu((X_i^T,W_i) - (X_j^T,W_j))\;:\;i, j=1,\ldots,n]\,.
	\end{equation*}
	We also denote 
	\begin{equation*}
		\bm{X}
		:=
		\begin{pmatrix} 
			X_1^T \\
			\vdots \\ 
			X_n^T
		\end{pmatrix},
		\quad 
		\bm{Y}
		:=
		\begin{pmatrix} 
			Y_1 \\
			\vdots \\ 
			Y_n
		\end{pmatrix},
		\quad 
		\bm{K}
		:=
		[K(Z_i, Z_j)\;:\;i, j=1,\ldots,n]\, ,
	\end{equation*}
	where $\bm{K}$ is the \(n \times n\) Gram matrix associated with the kernel~\(K\). We let \(\bm{I}_n\) be the identity matrix of order~\(n\), and we use the shortcut notations for the \(p \times p\) matrix 
	\begin{equation*}
		\bm{C}:=\bm{X}^T \bm{F} \bm{X}\, .
	\end{equation*}
	
	Let $\mathcal{N}(\bm{K}):=\{\bm{\gamma}\in\mathbb{R}^n:\bm{K}\bm{\gamma}=0\}$ denote the null space of $\bm{K}$ and let $\mathcal{N}(\bm{K})^\perp$ be its orthogonal complement. Formally, $\mathcal{N}(\bm{K})^\perp:=\{\bm{\gamma}_\perp\in\mathbb{R}^n:\bm{\gamma}_\perp^T \bm{\gamma}=0 \text{ for all } \bm{\gamma}\in \mathcal{N}(\bm{K})\}$. 
	\begin{proposition}
		\label{prop:computation_solution_empirical_program}
		Let $\mathcal{H}$ be an RKHS with reproducing kernel $K$. Assume that $\mathcal{F}_\mu$ is symmetric about zero, $\bm{X}$ has full-column rank, the \(Z_i\)'s are distinct, and the (\(X_i^T,W_i\))'s are distinct.
		Then, the solution to \eqref{eq:empirical_program} exists, is unique, and satisfies
		\begin{equation}\label{eq: expresson of hhat and betahat}
			\hhat(\cdot)
			=
			\sum_{i = 1}^n \alphahat_{i} K(\cdot, Z_i),
			\quad 
			\betahat
			=
			\bm{C}^{-1}\bm{X}^T\bm{F}(\bm{Y}-\bm{K}\bm{\widehat \alpha})\,,
		\end{equation}
		where \(\bm{\alphahat}:=(\alphahat_{ 1}, \ldots, \alphahat_{ n})^T  \in \Reels^n\) is the minimum-norm solution to 
		\begin{equation}\label{eq: alphahat equation}
			\big( 
			\bm{K}\bm{F} \bm{K} 
			-
			\bm{K} \bm{F} \bm{X} \bm{C}^{-1} \bm{X}^T \bm{F} \bm{K}
			+
			n^2 \lambda \bm{K}
			\big) 
			\bm{\alphahat}
			=
			\bm{K}\bm{F}
			\big( 
			\bm{Y} - \bm{X} \bm{C}^{-1} \bm{X}^T \bm{F}\bm{Y} 
			\big)\,,
		\end{equation}
		and is the unique element of $\mathcal{N}(\bm{K})^\perp$ satisfying this equation. 
	\end{proposition}
	The proof of Proposition \ref{prop:computation_solution_empirical_program} can be found in Appendix \ref{sec: appendix proofs of of the main results}. \cite{singh2019kernel} obtain a similar result for a fully nonparametric IV model without the linear component.
	Let us make two remarks about the above proposition. First,  when the matrix $\bm{K}$ is invertible, the solution to \eqref{eq: alphahat equation} is unique. Second, from the proof of Proposition \ref{prop:computation_solution_empirical_program}, any solution $\bm{\alpha}=(\alpha_1,\ldots,\alpha_n)^T$ to \eqref{eq: alphahat equation} gives rise to the same function $\sum_{i=1}^n\alpha_i K(\cdot,Z_i)$. Thus, the specific solution to \eqref{eq: alphahat equation} used to compute $\widehat h$ in  \eqref{eq: expresson of hhat and betahat} has no effect on $\widehat h$. 
	
	As our parameter of interest involves the derivative of \(h_0\), we assume that \(K\) is a differentiable kernel and denote by $K'(\cdot,Z_i)$ the first derivative of $K(\cdot,Z_i)$. From Equation~\eqref{eq: expresson of hhat and betahat}, we estimate the first derivative of $h_0$ by 
	\begin{equation}\label{eq: hhat prime}
		\widehat h'(\cdot)=\sum_{i=1}^n \widehat \alpha_i K'(\cdot,Z_i)\, . 
	\end{equation}
	We then estimate the targeted quantity \(\theta_0\) using Equation~\eqref{eq: ame estimator}, with $\widehat h'$ computed as  above. By construction, our estimator of $\theta_0$ is readily obtained by solving the linear system of equations in~\eqref{eq: alphahat equation}, and is straightforward to compute. Moreover, it does not rely on multi-step nonparametric estimators, and involves only a single regularization parameter.   
	
	\section{Asymptotic Analysis}
	\label{sec:asymptotic_analysis}
	
	To study the asymptotic properties of our estimator, we first reformulate the integral equation \eqref{eq: Bierens reformulation of integral equation} that identifies $(\beta_0,h_0)$. Let $L^2_\mu$ be the space of functions square-integrable with respect to $\mu$, with inner product  $\left<g_1,g_2\right>_\mu=\int_{ }g_1(t)\overline{g_2(t)}d \mu(t)$ for $g_1,g_2\in L^2_\mu$, where $\overline{g_2(t)}$ denotes the complex conjugate of $g_2(t)$. Denote by $\|\cdot\|_\mu$ the norm on $L^2_\mu$ induced by $\left<\cdot,\cdot\right>_\mu$, that is, $\|g\|^2_\mu=\left<g,g\right>_\mu$.
	We define
	\begin{align}\label{eq: definitions of A, B, and s}
		\operatorname{A} h:=\operatorname{E}\{h(Z) \exp(\textbf{i}\,(X^T,W)\,\cdot)\}\,,\, \operatorname{A}:\mathcal{H}\to L^2_\mu\nonumber \\
		\operatorname{B}\beta:=\operatorname{E}\{\beta^T X \exp(\textbf{i}\,(X^T,W)\,\cdot)\}\,,\, \operatorname{B}:\mathbb{R}^p\to L^2_\mu\nonumber\\
		s:=\operatorname{E}\{Y \exp(\textbf{i}\,(X^T,W)\, \cdot)\}\,,\, s\in L^2_\mu\, .
	\end{align}
	
	Below we provide the integrability conditions ensuring that $s\in L^2_\mu$, and that both $\A$ and $\B$ are valued in $L^2_\mu$.
	Equation \eqref{eq: Bierens reformulation of integral equation} can be reformulated as follows:
	\begin{equation*}
		\operatorname{B}\beta_0+\operatorname{A} h_0=s\, .
	\end{equation*}
	Thus, $(\beta_0,h_0)$ can be identified as the unique solution to the following  population program
	\begin{equation*}
		(\beta_0,h_0)=\arg \min_{\beta\in\mathbb{R}^p,\,h\in\mathcal{H}}\|\operatorname{B}\beta+\operatorname{A} h-s\|^2_\mu\,. 
	\end{equation*}
	An advantage of this reformulation for our proof is that $\theta_0$ can be expressed in terms of the operators $(\operatorname{A}, \operatorname{B})$. In our proofs, we also consider the empirical (penalized) version of the above minimization program, which allows us to express our estimator $\widehat \theta$ in terms of the empirical counterparts of $\A$ and $\B$. 
	This facilitates the analysis of the asymptotic behavior of $\widehat \theta$.
	
	\begin{assumption}\label{as: iid and T}
		(i) $\{Y_i,Z_i,X_i,W_i\}_{i=1}^n$ is an i.i.d. sample defined on the probability space $(\Omega,\mathcal{A},P)$, with
		$\operatorname{E}\|X\|^2<\infty$, $\E W^2<\infty$, and $\operatorname{E}Y^2<\infty$; (ii) $Z$ has a continuously differentiable density $f_Z$ supported on $\mathcal{Z}$; (iii) $\mathcal{Z}$ is a bounded set with nonempty interior; (iv) $\mu$ is a probability measure supported on $\mathbb{R}^{p+1}$.
	\end{assumption}
	Part (i) of the above assumption imposes square integrability conditions ensuring that $s\in L^2_\mu$  and that $\B$ is valued in $L^2_\mu$. Parts (ii) and (iii) impose smoothness and boundedness assumptions that are common in the semiparametric literature.
	
	\begin{assumption}\label{as: rkhs}
		(i)  $\mathcal{H}$ is a Reproducing Kernel Hilbert Space of real-valued differentiable functions defined on $\mathcal{Z}$, with inner product $\left<\cdot,\cdot\right>_{\mathcal{H}}$ and reproducing differentiable kernel $K:\mathcal{Z}\times \mathcal{Z}\to \mathbb{R}$;
		(ii) $h_0, f'_Z\in\mathcal{H}$.
	\end{assumption}
	Part (i) of the above assumption formally defines $\mathcal{H}$ and  implies  that  the operator $\A$ is valued in $L^2_\mu$. 
	Part (ii) states that $h_0$ must belong to the RKHS $\mathcal{H}$. The inclusion of $f_Z'$ in $\mathcal{H}$ is needed to handle a bias term that appears in the expansion of $\widehat \theta$. 
	
	The following assumption ensures the identification of $(h_0,\beta_0)$. 
	We let \(\mathcal{R}(\A):=\{ \A h\) such that \(h\in\mathcal{H}\}\) and $\mathcal{R}(\B):=\{ \B \beta\text{ such that }\beta\in\mathbb{R}^p\}$ be the ranges of $\A$ and $\B$, respectively. 
	
	\begin{assumption}\label{as: completeness}
		(i) $\A$ and $\B$ are injective; (ii) $\mathcal{R}(\A)\cap \mathcal{R}(\B)=\{0\}  $.
	\end{assumption}
	\cite{florens2012instrumental} show that Assumption \ref{as: completeness} is a necessary and sufficient condition for identification of the partially linear IV model. Examples of Assumption \ref{as: completeness} can be found in    \cite{florens2012instrumental}. In the lines below Proposition \ref{prop: ifr}, we discuss how our results can be extended to the case where the operator $\A$ is not injective. 
	
	For the following assumption, given any class of functions $\mathcal{G}$, we define its bracketing entropy $N_{[\,]}(\epsilon,\mathcal{G},L^2(P))$ as the minimal number of brackets of $L^2(P)$ size $\epsilon$ required to cover $\mathcal{G}$. For a formal definition, see \citet[Chapter 19]{van2000asymptotic}.
	
	\begin{assumption}\label{as: belonging condition}
		(i) $\sup_{z\in\mathcal{Z}}|\widehat h'(z) - h'_0(z)|=o_P(1)$ ; (ii) $\Pr(\widehat h'\in\mathcal{G})\rightarrow 1$, with \(\log N_{[\;]}(\epsilon,\mathcal{G},L^2(P))\) \(\leq C \epsilon^{-v}\) and fixed constants $C$ and $v\in(0,2)$;
		(iii) $\lim_{z\rightarrow \pm \infty}f_Z(z)[\widehat h(z) - h_0(z)]=0$.
	\end{assumption}
	
	Parts (i) and (ii) of the above assumption are needed to obtain the asymptotic stochastic equicontinuity of an empirical process appearing in the expansion of $\widehat \theta$. 
	Part~(i) requires the uniform consistency of $\widehat h'$, without imposing a convergence rate.  
	Part~(ii) requires that $\widehat h'$ must be included asymptotically in a class of functions with a bounded entropy. This condition is satisfied if $\widehat h$ is twice differentiable and its first two derivatives are bounded in probability. 
	Part~(iii) is a boundary condition commonly used in semiparametric frameworks, see, e.g., \cite{powell1989semiparametric}, \cite{escanciano2010testing}. 
	It is satisfied when, as we approach the boundaries of $\mathcal{Z}$, $f_Z$ decays to zero and $\widehat h -h_0$ does not explode. 
	In Appendix~\ref{sec: low level conditions}, we provide mild primitive conditions guaranteeing Assumption~\ref{as: belonging condition}. 
	Such primitive conditions impose some smoothness on the kernel $K$ of the RKHS. 
	
	To introduce the next condition, we need some additional pieces of notation.
	Let $\P:L^2_\mu\to \mathcal{R}(\B)$ be the projection operator onto $\mathcal{R}(\B)$. As the linear operator $\B:\mathbb{R}^p\to L^2_\mu$ is defined on a finite-dimensional space, its range  $\mathcal{R}(\B)$ is a linear finite-dimensional space (see \citet[Theorem 2.6-9]{kreyszig1991introductory}). From \citet[Theorem 2.4-2]{kreyszig1991introductory}, linear finite-dimensional spaces are closed. Thus $\mathcal{R}(\B)$ is linear and closed. 
	Since from \citet[Theorem 1.26]{kress1999linear} projection operators onto linear and closed spaces are well defined, $\P$ is a well defined operator. Let $\T:=(\I-\P)\A:\mathcal{H}\to L^2_\mu$, where the operator $\I$ is the identity operator, and let us denote with $\operatorname{T}^*$ the Hilbert adjoint of $\operatorname{T}$. By definition, the Hilbert adjoint of $\operatorname{T}$ is the operator $\operatorname{T}^*:L^2_\mu\to \mathcal{H}$ such that $\left<\operatorname{T}\varphi,\psi\right>_{\mu}=\left<\varphi,\operatorname{T}^*\psi\right>_{\mathcal{H}}$ for any $(\varphi,\psi)\in\mathcal{H}\times L^2_\mu$.  See \citet[Chapter 4]{kress1999linear}. The next assumption is a \textit{source condition} common in the literature on inverse problems. 
	\begin{assumption}\label{as: source conditions}
		$h_0, \int_{\mathcal{Z}}K(\cdot,z) f_Z'(z) dz \in \mathcal{R}((\T^*\T)^{\gamma/2})$ for some $\gamma\geq 2$. 
	\end{assumption}
	The above source condition is common in the nonparametric IV literature and can be interpreted as a smoothness assumption, see, e.g., \cite{carrasco2007linear}, \cite{beyhum2023one}, \cite{babii2022high}, \cite{darolles2011nonparametric}, \cite{florens2012instrumental}. It essentially implies a requirement that is necessary for the asymptotic normality of the estimator, see \cite{severini2012efficiency}. In our proofs establishing the asymptotic behavior of $\widehat \theta$, we employ the above condition to deal with a bias term appearing in the expansion of $\widehat \theta$. 
	
	\begin{proposition}\label{prop: ifr}
		Let Assumptions \ref{as: iid and T}, \ref{as: rkhs}, \ref{as: completeness}, \ref{as: belonging condition}, and \ref{as: source conditions} hold. Let $n\lambda \rightarrow \infty$ and $n \lambda^2=o(1)$.
		Then,
		\begin{enumerate}[label=(\roman*)]
			\item     \begin{align*}
				\sqrt{n}(\widehat \theta-\theta_0)=\sqrt{n}(\operatorname{E}_n-\operatorname{E})\psi(Y_i,Z_i,X_i,W_i)+o_P(1)\,,
			\end{align*}
			where 
			\begin{align*}
				\psi(Y_i,Z_i,X_i,W_i)&=h_0'(Z_i)\\
				&-[Y_i-h_0(Z_i)]\left<(\I-\P)\exp(\textbf{i}(X_i^T,W_i)\cdot),\T f_0\right>_\mu \\
				&+ [X_i^T (\B^*\B)^{-1}\B^*(s-\A h_0) ]    \left<\exp(\textbf{i}(X_i^T,W_i)\cdot), \T f_0\right>_\mu\,,     \end{align*}
			$f_0=(\T^*\T)^{-1}g$, and $g=\int_{ }K(\cdot,z) f_Z'(z) dz$.
			
			\item $\sqrt{n}(\widehat \theta - \theta_0)\leadsto \mathcal{N}(\,0\,,\,\operatorname{Var}[\psi(Y_i,Z_i,X_i,W_i)]\,)$\, .
		\end{enumerate}
	\end{proposition}
	The proof of Proposition \ref{prop: ifr} can be found in Appendix \ref{sec: appendix proofs of of the main results}. The condition $n \lambda \rightarrow \infty$ guarantees that several variance terms appearing in the expansion of $\sqrt{n}( \widehat  \theta - \theta_0)$ are asymptotically negligible, while the condition $n \lambda^2=o(1)$
	ensures that a (regularization) bias term vanishes. The result can be extended to accommodate a data-dependent choice of $\lambda$. However, for simplicity, we consider only a deterministic regularization parameter.\\
	The asymptotic normality result of Proposition \ref{prop: ifr} is obtained under Assumption \ref{as: completeness}, which guarantees the identification of $h_0$. Such an assumption requires the injectivity of $\A$, that is, the completeness of the distribution of $Z$ conditional on $(X,W)$. We point out that the convergence in distribution of our statistic can also be obtained without requiring the injectivity of $\A$. Under Assumption \ref{as: source conditions}, the conditions in \citet[Lemma 3.1 and 4.1]{severini2012efficiency} are satisfied, and the AME $\theta_0$ is identified also if identification of $h_0$ does not hold. We can then extend our results and obtain the asymptotic normality of  $\widehat \theta$ by replacing $h_0$ with the projection of $h_0$ onto $\mathcal{N}(\T)^\perp$, the orthogonal complement of the null space of $\T$. For simplicity of presentation, we choose not to provide such a more general version of our results, and to work under the completeness condition in Assumption \ref{as: completeness}. \\ 
	We finally note that our estimator is based on a moment condition that is not locally robust with respect to the preliminary estimation of $h_0$; see \cite{chernozhukov2022locally} and \cite{bennett2023source}. A locally robust approach would require estimating additional correction terms and hence selecting  multiple regularization parameters. Differently, our approach relies on a single regularization parameter, which simplifies its implementation.
	
	Although from Proposition \ref{prop: ifr}   \(\widehat{\theta}\) is asymptotically normal,  the variance of the limiting distribution has a rather involved expression. This makes it inconvenient to use the quantiles of the asymptotic distribution for hypothesis testing on \(\theta_0\). Therefore, in the next section, we construct a simple bootstrap test.
	
	\section{Bootstrap test for \texorpdfstring{\(\theta_0\)}{theta0}}
	\label{sec:bootstrap_test_for_theta0}
	
	We construct a Bayesian bootstrap test for the null hypothesis 
	\begin{equation*}
		\mathcal{H}_0:\theta_0=\theta_{\mathcal{H}_0}\quad\text{versus}\quad\mathcal{H}_1:\theta_0 \neq \theta_{\mathcal{H}_0}\,,
	\end{equation*}
	with \(\theta_{\mathcal{H}_0} \in \Reels\) a user-chosen tested value of the AME.
	For the sake of brevity, we focus on a two-sided test, but our procedure can be easily adapted to implement a one-sided test as well. In view of Proposition~\ref{prop: ifr}, we construct a bootstrap test for $\theta_0$ by bootstrapping the Wald statistic $\sqrt{n}(\widehat \theta - \theta_0)$. To this end, we first define the Bayesian bootstrap versions of $\widehat h$ and $\widehat \beta$, and then we construct the Bayesian bootstrap version of $\widehat \theta$. 
	
	Let $\{\xi_i\}_{i=1, \ldots, n}$ be i.i.d. random variables with $\E\{\xi\} = \operatorname{Var}\{\xi\} = 1$, and independent from the sample data.
	For example, we can sample each $\xi_i$ from an exponential distribution with density $f(\xi)=\exp(-\xi)$. Let $\overline{\xi}:= n^{-1} \sum_{i=1}^n \xi_i$. Note that we use $\overline{\cdot}$ for both complex conjugation and empirical means to avoid introducing additional notation. The distinction will be clear from the context. 
	We define the bootstrap version of $(\betahat, \hhat)$ as 
	\begin{equation}
		\label{eq: bootstrap version of hhat}
		(\betahat_b, \hhat_b)
		:=
		\argmin_{\beta\in\Reels^{p}\,,\,h\in\mathcal{H}}
		\int 
		\left|
		\E_n\!\left\{
		\big[Y_i- X_i^T \beta - h(Z_i)\big]
		\exp\!\big(\bm{i}(W_i, X_i^T) t\big)
		\frac{\xi_i}{\overline{\xi}}
		\right\}
		\right|^2 d\mu(t) 
		+ \lambda \|h\|^2_{\mathcal{H}}\,.
	\end{equation} 
	The above objective function is similar to the one used for $(\betahat, \hhat)$, with the exception that the $i$th observation has weight $\xi_i/\overline{\xi}$, so each observation is represented in the bootstrap sample with that weight.
	From the minimization program in \eqref{eq: bootstrap version of hhat}, the bootstrapped estimators $\widehat h_b$ and $\widehat \beta_b$ can be computed in the same way as their sample counterparts in Equations \eqref{eq: expresson of hhat and betahat} and \eqref{eq: alphahat equation}, provided that the weighting matrix  $\bm{F}=\{\mathcal{F}_\mu((X_i^T,W_i)-(X_j^T,W_j)) : i,j=1,\ldots,n\}$ is replaced with 
	\begin{equation*}
		\bm{F}_b:=[\mathcal{F}_\mu((X_i^T,W_i)-(X_j^T,W_j)) \xi_i \xi_j /(\overline{\xi})^2 : i,j=1,\ldots,n] \, .  
	\end{equation*}
	Similarly, we define \(\bm{C}_b:= \bm{X}^T \bm{F}_b \bm{X}\).
	Thus, with $\bm{\widehat \alpha}_b=(\widehat \alpha_{b,1},\ldots,\widehat \alpha_{b,n})^T$ minimum norm solution to 
	\begin{equation}\label{eq: alphahat bootstrap equation}
		\big( 
		\bm{K}\bm{F}_b \bm{K} 
		-
		\bm{K} \bm{F}_b \bm{X} \bm{C}_b^{-1} \bm{X}^T \bm{F}_b \bm{K}
		+
		n^2 \lambda \bm{K}
		\big) 
		\bm{\alphahat}_b
		=
		\bm{K}\bm{F}_b
		\big( 
		\bm{Y} - \bm{X} \bm{C}_b^{-1} \bm{X}^T \bm{F}_b\bm{Y} 
		\big)
	\end{equation}
	we compute 
	\begin{equation}\label{eq: hhat prime bootstrap}
		\widehat h_b'(\cdot)=\sum_{i=1}^n \widehat \alpha_{b,i}K'(\cdot,Z_i)
	\end{equation}
	and define the bootstrap version of $\widehat \theta$ as 
	\begin{equation}\label{eq: thetahat bootstrapped}
		\thetahat_{b} :=
		\E_n\!\left\{\frac{\xi_i}{\overline{\xi}} \, \hhat_b'(Z_i) \right\}
		\,. 
	\end{equation}
	Hence, the bootstrap version of the statistic $\sqrt{n}(\widehat \theta - \theta_0)$ is 
	\begin{equation}
		\label{eq: bootstrapped statistic}
		\sqrt{n}(\widehat \theta_{b} - \widehat \theta)\, .
	\end{equation}
	We implement the test by using symmetric bootstrap p-values, but our procedure can be easily modified to implement equal-tail bootstrap p-values.
	To run a test of nominal size $\alpha$, we set the critical value to the $(1-\alpha)$ quantile of the bootstrap distribution of the bootstrapped statistic $|\sqrt{n}(\widehat \theta_b - \widehat \theta)|$. Formally,
	\begin{align*}
		\widehat c _{1 - \alpha}=\inf_c\left\{{\Pr}_{\xi}\left(|\sqrt{n}(\widehat \theta_b - \widehat \theta)|\leq c\right)\geq 1-\alpha\right\}\, ,
	\end{align*}
	where $\Pr_{\xi}$ represents the probability that considers as random only the bootstrap weights $\{\xi_i:i=1,\ldots,n\}$ and as fixed the sample data.  We then reject the null hypothesis if and only if $|\sqrt{n}(\widehat \theta - \theta_{\mathcal{H}_0})|>\widehat c_{1-\alpha}$. 
	
	\medskip
	
	To show the validity of the bootstrap test, we introduce some regularity conditions.  To avoid introducing additional notation, we make a slight abuse of notation and denote by $\Pr$ the probability treating both the bootstrap weights $\{\xi_i:i=1,\ldots,n\}$ and the sample data $\{(Y_i,Z_i,X_i,W_i):i=1,\ldots,n\}$ as random. 
	
	\begin{assumption}\label{as: bootstrap weights}
		The bootstrap weights $\{\xi_i:i=1,\ldots,n\}$ are i.i.d., independent from the sample data $\{(Y_i,Z_i, X_i,W_i):i=1,\ldots,n\}$, and satisfy $\operatorname{E}\{\xi\}=\operatorname{Var}\{\xi\}=1$.
	\end{assumption}
	The following assumption is the bootstrap counterpart of Assumption 
	\ref{as: belonging condition}. Mild primitive conditions guaranteeing this assumption are provided in Appendix~\ref{sec: low level conditions}. 
	
	\begin{assumption}\label{as: belonging condition for bootstrap}
		(i) $\sup_{z\in\mathcal{Z}}|\widehat h'_b(z) - h'_0(z)|=o_P(1)$ ; (ii) $\Pr(\widehat h'_b\in\mathcal{G})\rightarrow 1$, with \(\log N_{[\;]}(\epsilon,\mathcal{G},L^2(P)) \) \(\leq C \epsilon^{-v}\) and fixed constants $C$ and $v\in(0,2)$; (iii) $\lim_{z\rightarrow \pm \infty}f_Z(z)[\widehat h_b(z) - h_0(z)]=0$.
	\end{assumption}
	
	\begin{proposition}\label{prop: bootstrap validity} 
		Let the assumptions of Proposition \ref{prop: ifr}, Assumption 
		\ref{as: bootstrap weights}, 
		and Assumption \ref{as: belonging condition for bootstrap} hold. Then,
		\begin{enumerate}[label=(\roman*)]
			\item    \begin{align*}
				\sqrt{n}(\widehat \theta_b - \widehat \theta)=\sqrt{n}\operatorname{E}_n\!\left\{\left(\frac{\xi_i}{\overline{\xi}}-1\right)\psi(Y_i,Z_i,X_i,W_i)\right\}+o_P(1)\,,
			\end{align*}
			
			where $\psi$ is defined in Proposition \ref{prop: ifr}, and the probability space is the joint probability on the random bootstrap weights and the sample data.
			\item Under $\mathcal{H}_0: \theta_0=\theta_{\mathcal{H}_0}$, $\Pr( |\sqrt{n}(\widehat \theta - \theta_{\mathcal{H}_0})|>\widehat c_{1-\alpha} ) \rightarrow \alpha$.
			\item Under $\mathcal{H}_1: \theta_0\neq\theta_{\mathcal{H}_0}$ , $\Pr( |\sqrt{n}(\widehat \theta - \theta_{\mathcal{H}_0})|>\widehat c_{1-\alpha} ) \rightarrow 1$.
		\end{enumerate}
	\end{proposition}
	The above result establishes the validity of the bootstrap test, and its proof  can be found in Appendix \ref{sec: appendix proofs of of the main results}. Part (i) obtains the bootstrap influence function representation of $\widehat \theta_b$, showing that such a representation is a reweighed version of the expansion of $\widehat \theta$ from Proposition \ref{prop: ifr}. The new weights are the bootstrap weights $\{\xi_i/\overline{\xi}-1\;:\;i=1,\ldots,n\}$. Such a reweighting ensures that the leading term of the  bootstrap expansion estimates consistently the distribution of the statistic under the null, while remaining bounded in probability under the alternative. Parts (ii) and (iii) of Proposition \ref{prop: bootstrap validity} establish that the test has the correct size asymptotically and is consistent, with power converging to one when the sample size goes to infinity.
	
	\section{Implementation}
	\label{sec:implementation}
	
	This section describes the implementation of our method. We first discuss the choice of our single regularization parameter. We select $\lambda$ by minimizing a cross-validated version of the objective function in \eqref{eq:def_objective_function_Mn_beta_h}. Specifically, let us focus on a  two-fold cross-validation.
	We partition the sample into two folds, denoted $S_1$ and $S_2$. We first use  $S_1$ for estimation and $S_2$ for validation, and then interchange their roles. Let $\widehat \beta_{S_1,\lambda}$ and $\widehat h_{S_1,\lambda}$ denote the estimators obtained from the first fold for a given value of $\lambda$. Define  $\widehat \beta_{S_2,\lambda}$ and $\widehat h_{S_2,\lambda}$ analogously for the second fold. The cross-validated version of \eqref{eq:def_objective_function_Mn_beta_h} is given by
	\begin{align}\label{eq: 2fold CvM criterion}
		\int 
		\Bigg|
		\frac{1}{n}
		\bigg(
		\sum_{i \in S_1}
		\Big[\big(Y_i - X_i^T \betahat_{S_2, \lambda} & - \hhat_{S_2, \lambda}(Z_i)\big)  \exp\!\big(\bm{i}(W_i, X_i^T) t\big)
		\Big] \nonumber  \\
		& +
		\sum_{i \in S_2}
		\Big[\big(Y_i - X_i^T \betahat_{S_1, \lambda} - \hhat_{S_1, \lambda}(Z_i)\big)  \exp\!\big(\bm{i}(W_i, X_i^T) t\big)
		\Big]
		\bigg)    
		\Bigg|^2 
		d\mu(t),
	\end{align}
	where $\sum_{i\in S_{j}}$ denotes the sum over observations in fold $S_j$, for $j=1,2$. Finally, we minimize the above criterion with respect to $\lambda$, and we take the minimizer as the selected regularization parameter.\\ 
	Compared to the usual out-of-sample mean squared error criterion, the above cross-validation objective employs the appropriate weights, those corresponding to our objective function \(M_n(\beta, h)\) in Equation~\eqref{eq:def_objective_function_Mn_beta_h}. In our simulations and empirical applications, we implement the above two-fold cross-validation method with good results. \\
	To reduce the computational burden of the bootstrap procedure, we fix the regularization parameter at the value selected from the original sample and use this same value across all bootstrap replications. As shown in our simulation study, this is enough to provide a good performance of our procedure. Algorithm~\ref{algo:our_bootstrap_test} summarizes the practical steps required to implement our test.
	
	\begin{algorithm}[ht]
		\caption{Bootstrap test for the null hypothesis \(\mathcal{H}_0:\theta_0 = \theta_{\mathcal{H}_0}\)}
		\label{algo:our_bootstrap_test}
		\textbf{Inputs}: a sample \(\{Y_i, Z_i, X_i, W_i\}_{i=1, \ldots, n}\), a tested value \(\theta_{\mathcal{H}_0}\), a number \(B\) of bootstrap repetitions, a nominal level \(\alpha \in (0, 1)\), a distribution for \(\xi\).
		
		\begin{enumerate}
			
			\item 
			Obtain a cross-validated \(\lambdahat\) by minimizing the two-fold Cramér-von Mises criterion in \eqref{eq: 2fold CvM criterion}.
			
			\item 
			Compute \(\thetahat\) using the penalization \(\lambdahat\) as in Equations \eqref{eq: ame estimator}, \eqref{eq: alphahat equation}, and \eqref{eq: hhat prime}.
			
			\item \textbf{For} \(b = 1, \ldots, B\) 
			\begin{enumerate}
				\item Draw an i.i.d. sample $(\xi_1, \ldots, \xi_n)$ from the distribution of $\xi$.
				\item Compute $\thetahat_b$ using the penalization \(\lambdahat\) and the weights $(\xi_1, \ldots, \xi_n)$ as in Equations~\eqref{eq: alphahat bootstrap equation}, \eqref{eq: hhat prime bootstrap}, and~\eqref{eq: thetahat bootstrapped}.
			\end{enumerate}
			
			\item Compute the $1-\alpha$ quantile of $\big\{|\widehat \theta_{b} - \widehat \theta|:b=1,\ldots, B \big\}$, denoted $\widehat q_{1-\alpha}$.
			
			\item Reject $\mathcal{H}_0$ if and only if $|\widehat \theta - \theta_{\mathcal{H}_0}|>\widehat q_{1-\alpha}$.
		\end{enumerate}
	\end{algorithm}
	
	\section{Small Sample Behavior}
	\label{sec:small_sample_behavior_simulations}
	
	In this section, we study the finite-sample performance of our test in two settings.
	We first consider a fully nonparametric model corresponding to Equation \eqref{eq: npiv regression} without~\(X\). We then consider the general case of a partially linear model, where we also include the regressor $X$.
	
	\paragraph{Fully nonparametric model.}
	We use a data-generating process consistent with model~\eqref{eq: npiv regression} without~\(X\), where we draw the variables \(Z\), \(W\), and \(\eps\) as follows:
	\begin{equation}
		\label{eq: simulations setting for epx and Z}
		\eps = \frac{aV + U}{\sqrt{1 + a^2}}, 
		\; \text{where} \; a = \sqrt{\frac{\rho_{\eps V}^2}{1 - \rho_{\eps V}^2}} \,,
		\quad \text{and} \quad 
		Z = \frac{b W + V}{\sqrt{1 + b^2}},
		\; \text{where} \;
		b = \sqrt{\frac{\rho_{ZW}^2}{1 - \rho_{ZW}^2}}\,,
	\end{equation}
	with \((W, V, U)\) mutually independent standard Gaussians. 
	In this setup, the marginal distributions of $\varepsilon$ and $Z$ are also standard Gaussians. 
	For a non-zero value of $a$, the regressor $Z$ is endogenous and correlated with the error $\varepsilon$. The parameter $\rho_{\varepsilon V}$ measures the degree of such correlation, and hence the level of endogeneity of $Z$.  For a non-null value of $b$, $W$ is a valid instrument for $Z$.
	
	We implement our estimator with $K(z,u)=\exp(-|z-u|^2/2)$, which corresponds to a Gaussian kernel. 
	This is a well-established choice for kernel ridge regression in the machine learning literature (see \citet[Chapter 4]{steinwart2008support}) and enables the approximation of a wide range of functions.
	Specifically, it has a universal approximation property, in the sense that any continuous function can be arbitrarily well approximated by a function in the RKHS with Gaussian reproducing kernel (see \citet[Corollary 4.58]{steinwart2008support}).
	We set $\mu$ to a Laplace distribution with mean zero and unit variance, thus its characteristic function $\mathcal{F}_\mu$ is a Cauchy density. We choose the penalty parameter $\lambda$ by the two-fold cross-validation described in Section \ref{sec:implementation}, and perform a grid search to minimize the cross-validated objective function.
	For bootstrapping, we draw the bootstrap weights $\xi$ from an Exponential distribution with expectation one.
	
	In this nonparametric setting, we compare our procedure to an inference method for the average derivative based on two-stage series regressions \citep{ai2007estimation}. 
	The method is implemented as follows. We set the series basis to B-Splines and select the number of series terms in the first and second stages using the method proposed in \citet*{chen2025adaptive}. Given the two-stage series estimator of the IV regression, we compute the AME similarly to Equation \eqref{eq: ame estimator}. We then obtain the p-values by pairwise bootstrap. Specifically, we resample with replacement from the original sample $\{Y_i,Z_i,W_i\}_{i=1}^n$ and, for each bootstrap sample, compute the AME, keeping the same number of series terms selected in the original sample. 
	
	We consider two functional forms for the treatment effects: $h_{0,1}(z)=z^2/\sqrt{2}$, corresponding to a quadratic function, and $h_{0,2}(z)=\sqrt{3 \sqrt{3}}\exp(-z^2/2)$, corresponding to a non-polynomial function. Both expressions are normalized so that $h_{0,1}(Z)$ and $h_{0,2}(Z)$ have unit variance. With these functions and the marginal distribution of $Z$, the corresponding true values of the AMEs (\(\theta_0\)) are, respectively, 0 in the quadratic case, and approximately 0.81 in the non-polynomial case. Finally, we set $\rho_{ZW} = 0.8$ and consider two levels of endogeneity: 
	$\rho_{\varepsilon V} = 0.5$, corresponding to a moderate level of endogeneity, 
	and $\rho_{\varepsilon V} = 0.8$, corresponding to a high level of endogeneity.
	
	We ran 5,000 simulations with the sample sizes of $n=100, 400$. To speed up computations, we use the warp-speed method of \cite{davidson2007improving} and \cite{giacomini2013warp}.
	Thus, for each simulated sample, we draw one bootstrap sample and use the collection of bootstrapped statistics to compute the bootstrap p-value for each original statistic.
	Table~\ref{tab:rejection_rate_under_HO_quadratic_and_nonpolynomial_case} reports the rejection probabilities under the null hypothesis at the usual nominal levels of 5\% and 10\% for the different functions \(h_0\) and levels of endogeneity \(\rho_{\eps V}\).\footnote{
		We also consider the nominal level of~1\%, but the rejection rates for the two-stage series regression method are very close to~0. 
	}
	We label our method ``one-step'' and the two-stage series regression method ``two-step''.
	Under the null hypothesis, both methods control Type~I error, with rejection probabilities below the nominal levels.
	However, our method delivers rejection rates closer to the targeted nominal levels in most configurations -- it is systematically the case with \(n = 100\) and with the non-polynomial function \(h_{0,2}\).
	The improvement is notably large for the smaller sample size \(n = 100\). 
	
	\begin{table}[H]
		\centering
		\caption{Rejection probabilities under the null hypothesis in the fully nonparametric model.}
		\label{tab:rejection_rate_under_HO_quadratic_and_nonpolynomial_case}
		\begin{tabular}{c c | c c | c c | c c | c c}
			\hline 
			\multicolumn{2}{c|}{Case} &
			\multicolumn{4}{c|}{quadratic} & \multicolumn{4}{c}{non-polynomial} \\
			\hline 
			\multicolumn{2}{c|}{Nominal level} &
			\multicolumn{2}{c|}{5\%} & \multicolumn{2}{c|}{10\%} &
			\multicolumn{2}{c|}{5\%} & \multicolumn{2}{c}{10\%}
			\\
			\hline 
			\multicolumn{2}{c|}{Method} & 1-step & 2-step & 1-step & 2-step & 1-step & 2-step & 1-step & 2-step \\
			\hline         
			\(\rho_{\eps V} = 0.5\) &\(n = 100\) & 
			2.20 & 0.73 & 
			5.14 & 3.13 & 
			2.42 & 0.65 & 
			6.70 & 3.11 
			\\
			& \(n = 400\) &
			2.20 & 2.69 & 
			5.94 & 7.42 & 
			3.38 & 1.57 & 
			7.22 & 5.30 
			\\
			\hline 
			\(\rho_{\eps V} = 0.8\) & \(n = 100\) & 
			2.48 & 0.77 & 
			5.56 & 3.36 & 
			3.00 & 0.80 & 
			6.92 & 2.87 
			\\
			& \(n = 400\) & 
			1.98 & 2.72 & 
			6.16 & 8.01 & 
			3.26 & 1.51 & 
			7.42 & 5.43 
			\\
			\hline
		\end{tabular}
		\tablefigurenote{Rejection probabilities for our one-step regression test (labeled ``1-step'') and for the two-step series regression test (labeled ``2-step'') in the fully nonparametric model under the null hypothesis. 
			The left and right panels report results for the quadratic case, \(h_{0, 1}(z) = z^2 / \sqrt{2}\), and the 
			non-polynomial case, \(h_{0, 2}(z) = \sqrt{3\sqrt{3}} \times z \times \exp(-{z^2}/{2})\).
			The upper and lower panels report results with a moderate (\(\rho_{\eps V} = 0.5\)) or a high (\(\rho_{\eps V} = 0.8\)) level of endogeneity, for \(n = 100\) and \(n = 400\), at nominal levels of 5\% and of 10\%.
			The rejection probabilities are based on 5,000 Monte Carlo replications and are displayed in percentages.}
	\end{table}
	
	
	We also study the behavior of our test under the alternative. 
	To do so, we compute the rejection rates of the test of the hypothesis \(\theta_{\mathcal{H}_0} = \theta_0 + \gamma\), for a deviation \(\gamma \in (0, 1)\) from the null, using 5,000 Monte Carlo replications.
	In Figures~\ref{fig:power_curve_RKHS_CCK_case_quadratic_moderate_endogeneity} and~\ref{fig:power_curve_RKHS_CCK_case_nonpolynomial_high_endogeneity}, we compare our one-step regression test (displayed in blue) with the two-step series regression test (displayed in red).
	Figure~\ref{fig:power_curve_RKHS_CCK_case_quadratic_moderate_endogeneity} presents the results for the quadratic function $h_{0,1}$ and a moderate level of endogeneity, while Figure~\ref{fig:power_curve_RKHS_CCK_case_nonpolynomial_high_endogeneity} presents the results for the non-polynomial function $h_{0,2}$ and a high level of endogeneity.
	The other settings display similar results and are reported in Appendix~\ref{section:appendix_results_simulations}, Figures~\ref{fig:power_curve_RKHS_CCK_case_quadratic_high_endogeneity} and~\ref{fig:power_curve_RKHS_CCK_case_nonpolynomial_moderate_endogeneity}.
	They show that our method has considerably higher power for the small sample size \(n = 100\).
	For \(n = 400\), both methods have comparable power: the two-step test slightly outperforms in the quadratic case \(h_{0, 1}\), whereas ours outperforms in the non-polynomial case \(h_{0, 2}\).
	
	\begin{figure}[ht]
		\centering
		\caption{Power curves in the fully nonparametric model: quadratic case and moderate endogeneity.}
		\label{fig:power_curve_RKHS_CCK_case_quadratic_moderate_endogeneity}
		\includegraphics[width=0.88\textwidth]{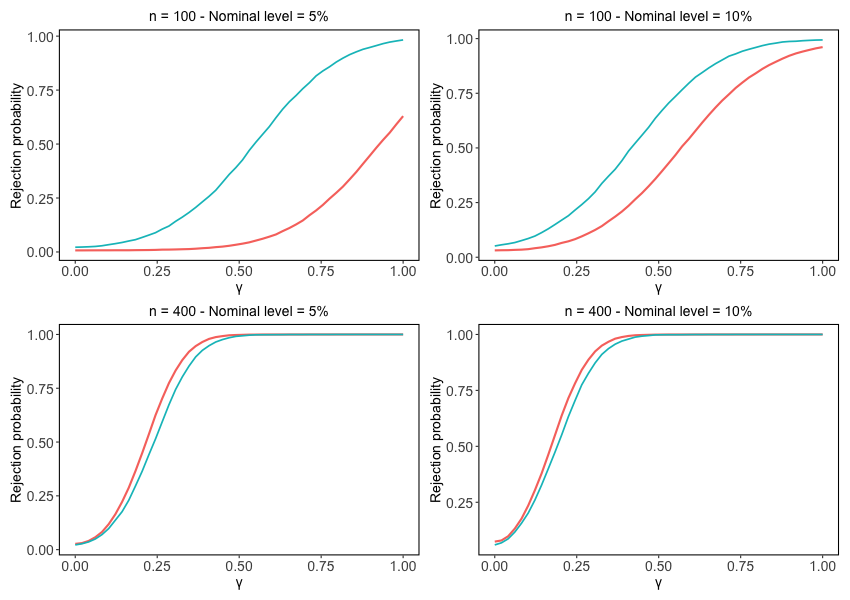}
		\tablefigurenote{Power curves for our one-step regression test (in blue) and for the two-step series regression test (in red) in the fully nonparametric model, for the quadratic case, \(h_{0, 1}(z) = z^2 / \sqrt{2}\), with a moderate level of endogeneity, \(\rho_{\eps V} = 0.5\). The upper and lower panels report results for \(n = 100\) and \(n = 400\). The left and right panels report results at nominal levels of 5\% and 10\%. The x-axis shows the deviation $\gamma$ from the null, and the y-axis shows the rejection probability based on 5,000 Monte Carlo replications.}
	\end{figure}
	
	\begin{figure}[ht]
		\centering
		\caption{Power curves in the fully nonparametric model: non-polynomial case and high endogeneity.}
		\label{fig:power_curve_RKHS_CCK_case_nonpolynomial_high_endogeneity}
		\includegraphics[width=0.88\textwidth]{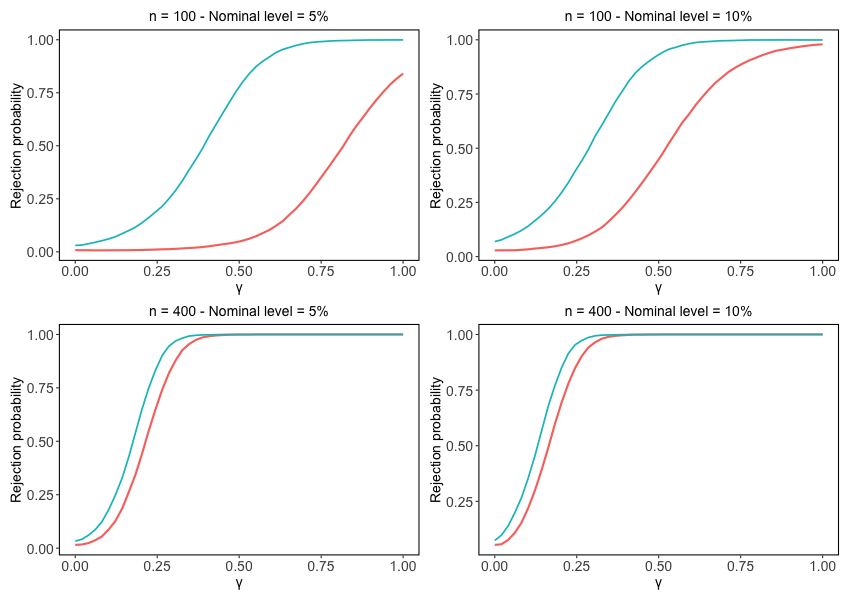}
		\tablefigurenote{Power curves for our one-step regression test (in blue) and for the two-step series regression test (in red) in the fully nonparametric model, for the non-polynomial case, \(h_{0, 2}(z) = \sqrt{3\sqrt{3}} \times z \times \exp(-{z^2}/{2})\), with a high level of endogeneity, \(\rho_{\eps V} = 0.8\). The upper and lower panels report results for \(n = 100\) and \(n = 400\). The left and right panels report results at nominal levels of 5\% and 10\%. The x-axis shows the deviation $\gamma$ from the null, and the y-axis shows the rejection probability based on 5,000 Monte Carlo replications.}
	\end{figure}
	
	Overall, when considering both size and power, our test performs best in almost all cases, often substantially outperforming the two-step series regression test, while in the few scenarios where it is slightly outperformed, the difference is minor.
	
	\paragraph{Partially linear model.}
	
	We now study the finite-sample performance of our test in the general partially linear setting corresponding to Equation \eqref{eq: npiv regression}.
	Similarly to the fully nonparametric model, ($Z$, $\varepsilon$) are drawn from Equation \eqref{eq: simulations setting for epx and Z}.  
	$U$ and $V$ are mutually independent standard Gaussians, and are jointly independent of the covariate $X$ and the instrument $W$. The pair $(X,W)$ is drawn from a zero-mean bivariate Gaussian, with $\operatorname{Var}(X)=\operatorname{Var}(W)=1$, and a correlation coefficient between $X$ and $W$ equal to 0.5.  
	In the simulations, we set the slope coefficient of \(X\) to~1.
	We consider the same two scenarios as in the fully nonparametric model, depending on \(\rho_{\eps V}\) to monitor the magnitude of the endogeneity, as well as the two cases of a quadratic and a non-polynomial functional form for \(h_0\). 
	Our test is implemented as described in the fully nonparametric case. 
	
	\begin{table}[ht]
		\centering
		\caption{Rejection probabilities under the null hypothesis in the partially linear model.}    \label{tab:rejection_rate_under_HO_partlylinearinX}
		\begin{tabular}{c c | c c | c c}
			\hline 
			\multicolumn{2}{c|}{Case} &
			\multicolumn{2}{c|}{quadratic} & \multicolumn{2}{c}{non-polynomial} \\
			\hline 
			\multicolumn{2}{c|}{Nominal level} & 5\% & 10\% & 5\% & 10\% \\
			\hline
			\(\rho_{\eps V} = 0.5\) &\(n = 100\) & 
			1.64 & 4.82 &
			1.92 & 5.18
			\\
			& \(n = 400\) &
			2.06 & 5.52 &
			2.46 & 6.62
			\\
			\hline 
			\(\rho_{\eps V} = 0.8\) & \(n = 100\) & 
			1.42 & 4.14 &
			2.30 & 5.28
			\\
			& \(n = 400\) & 
			2.04 & 5.22 &
			2.28 & 6.20
			\\
			\hline
		\end{tabular}
		\tablefigurenote{Rejection probabilities for our one-step regression test in the partially linear model under the null hypothesis.
			The left and right panels report results for the quadratic case, \(h_{0, 1}(z) = z^2 / \sqrt{2}\), and the 
			non-polynomial case, \(h_{0, 2}(z) = \sqrt{3\sqrt{3}} \times z \times \exp(-{z^2}/{2})\).
			The upper and lower panels report results with a moderate (\(\rho_{\eps V} = 0.5\)) or a high (\(\rho_{\eps V} = 0.8\)) level of endogeneity, for \(n = 100\) and \(n = 400\), at nominal levels of 5\% and of 10\%.
			The rejection probabilities are based on 5,000 Monte Carlo replications and are displayed in percentages.}
	\end{table}
	
	\medskip
	
	Table~\ref{tab:rejection_rate_under_HO_partlylinearinX} reports the rejection rates of our test under the null hypothesis \(\theta_{\mathcal{H}_0} = \theta_0\) at the nominal levels of 5\% and 10\%, based on 5,000 Monte Carlo replications.
	The results show that our test controls Type~I error with rejection rates below the nominal levels and approaching them when the sample size increases.
	Like in the fully nonparametric model, we also study the behavior of our test under the alternative hypothesis by testing \(\theta_{\mathcal{H}_0} = \theta_0 + \gamma\), with \(\gamma\) ranging between 0 and 1.
	Figure~\ref{fig:power_curve_PLinX_case_quadratic_moderate_endogeneity} considers the quadratic case (\(h_{0, 1}\)) with a moderate level of endogeneity (\(\rho_{\eps V} = 0.5\)).
	The results are qualitatively the same for the other settings reported in Appendix~\ref{section:appendix_results_simulations} (Figures~\ref{fig:power_curve_PLinX_case_quadratic_high_endogeneity} to~\ref{fig:power_curve_PLinX_case_nonpolynomial_high_endogeneity}). 
	The power curves illustrate the consistency of our test.
	Overall, these simulations indicate that the solid performance of our test in the fully nonparametric model is preserved when the model is extended to a partially linear specification with additional exogenous covariates. 
	
	\begin{figure}[H]
		\centering
		\caption{Power curves in the partially linear model: quadratic case and moderate endogeneity.} \label{fig:power_curve_PLinX_case_quadratic_moderate_endogeneity}
		\includegraphics[width=0.88\textwidth]{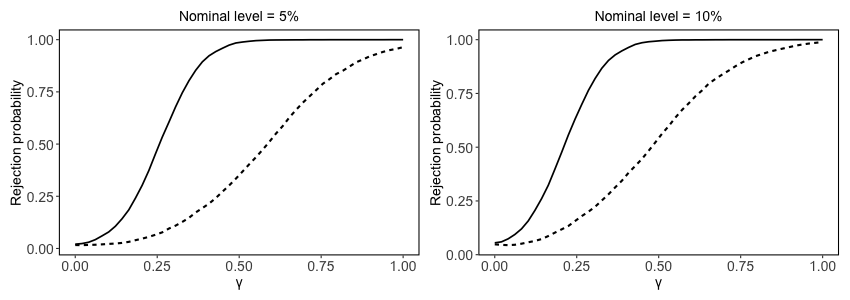}
		\tablefigurenote{Power curves for our one-step regression test in the partially linear model, for the quadratic case, $h_{0,1}(z) = z^2 / \sqrt{2}$, with a moderate level of endogeneity, $\rho_{\varepsilon V} = 0.5$. The dashed and solid lines correspond to $n = 100$ and $n = 400$, respectively. The left and right panels report results at nominal levels of 5\% and 10\%. The x-axis shows the deviation $\gamma$ from the null, and the y-axis shows the rejection probability based on 5,000 Monte Carlo replications.}
	\end{figure}
	
	\section{Empirical applications}
	\label{sec:empirical_application}
	
	We apply our method to three empirical datasets to illustrate its potential.
	
	\subsection{Class size effects on student achievement}\label{subsec: class size effects on student achievement}
	
	Our first application revisits the seminal study by \citet{angrist1999using} on the effects of class size on student achievement. 
	We briefly recall the setting and refer to the original article for further details.
	The data come from a national testing program conducted in Israeli elementary schools in June 1991, with test scores in mathematics and reading comprehension scaled from 0 to 100. 
	\citet{angrist1999using} exploits Maimonides' Rule -- an administrative regulation that caps class sizes at 40 students -- to construct the predicted class size, \(W = \text{Enrollment} / (\lfloor(\text{Enrollment} - 1)/40\rfloor + 1)\).
	The discontinuities at enrollment multiples of 40 provide exogenous variation in class size, making \(W\) a valid instrument for the actual class size, denoted \(Z\) in our notation.
	We focus on fifth-grade students, using a sample of \(n = 2{,}024\) classes.
	Following the original study, we also consider an exogenous covariate \(X\), the percentage of disadvantaged students (PD index).
	The outcome variable \(Y\) is the average grade across students for a given class, either in mathematics or reading.
	
	\medskip
	
	We estimate the partially linear IV model~\eqref{eq: npiv regression}.
	Our parameter of interest is the AME \(\theta_0 = \operatorname{E}[h'_0(Z)]\), which here measures the average impact of a one-student increase in class size on test scores.
	In comparison, we recall that a classical application of two-stage least squares assumes a linear relationship between average test score~\(Y\) and class size~\(Z\), so that the AME equals the corresponding coefficient.
	We implement our RKHS-based method as described in Sections~\ref{sec:implementation} and \ref{sec:small_sample_behavior_simulations}.
	Inference is conducted using the Bayesian bootstrap test with \(B = 499\) replications.
	To report the confidence intervals, we invert the bootstrap test of Algorithm~\ref{algo:our_bootstrap_test}. Specifically, for any \(\theta\), we do not reject the null hypothesis {\(\theta_{0} = \theta\)} if and only if \(|\thetahat - \theta| \leq \widehat{q}_{1-\alpha}\), with \(\widehat{q}_{1-\alpha}\) the \(1 - \alpha\) quantile of the bootstrap distribution of \(|\thetahat_b - \thetahat|\). This leads to the interval \([\thetahat \pm \widehat{q}_{1-\alpha}]\) that uses the quantiles of the bootstrapped statistic 
	\(|\thetahat_b - \thetahat|\) as critical values.
	
	\medskip
	
	Table~\ref{tab:fifth_grade_results} presents the results. Applying our method, we find that the effects of class size on test scores are not statistically significant in either mathematics or reading.
	This finding contrasts with the statistically significant negative effects estimated by two-stage least squares in the original study: \(-0.372\) for reading and \(-0.355\) for mathematics (see Table~IV of \citet{angrist1999using}).
	Our results are consistent with \citet{horowitz2011applied}, Section~5.2, who notes that
	``the [fully] nonparametric [IV] model does not support the conclusion drawn from the linear model that increases in class size are associated with decreased test scores'' (page~374).
	Our results confirm the findings of \citet[Section~5.2]{horowitz2011applied}, which are based on an informal bootstrap inference procedure for the nonparametric IV model. He notes that  ``the [fully] nonparametric [IV] model does not support the conclusion drawn from the linear model that increases in class size are associated with decreased test scores'' (page~374).
	Importantly, our analysis shows that, 
	absent strong functional form assumptions, the data provide limited information about the precise nature of class-size effects. 
	This highlights the value of our inference procedure, which properly accounts for uncertainty in a flexible semiparametric model with theoretical guarantees, while remaining computationally simple through a single regularization parameter.
	
	\begin{table}[ht]
		\centering
		\caption{Estimates and confidence intervals for the AME of class size on test scores.}
		\begin{tabular}{lccc}
			\hline
			Outcome & Estimate & 90\% CI & 95\% CI \\
			\hline
			Mathematics &
			\(0.133\) & \([-0.339,0.605]\) & \([-0.450, 0.716] \) \\
			Reading &
			\(-0.076\) & \([-0.387, 0.235]\) & \([-0.488, 0.337]\) \\
			\hline
		\end{tabular}
		\tablefigurenote{Estimates and confidence intervals from our method for the AME of class size on student mathematics and reading scores in fifth grade. The confidence intervals are obtained by inverting the bootstrap test with \(B = 499\) replications.
			Data: \citet{angrist1999using}, \(n = 2{,}024\) observations.}
		\label{tab:fifth_grade_results}
	\end{table}
	
	\subsection{Trade effects on per capita income}
	
	The previous application uses samples of about 2,000 observations.
	We now turn to two applications with smaller sample sizes to illustrate the small-sample performance of our method.
	As we show below, our procedure proves able to detect significant effects while allowing a flexible semiparametric relationship even with small samples.
	
	\medskip
	
	First, we revisit \citet{frankel1999}, who investigate whether international trade causes economic growth using a sample of \(n = 150\) countries.
	The endogenous treatment, the trade share \(Z\) of a given country, is measured as imports plus exports divided by GDP.
	\citet{frankel1999} address the endogeneity issue by a constructed trade share based on geographic factors (bilateral distances, country sizes, and whether countries share borders or are landlocked).
	We refer to the original article for further details.
	Here, we use their data and instrumental variable strategy while allowing for a flexible semiparametric relationship between trade and income.
	More precisely, we estimate the following partially linear model 
	\begin{equation}
		\label{eq:our_specification_for_trade}
		\log(Y_i) = h_0(Z_i) + \beta_1 N_i + \beta_2 A_i  + \varepsilon_i, \quad \operatorname{E}\{\varepsilon_i | N_i, A_i,W_i\} = 0,
	\end{equation}
	where the outcome $Y_i$ is GDP per capita in country $i$, the endogenous treatment $Z_i$ is the trade share, the exogenous covariates are population (\(N_i\)) and area (\(A_i\)), and \(W_i\) is the constructed geographic component of trade obtained from the original study. 
	Equation~\eqref{eq:our_specification_for_trade} corresponds to the main specification of \citet{frankel1999}, with the difference that the trade share enters nonlinearly through \(h_0\), which is nonparametrically specified.
	Our application thus relaxes the linearity assumption, allowing for a flexible relationship between trade and income.
	
	\medskip
	
	Table~\ref{tab:our_results_trade} presents the estimation and inference results for the AME using our method.
	We find a positive and statistically significant effect of the trade share on GDP per capita, which illustrates that our method can detect effects even with a few hundred observations.
	The resulting estimate implies that a one-percentage-point increase in trade share raises GDP per capita by 1.15 percent.
	In comparison, \citet{frankel1999} find that a one-percentage-point increase raises income per capita by 1.97 percent (see their Table~3, column~2, page~387).
	This suggests that allowing for nonlinearities in the relationship changes the estimate, reducing the effects of trade on income, although the effects remain of the same order of magnitude and statistically significant.
	Based on \citet{frankel1999} data and their identification strategy, our semiparametric method thus confirms a positive causal effect of trade on income while relaxing the linearity assumption.
	This application illustrates the robustness of our flexible inference procedure, which avoids imposing functional form restrictions while delivering reliable inference even in small samples.
	
	\begin{table}[ht]
		\centering
		\caption{Estimates and confidence intervals for the AME of trade share on GDP per capita.}
		\label{tab:our_results_trade}
		\begin{tabular}{lccc}
			\hline Outcome & Estimate & 90\% CI & 95\% CI \\
			\hline
			GDP per capita & \(1.15\) & \([0.42, 1.89]\) & \([0.38, 1.92]\)
			\\
			\hline
		\end{tabular}
		\tablefigurenote{Estimates and confidence intervals from our method for the AME of trade share on GDP per capita. The confidence intervals are obtained by inverting the bootstrap test with \(B = 499\) replications. Data: \citet{frankel1999}, \(n = 150\) observations.}
	\end{table}
	
	\subsection{Advertisement effects on reader newspaper demand}
	
	Our third application is based on \citet{sokullu2016semi}, who proposes an empirical semiparametric model for two-sided markets applied to local daily newspapers in the US.
	Two-sided markets (where a platform serves two distinct groups of users who benefit from each other's participation) are a central object of study in empirical industrial organization.
	Examples include newspaper platforms, connecting readers and advertisers \citep{rysman2004competition, chandra2009mergers}, and payment card networks, connecting merchants and cardholders \citep{rysman2007empirical}; see, e.g., \citet{rochet2006two} or \citet{armstrong2006competition} for theoretical contributions and \citet{sanchez2021multisided} for a recent review.
	A key feature of these markets is the presence of \emph{indirect network effects}: the value that agents on one side derive from the platform depends on the level of participation on the other side.
	In the newspaper context, advertisers value newspapers with larger readership, and readers may be attracted or deterred by the volume of advertising.
	Correctly estimating these network effects is important for competition policy analysis.
	
	Several empirical studies of two-sided markets have estimated network effects under linearity assumptions.
	\citet{sokullu2016semi} relaxes this assumption by specifying the network effect functions nonparametrically. 
	Specifically, she estimates a nonparametric IV model with Tikhonov regularization on a sample of $n = 117$ monopoly newspaper markets.
	The main finding is that network effects are neither linear nor monotonic: readers' benefits initially increase with advertising but decrease beyond a threshold, and similarly for advertisers with respect to circulation.
	
	\medskip
	
	Here, to illustrate our method, we consider the readers' side and estimate the demand function.
	\citet{sokullu2016semi} uses the results from the nonparametric IV model to select a parametric specification and re-estimates the model.
	On the readers' side, the nonparametric estimates suggest approximating the network effect of advertising on readers by a third-order polynomial, leading to the following equation
	\begin{equation}
		\label{eq:sokullu_reader_demand}
		\log\!\left(\frac{1 - N^r}{N^r}\right)
		=
		\alpha_0 
		+
		\alpha_1 N^a
		+
		\alpha_2 (N^a)^2
		+
		\alpha_3 (N^a)^3
		+
		\beta_0 P^r
		+ U,
	\end{equation}
	where all variables are defined at the level of each local newspaper market:
	$N^r$ is the share of readers buying the newspaper;
	$N^a$ is the share of advertisements in the newspaper;
	$P^r$ is the daily cover price;
	and $U$ is an error term capturing unobserved newspaper characteristics affecting readers' demand.
	The left-hand side is the log-odds transformation of the readers' market share, which arises from the parametric specification of the distribution of readers' and advertisers' benefits (see Section~4.3.1 in \citet{sokullu2016semi}).
	The polynomial $\alpha_1 N^a + \alpha_2 (N^a)^2 + \alpha_3 (N^a)^3$ captures the \emph{indirect network effect} of advertising on reader demand: it measures how the presence of advertisers on the platform affects readers' propensity to buy the newspaper.
	We apply our RKHS-based method to nonparametrically estimate this network effect by considering
	\begin{equation}
		\label{eq:sokullu_ours}
		\log\!\left(\frac{1 - N^r}{N^r}\right)
		=
		h_0(N^a)
		+
		\beta_0 P^r    
		+ U.
	\end{equation}
	Compared to~\eqref{eq:sokullu_reader_demand}, the advertising share $N^a$ enters through the unknown function $h_0$, which we estimate nonparametrically while maintaining the linear specification for the price.
	Our goal is to estimate and perform inference on the average marginal effect $\theta_0 = \operatorname{E}\{h_0'(N^a)\}$, which summarizes how an increase in advertising share affects readers' demand on average.
	We provide a formal inference procedure for the AME, complementing \citet{sokullu2016semi}'s nonparametric estimation results.
	
	Both $N^a$ and $P^r$ are endogenous in this model: the advertising share is jointly determined with readership on the platform, and the cover price is set by the newspaper in response to unobserved demand characteristics.
	Our method can be straightforwardly adapted to handle such cases with multiple endogenous regressors.
	We require two instruments $(W_1, W_2)$ and adapt the matrix $\boldsymbol{F}$ as $[\mathcal{F}_\mu((W_{1i}, W_{2i})^\top - (W_{1j}, W_{2j})^\top) : i, j = 1, \ldots, n]$.
	Following \citet{sokullu2016semi}, we use the area of the city (in square miles) as $W_1$ and the average wages in the printing industry at the county level as $W_2$.
	The first instrument is correlated with advertising rates but plausibly independent of unobserved reader demand characteristics; the second is an income-related variable correlated with circulation but independent of unobserved newspaper characteristics in the advertising equation.
	
	\medskip
	
	Table~\ref{tab:sokullu_our_results} presents the results obtained from our method, implemented as described in the previous sections.
	As a comparison, \citet{sokullu2016semi} obtains the following estimates for \eqref{eq:sokullu_reader_demand}: $\widehat{\alpha}_1 = -48.34$, $\widehat{\alpha}_2 = 96.75$, and $\widehat{\alpha}_3 = -62.06$.
	Combined with the descriptive statistics for $N^a$ reported in Table~1 of \citet{sokullu2016semi}, the average marginal effect of $N^a$ implied by the cubic polynomial specification is thus
	$
	\widehat{\alpha}_1
	+
	2 \, \widehat{\alpha}_2
	\, 
	\operatorname{E}_n\{N^a\}
	+
	3 \, \widehat{\alpha}_3 \, 
	\operatorname{E}_n\{(N^a)^2\}
	=
	-8.19.
	$
	In contrast, our nonparametrically estimated AME is $-5.53$ and is statistically significant.
	The estimate is smaller, although both remain of the same order of magnitude; in particular, $-8.19$ falls within our confidence intervals.
	Recall that the cubic specification in~\eqref{eq:sokullu_reader_demand} is itself guided by \citet{sokullu2016semi}'s prior nonparametric estimation.
	The similarity between our AME estimate and that implied by the cubic specification therefore yields supporting evidence for the validity of our method.
	
	\begin{table}[ht]
		\centering
		\caption{Estimates and confidence intervals for the average marginal network effect of advertisement on reader demand.}
		\label{tab:sokullu_our_results}
		\begin{tabular}{lccc}
			\hline Outcome & Estimate & 90\% CI & 95\% CI \\
			\hline
			Reader demand & $-5.53$ & $[-8.50, -2.56]$ & $[-9.06, -2.01]$ \\
			\hline
		\end{tabular}
		\tablefigurenote{Estimates and confidence intervals from our method for the average marginal network effects of advertisement on reader demand. The confidence intervals are obtained by inverting the bootstrap test with \(B = 499\) replications. Data: \citet{sokullu2016semi}, \(n = 117\) observations.}
	\end{table}
	
	\medskip
	
	In conclusion, the three applications presented in this section, spanning sample sizes from 117 to over 2,000 observations, illustrate the value of our method: it provides a formal inference procedure for the AME that delivers economically meaningful results and maintains computational simplicity with a single regularization parameter.
	
	\section{Conclusion}
	\label{sec:conclusion}
	
	We propose a bootstrap inference procedure for the average marginal effect of an endogenous treatment in the partially linear IV model.
	Our procedure builds on the RKHS framework and requires a single regularization parameter, which simplifies its practical implementation relative to existing approaches.
	Our approach can be extended to consider linear functionals of nonparametric IV regressions.
	In our simulations, we show that our procedure has good size control and power in small and moderate samples, while three empirical applications illustrate its effectiveness on real datasets. 
	
	\section*{Appendix}
	
	The Appendix is divided into four sections.
	In Appendix~\ref{section:appendix_results_simulations}, we provide additional simulation results.
	Appendix~\ref{sec: appendix proofs of of the main results} contains the proofs of our main results, namely Propositions \ref{prop:computation_solution_empirical_program}, \ref{prop: ifr}, and \ref{prop: bootstrap validity}. 
	In Appendix~\ref{sec: appendix auxiliary lemmas}, we collect the auxiliary lemmas and their proofs.
	Finally, Appendix~\ref{sec: low level conditions} contains primitive conditions that satisfy the high-level assumptions.
	
	\appendix
	
	\section{Additional results for the simulations}
	\label{section:appendix_results_simulations}
	
	This section presents the power curves for the remaining settings not shown in Section~\ref{sec:small_sample_behavior_simulations}.
	
	\paragraph{Fully nonparametric model.}
	
	In the fully non-parametric model, Figure~\ref{fig:power_curve_RKHS_CCK_case_quadratic_high_endogeneity} presents the quadratic case with high level of endogeneity while Figure~\ref{fig:power_curve_RKHS_CCK_case_nonpolynomial_moderate_endogeneity} presents the non-polynomial case with moderate level of endogeneity.
	The two remaining settings are presented in the body of the article (Figure~\ref{fig:power_curve_RKHS_CCK_case_quadratic_moderate_endogeneity} for the quadratic case with moderate level of endogeneity and Figure~\ref{fig:power_curve_RKHS_CCK_case_nonpolynomial_high_endogeneity} for the non-polynomial case with high level of endogeneity).
	
	\paragraph{Partially linear model.}
	
	In the partially linear model, Figure~\ref{fig:power_curve_PLinX_case_quadratic_moderate_endogeneity} presents the quadratic case with high level of endogeneity, Figure~\ref{fig:power_curve_PLinX_case_nonpolynomial_moderate_endogeneity} and~\ref{fig:power_curve_PLinX_case_nonpolynomial_high_endogeneity} show the non-polynomial case with, respectively, low and high level of endogeneity. 
	The remaining setting (quadratic case with moderate level of endogeneity) is presented in Section~\ref{sec:small_sample_behavior_simulations} (Figure~\ref{fig:power_curve_PLinX_case_quadratic_moderate_endogeneity}).
	
	\begin{figure}[H]
		\centering
		\caption{Power curves in the fully nonparametric model: quadratic case and high endogeneity.}
		\label{fig:power_curve_RKHS_CCK_case_quadratic_high_endogeneity}
		\includegraphics[width=0.88\textwidth]{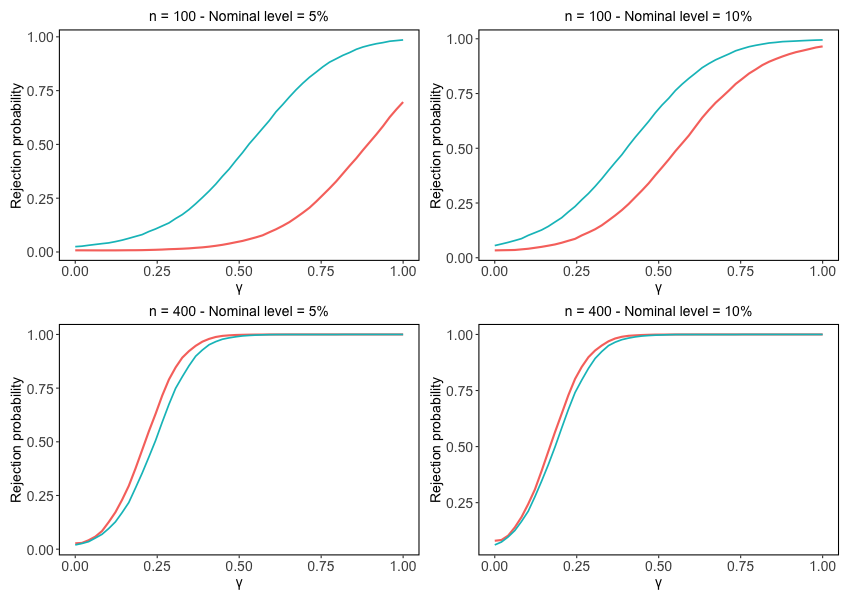}
		\tablefigurenote{Power curves for our one-step regression test (in blue) and for the two-step series regression test (in red) in the fully nonparametric model, for the quadratic case, \(h_{0, 1}(z) = z^2 / \sqrt{2}\), with a high level of endogeneity, \(\rho_{\eps V} = 0.8\). The upper and lower panels report results for \(n = 100\) and \(n = 400\). The left and right panels report results at nominal levels of 5\% and 10\%. The x-axis shows the deviation $\gamma$ from the null, and the y-axis shows the rejection probability based on 5,000 Monte Carlo replications.}
	\end{figure}
	
	\begin{figure}[H]
		\centering
		\caption{Power curves in the fully nonparametric model: non-polynomial case and moderate endogeneity.}
		\label{fig:power_curve_RKHS_CCK_case_nonpolynomial_moderate_endogeneity}
		\includegraphics[width=0.88\textwidth]{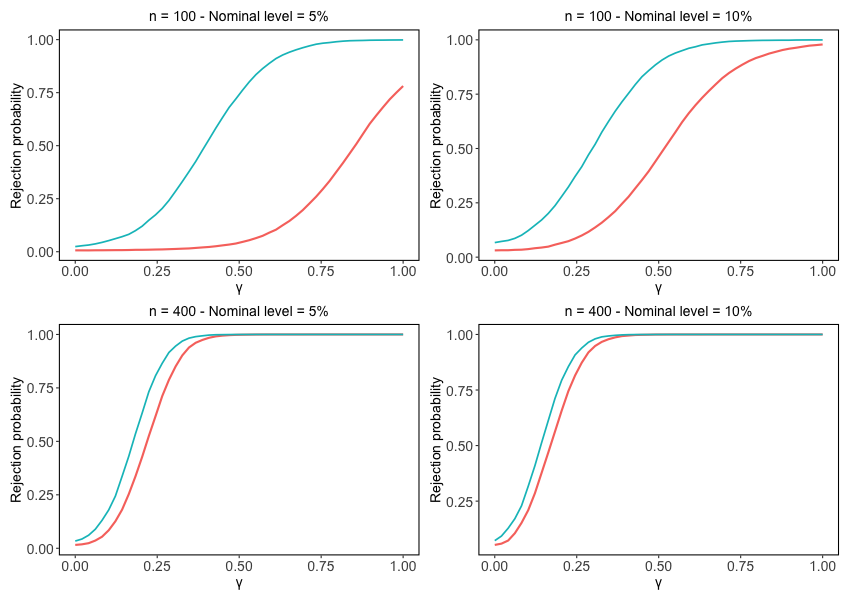}
		\tablefigurenote{Power curves for our one-step regression test (in blue) and for the two-step series regression test (in red) in the fully nonparametric model, for the non-polynomial case, \(h_{0, 2}(z) = \sqrt{3\sqrt{3}} \times z \times \exp(-{z^2}/{2})\), with a moderate level of endogeneity, \(\rho_{\eps V} = 0.5\). The upper and lower panels report results for \(n = 100\) and \(n = 400\). The left and right panels report results at nominal levels of 5\% and 10\%. The x-axis shows the deviation $\gamma$ from the null, and the y-axis shows the rejection probability based on 5,000 Monte Carlo replications.}
	\end{figure}
	
	\begin{figure}[H]
		\centering
		\caption{Power curves in the partially linear model: quadratic case and high endogeneity.} \label{fig:power_curve_PLinX_case_quadratic_high_endogeneity}
		\includegraphics[width=0.88\textwidth]{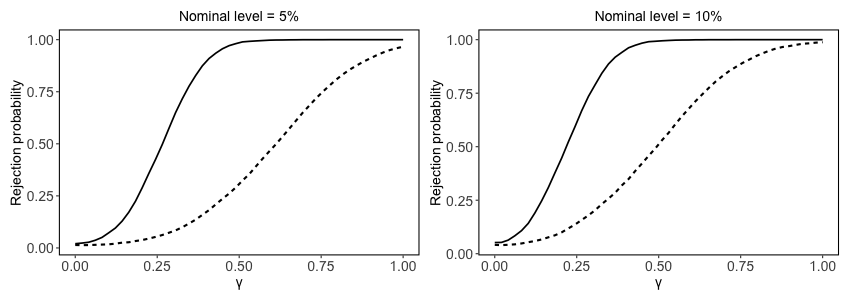}
		\tablefigurenote{Power curves for our one-step regression test in the partially linear model, for the quadratic case, $h_{0,1}(z) = z^2 / \sqrt{2}$, with a high level of endogeneity, $\rho_{\varepsilon V} = 0.8$. The dashed and solid lines correspond to $n = 100$ and $n = 400$, respectively. The left and right panels report results at nominal levels of 5\% and 10\%. The x-axis shows the deviation $\gamma$ from the null, and the y-axis shows the rejection probability based on 5,000 Monte Carlo replications.}
	\end{figure}
	
	\begin{figure}[H]
		\centering
		\caption{Power curves in the partially linear model: non-polynomial case and moderate endogeneity.} \label{fig:power_curve_PLinX_case_nonpolynomial_moderate_endogeneity}
		\includegraphics[width=0.88\textwidth]{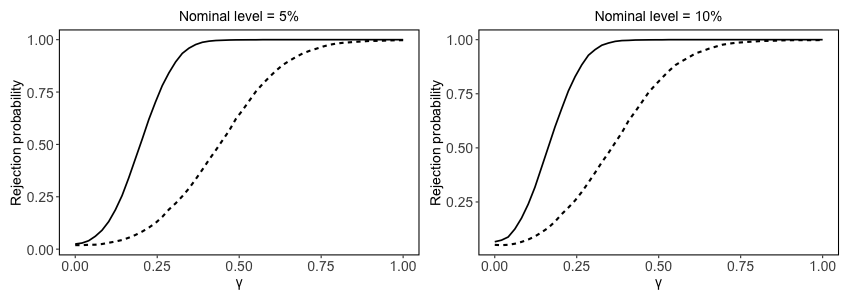}
		\tablefigurenote{Power curves for our one-step regression test in the partially linear model, for the non-polynomial case, \(h_{0, 2}(z) = \sqrt{3\sqrt{3}} \times z \times \exp(-{z^2}/{2})\), with a moderate level of endogeneity, \(\rho_{\varepsilon V} = 0.5\). The dashed and solid lines correspond to $n = 100$ and $n = 400$, respectively. The left and right panels report results at nominal levels of 5\% and 10\%. The x-axis shows the deviation $\gamma$ from the null, and the y-axis shows the rejection probability based on 5,000 Monte Carlo replications.}
	\end{figure}
	
	\begin{figure}[H]
		\centering
		\caption{Power curves in the partially linear model: non-polynomial case and high endogeneity.} \label{fig:power_curve_PLinX_case_nonpolynomial_high_endogeneity}
		\includegraphics[width=0.88\textwidth]{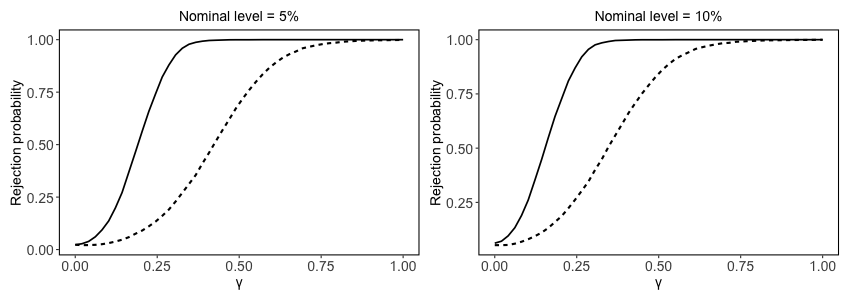}    
		\tablefigurenote{Power curves for our one-step regression test in the partially linear model, for the non-polynomial case, \(h_{0, 2}(z) = \sqrt{3\sqrt{3}} \times z \times \exp(-{z^2}/{2})\), with a high level of endogeneity, \(\rho_{\varepsilon V} = 0.8\). The dashed and solid lines correspond to $n = 100$ and $n = 400$, respectively. The left and right panels report results at nominal levels of 5\% and 10\%. The x-axis shows the deviation $\gamma$ from the null, and the y-axis shows the rejection probability based on 5,000 Monte Carlo replications.}
	\end{figure}
	
	\newpage
	
	\section{Proofs of the main results}
	\label{sec: appendix proofs of of the main results}
	Before delving into the proofs, we recall some standard definitions and notation. For a linear operator $\operatorname{C}:\mathcal{X}\to \mathcal{Y}$ between inner product spaces$(\mathcal{X}, \left<\cdot,\cdot\right>_{\mathcal{X}})$  and $(\mathcal{Y}, \left<\cdot,\cdot\right>_{\mathcal{Y}})$ with associated norms $\|\cdot\|_{\mathcal{X}}=\sqrt{\left<\cdot,\cdot\right>_\mathcal{X}}$ and $\|\cdot\|_{\mathcal{Y}}=\sqrt{\left<\cdot,\cdot\right>_\mathcal{Y}}$,
	the operator norm is defined as  
	\begin{equation*}
		\|\operatorname{C}\|_{op}:=\sup_{\varphi\in\mathcal{X},\|\varphi\|_{\mathcal{X}}=1}\|\operatorname{C}\varphi\|_{\mathcal{Y}}\, .
	\end{equation*}
	An operator is bounded if $\|\operatorname{C}\|_{op}<\infty$; see \citet[Chapter 2]{kress1999linear}. The adjoint of $\operatorname{C}$ is the operator $\operatorname{C}^*:\mathcal{Y}\to \mathcal{X}$ such that $\left<\operatorname{C}\varphi,\psi\right>_{\mathcal{Y}}=\left<\varphi,\operatorname{C}^*\psi\right>_{\mathcal{X}}$ for any $(\varphi,\psi)\in\mathcal{X}\times \mathcal{Y}$. Whenever $\operatorname{C}$ is linear and bounded, it admits an adjoint $\operatorname{C}^*$ that is also linear and bounded, and such that $\|\operatorname{C}^*\|_{op}=\|\operatorname{C}\|_{op}$. See \citet[Theorem 4.9]{kress1999linear}.\\
	
	\noindent \textbf{Proof of Proposition \ref{prop:computation_solution_empirical_program}}.  
	Let us define 
	\begin{align}\label{eq: Ahat and shat}
		\widehat{\A}h:=\E_n\{h(Z_i) \exp(\textbf{i}\,(X^T_i,W_i)\, \cdot)\}\,&,\, \widehat{\A}:\mathcal{H}\to L^2_\mu\,,\nonumber \\
		\widehat{\B}\beta=\operatorname{E}_n \{X_i^T\beta \exp(\textbf{i}\,(X_i^T,W_i)\,\cdot)\}\,,\, \widehat{\operatorname{\B}}:\mathbb{R}^p\to L^2_\mu\,,\nonumber \\
		\text{ and }\widehat s:=\E_n \{Y_i \exp(\textbf{i}\,(X^T_i,W_i)\, \cdot)\}\,&,\, \widehat s\in L^2_\mu\, .
	\end{align}
	We reformulate the the minimization problem in \eqref{eq:empirical_program} as 
	\begin{align}\label{eq: min problem reformulation}    \min_{h\in\mathcal{H},\beta\in\mathbb{R}^p}\int_{ }\Big|\operatorname{E}_n \Big\{[Y_i - X_i^T\beta - h(Z_i)]\exp(\textbf{i}(X_i^T,W_i)t)\Big\}\Big|^2& d \mu(t)+\lambda \|h\|^2_\mathcal{H}\\
		=&\min_{h\in\mathcal{H},\beta\in\mathbb{R}^p} \|\widehat s- \widehat{\B}\beta-\widehat{\A}h \|^2_\mu + \lambda \|h\|^2_{\mathcal{H}}\ .\nonumber 
	\end{align}
	By Lemma \ref{lem: compactness, and expressions of h0, hhat}\ref{lem: compactness, and expressions of h0, hhat:: hhat and betahat}, the solution to the above problem exists, is unique, and is given by  
	\begin{equation}\label{eq: hhat Tikhonov expression}
		\widehat h=(\widehat{\T}^*\widehat{\T}+\lambda\I)^{-1}\widehat{\T}^*\widehat r\,, \, \widehat \beta=(\widehat \B^*\widehat \B)^{-1}\widehat \B^*(\widehat s - \widehat \A \widehat h)\,,
	\end{equation}
	where $\widehat \T=(\I-\widehat \P)\widehat \A$, $\widehat r=(\I-\widehat{\P})\widehat{s}$, and $\widehat \P=\widehat \B(\widehat \B^*\widehat \B)^{-1}\widehat \B^*$ (see Lemma \ref{lem: rates and compacness}\ref{lem: rates and compacness:: expression of P and Phat}).\\
	Let us now show that $\widehat h$ and $\widehat \beta$ can be expressed as in \eqref{eq: expresson of hhat and betahat}. We begin by using arguments similar to those in the proof of  \citet[Proposition 12.32]{wainwright2019high}. Let  $\mathbb{L}_\alpha:=\{\sum_{i=1}^n \alpha_i K(\cdot,Z_i):(\alpha_1,\ldots,\alpha_n)\in\mathbb{R}_n\}$. Since $K(\cdot,Z_i)\in \mathcal{H}$ for each $i=1,\ldots,n$, $\mathbb{L}_\alpha$ is a finite-dimensional linear subspace of $\mathcal{H}$. From \citet[Theorem 2.4-2]{kreyszig1991introductory}, any finite-dimensional linear space is closed. So, $\mathbb{L}_\alpha$ closed. Thus, by the Direct Sum Theorem (see \citet[Theorem 3.3-4]{kreyszig1991introductory}), any $h\in\mathcal{H}$ can be decomposed as $h=h_\alpha+h_\alpha^\perp$, where  $h_\alpha\in\mathbb{L}_\alpha$,  $h_\alpha^\perp\in\mathbb{L}_\alpha^\perp$, and $\mathbb{L}_\alpha^\perp$ denotes the orthogonal complement of $\mathbb{L}_\alpha$ in $\mathcal{H}$. By the reproducing property of the RKHS, see Definition \ref{def: rkhs definition}, for any  $h\in\mathcal{H}$ we have $h(Z_i)=\left<h,K(\cdot,Z_i)\right>_\mathcal{H}$. It follows that 
	\begin{align*}
		h(Z_i)=\left<h,K(\cdot,Z_i)\right>_\mathcal{H}=\left<h_\alpha+h_\alpha^\perp,K(\cdot,Z_i)\right>_\mathcal{H}=&\left<h_\alpha,K(\cdot,Z_i)\right>_\mathcal{H}    +\left<h_\alpha^\perp,K(\cdot,Z_i)\right>_\mathcal{H}\\
		&= \left<h_\alpha,K(\cdot,Z_i)\right>_\mathcal{H}=h_\alpha(Z_i)\,.
	\end{align*}
	Also, 
	\begin{align*}
		\|h\|_\mathcal{H}^2=\|h_\alpha+h_\alpha^\perp\|_\mathcal{H}^2=\|h_\alpha\|_\mathcal{H}^2+\|h_\alpha^\perp\|_\mathcal{H}^2\geq \|h_\alpha\|_\mathcal{H}^2 \,,
	\end{align*}
	with strict inequality if $h\notin \mathbb{L}_\alpha$. From the two  displays above, the solution $(\widehat \beta, \widehat h)$ to the problem in the first line of \eqref{eq: min problem reformulation} must satisfy $\widehat h\in\mathbb{L}_\alpha$. Thus, $\widehat h$ can be expressed as in \eqref{eq: expresson of hhat and betahat} and 
	\begin{align*}
		\int_{ }\big|\operatorname{E}_n&[Y_i - X_i^T\widehat \beta-\widehat h(Z_i)]\exp(\textbf{i}(X_i,W_i)t)\Big|^2 d \mu(t)+\lambda \|\widehat h \|_\mathcal{H}^2\\
		=&\min_{\beta\in\mathbb{R}^p,h_\alpha\in\mathbb{L}_\alpha}\int_{ }\big|\operatorname{E}_n[Y_i - X_i^T \beta- h_\alpha(Z_i)]\exp(\textbf{i}(X_i,W_i)t)\Big|^2 d \mu(t)+\lambda \|h_\alpha \|_\mathcal{H}^2\, .
	\end{align*}
	Now, $\|h_\alpha\|_{\mathcal{H}}^2=\left<h_\alpha,h_\alpha\right>_\mathcal{H}=\left<\sum_i\alpha_i K(\cdot,Z_i),\sum_j\alpha_j K(\cdot,Z_j)\right>_\mathcal{H}$ $=\sum_{i,j}\alpha_i\,\alpha_j\left<K(\cdot,Z_i),K(\cdot,Z_j)\right>_\mathcal{H}$ $=\sum_{i,j}\alpha_i \alpha_j K(Z_i,Z_j)$, where in the last equality we have used the reproducing property of the space $\mathcal{H}$ by the kernel $K$. Using this and Equation \eqref{eq:def_objective_function_Mn_beta_h}, we deduce the following (see the comments below):
	\begin{align}\label{eq: min problem in alpha and beta}
		\min_{\beta\in\mathbb{R}^p,h_\alpha\in\mathbb{L}_\alpha}&\int_{ }\big|\operatorname{E}_n[Y_i - X_i^T \beta- h_\alpha(Z_i)]\exp(\textbf{i}(X_i^T,W_i)t)\Big|^2 d \mu(t)+\lambda \|h_\alpha \|_\mathcal{H}^2\nonumber\\
		=& \min_{\beta\in\mathbb{R}^p,h_\alpha\in\mathbb{L}_\alpha} \frac{1}{n^2}\sum_{1\leq i,j\leq n}[Y_i-X_i^T\beta-h_\alpha(Z_i)]\,\mathcal{F}_\mu((X_i^T,W_i)-(X_j^T,W_j))\,[Y_j-X_j^T\beta-h_\alpha(Z_j)]\nonumber\\
		&\hspace{13.8cm}+\lambda \|h_\alpha\|_\mathcal{H}^2\nonumber\\
		\overset{(a)}{=}& \min_{\beta\in\mathbb{R}^p,\bm{\alpha}\in\mathbb{R}^n} \frac{1}{n^2}[\bm{Y}-\bm{X}\beta - \bm{K}\bm{\alpha}]^T \bm{F} [\bm{Y}-\bm{X}\beta - \bm{K}\bm{\alpha}] + \lambda \bm{\alpha}^T\bm{K} \bm{\alpha}\nonumber \\
		=:& \frac{1}{n^2} \min_{\beta\in\mathbb{R}^p,\bm{\alpha}\in\mathbb{R}^n} f(\bm{\alpha},\beta)\, ,
	\end{align}
	where in equality \textit{(a)} we have used the fact that, for any $h_\alpha\in\mathbb{L}_\alpha$,  $(h_\alpha(Z_1),\ldots,h_\alpha(Z_n))^T=\bm{K}\bm{\alpha}$, with $\bm{\alpha}=(\alpha_1,\ldots,\alpha_n)^T$. Note that, from the two displays above and since $\widehat h\in \mathbb{L}_\alpha$ (as established earlier in this proof), the problem $\min_{\beta\in\mathbb{R}^p\bm{\alpha}\in\mathbb{R}^n}f(\alpha,\beta)$ admits at least one solution. Let us now derive an expression for such solutions. Since the kernel function $K$ is symmetric and positive semidefinite, the Gram matrix $\bm{K}$ is a symmetric positive semidefinite matrix; see \citet[Definition 12.6]{wainwright2019high}. $\bm{F}$ is symmetric, and from Lemma \ref{lem: positive definiteness of F} it is also positive semidefinite.
	The symmetry and positive semidefiniteness of both $\bm{K}$ and $\bm{F}$ imply that $f$ is convex. Accordingly, any solution to the first-order  conditions of the problem $\min_{\beta\in\mathbb{R}^p\bm{\alpha}\in\mathbb{R}^n}f(\alpha,\beta)$ is a global minimizer. To derive such first-order conditions, notice that 
	\begin{align}\label{eq: expression of f(alpha,beta)}
		f(\bm{\alpha},\beta)=& \bm{Y}^T\bm{F}\bm{Y} - 2 \beta^T\bm{X}^T\bm{F}\bm{Y} - 2 \bm{\alpha}^T \bm{K}\bm{F}\bm{Y} + 2 \bm{\alpha}^T \bm{K}\bm{F}\bm{X}\beta\nonumber \\
		&\hspace{5cm}+ \beta^T \bm{X}^T \bm{F}\bm{X}\beta +\bm{\alpha}^T \bm{K}\bm{F}\bm{K}\bm{\alpha}+n^2\lambda \bm{\alpha}^T \bm{K}\bm{\alpha}\, .
	\end{align}
	Thus, denoting by $(\bm{\widehat \alpha},\widehat \beta)$ a solution to the minimization problem in \eqref{eq: min problem in alpha and beta}, the first-order conditions are
	\begin{align}\label{eq: first order coniditions}
		\nabla_{\bm{\alpha}}f(\bm{\widehat \alpha},\widehat \beta)=&-2 \bm{K}\bm{F}\bm{Y} + 2 \bm{K}\bm{F}\bm{X}\widehat\beta+2 \bm{K}\bm{F}\bm{K}\bm{\widehat\alpha}+2 n^{2}\lambda \bm{K}\bm{\widehat\alpha}=0\nonumber \\
		\nabla_\beta f(\bm{\widehat\alpha},\widehat\beta)=& -2 \bm{X}^T\bm{F}\bm{Y} + 2 \bm{X}^T \bm{F}\bm{K}\bm{\widehat\alpha}+ 2 \bm{X}^T \bm{F}\bm{X}\widehat\beta=0\, .
	\end{align}
	From the second equation, we obtain 
	\begin{equation}\label{eq: beta matrix expression}
		\widehat \beta=(\bm{X}^T\bm{F}\bm{X})^{-1}\bm{X}^T\bm{F}(\bm{Y} - \bm{K}\bm{\widehat \alpha})\, .
	\end{equation}
	Note that since $\bm{X}$ is full column rank by assumption, and $\bm{F}$ is invertible by Lemma \ref{lem: positive definiteness of F}, the matrix $\bm{X}^T\bm{F}\bm{X}$ is also invertible. Plugging \eqref{eq: beta matrix expression} into the first equation of \eqref{eq: first order coniditions} and then rearranging terms, yields
	\begin{equation}\label{eq: computation of alphahat}
		[\bm{K}\bm{F}\bm{K} - \bm{K}\bm{F}\bm{X}(\bm{X}^T \bm{F}\bm{X})^{-1}\bm{X}^T \bm{F}\bm{K} + n^2\lambda \bm{K}]\bm{\widehat \alpha}=\bm{K}\bm{F}[\bm{Y} - \bm{X}(\bm{X}^T \bm{F}\bm{X})^{-1}\bm{X}^T \bm{F}\bm{Y}]\, .
	\end{equation}
	Let $\mathcal{N}(\bm{K})^\perp$ denote the orthogonal complement of $\mathcal{N}(\bm{K}):=\{\bm{\gamma}\in\mathbb{R}^n:\bm{K}\bm{\gamma}=0\}$, the null space of $\bm{K}$. We next show that there exists a unique $\bm{\alpha}\in\mathcal{N}(\bm{K})^\perp$ solving Equation \eqref{eq: computation of alphahat}. Let $\bm{\alpha}\in\mathbb{R}^n$. By the Direct Sum Theorem (\citet[Theorem 3.3-4]{kreyszig1991introductory}), there exists a unique $\bm{\alpha}_\perp\in\mathcal{N}(\bm{K})^\perp$ and a unique  $\bm{\alpha}_N\in\mathcal{N}(\bm{K})$ such that $\bm{\alpha}=\bm{\alpha}_\perp+\bm{\alpha}_N$. Since $\bm{K}\bm{\alpha}=\bm{K}\bm{\alpha}_\perp$, it follows from \eqref{eq: expression of f(alpha,beta)} that $f(\bm{\alpha},\beta)=f(\bm{\alpha}_\perp,\beta)$. Consequently,
	\begin{equation}\label{eq: min problem over N(K) perp}
		\min_{\bm{\alpha}\in\mathbb{R}^n,\beta\in\mathbb{R}^p}f(\bm{\alpha},\beta)=\min_{\bm{\alpha}\in\mathcal{N}(\bm{K})^\perp,\beta\in\mathbb{R}^p}f(\bm{\alpha},\beta)\, .
	\end{equation}
	To show that there exists a unique $\bm{\alpha}\in\mathcal{N}(\bm{K})^\perp$ solving Equation \eqref{eq: computation of alphahat}, we show that the problem on the right-hand side (RHS) of the above display admits a unique solution. Clearly, since the problem on the left-hand side of \eqref{eq: min problem over N(K) perp} admits a solution, as established earlier in this proof, the problem on the RHS also admits a solution. To show uniqueness, we establish the strict convexity  of  $f$ over $\mathcal{N}(\bm{K})^\perp\times\mathbb{R}^p$,  that is for any $(\bm{\alpha}_1,\beta_1),(\bm{\alpha}_2,\beta_2)\in \mathcal{N}(\bm{K})^\perp\times\mathbb{R}^p$ with $(\bm{\alpha}_1,\beta_1)\neq (\bm{\alpha}_2,\beta_2)$ and any $\gamma\in(0,1)$, 
	\begin{equation}\label{eq: strict convexity}
		f(\gamma \bm{\alpha}_1+(1-\gamma)\bm{\alpha}_2, \gamma \beta_1+(1-\gamma)\beta_2)<\gamma f(\bm{\alpha}_1,\beta_1)+(1-\gamma)f(\bm{\alpha}_2,\beta_2)\, .
	\end{equation}
	Since $\bm{F}$ and $\bm{K}$ are symmetric positive semidefinite matrices, they admit symmetric positive semidefinite square roots $\bm{F}^{1/2}$ and $\bm{K}^{1/2}$  satisfying $\bm{F}=\bm{F}^{1/2} \bm{F}^{1/2}$ and $\bm{K}=\bm{K}^{1/2} \bm{K}^{1/2}$. It follows that $f(\bm{\alpha},\beta)=\|\bm{F}^{1/2}(\bm{Y}-\bm{X}\beta-\bm{K}\bm{\alpha})\|^2+\lambda \|\bm{K}^{1/2}\bm{\alpha}\|^2$, and
	\begin{align*}
		f(\gamma \bm{\alpha}_1+(1-\gamma)\bm{\alpha}_2, \gamma \beta_1+(1-\gamma)\beta_2)=&\|\gamma\bm{F}^{1/2}(\bm{Y}-\bm{X}\beta_1-\bm{K}\bm{\alpha}_1) +(1-\gamma)\bm{F}^{1/2}(\bm{Y}-\bm{X}\beta_2-\bm{K}\bm{\alpha}_2) \|^2\\
		&+\lambda \|\gamma\bm{K}^{1/2}\bm{\alpha}_1 + (1-\gamma)\bm{K}^{1/2}\bm{\alpha}_2\|^2\,.
	\end{align*}
	Let us consider first the case where $\bm{\alpha}_1\neq \bm{\alpha}_2$.  Since $\bm{K}$ is injective on $\mathcal{N}(\bm{K})^\perp$, so is $\bm{K}^{1/2}$, thus $\bm{K}^{1/2}\bm{\alpha}_1\neq \bm{K}^{1/2}\bm{\alpha}_2$. By the strict convexity of the  squared norm, for any $a_1,a_2\in\mathbb{R}^n$, $\|\gamma a_1+(1-\gamma)a_2\|^2\leq \gamma\|a_1\|^2+(1-\gamma)\|a_2\|^2$,  with the inequality being strict whenever $a_1\neq a_2$. Thus, $\|\gamma\bm{K}^{1/2}\bm{\alpha}_1 + (1-\gamma)\bm{K}^{1/2}\bm{\alpha}_2\|^2< \gamma\|\bm{K}^{1/2}\bm{\alpha}_1\|^2 + (1-\gamma)\|\bm{K}^{1/2}\bm{\alpha}_2\|^2$ and $\|\gamma\bm{F}^{1/2}(\bm{Y}-\bm{X}\beta_1-\bm{K}\bm{\alpha}_1) +(1-\gamma)\bm{F}^{1/2}(\bm{Y}-\bm{X}\beta_2-\bm{K}\bm{\alpha}_2) \|^2$ $\leq \gamma \|\bm{F}^{1/2}(\bm{Y}-\bm{X}\beta_1-\bm{K}\bm{\alpha}_1) \|^2$ $+(1-\gamma)\|\bm{F}^{1/2}(\bm{Y}-\bm{X}\beta_2-\bm{K}\bm{\alpha}_2) \|^2$. Accordingly, \eqref{eq: strict convexity} holds. Let us now consider the case where $\bm{\alpha}_1=\bm{\alpha}_2$, but $\beta_1\neq \beta_2$. Since $\bm{X}$ is full rank, $\bm{X}\beta_1\neq \bm{X}\beta_2$. By Lemma \ref{lem: positive definiteness of F}, $\bm{F}$ is invertible, hence  $\bm{F}^{1/2}$ is also invertible. It follows that $\bm{F}^{1/2}(\bm{Y}-\bm{X}\beta_1-\bm{K}\bm{\alpha}_1)\neq \bm{F}^{1/2}(\bm{Y}-\bm{X}\beta_2-\bm{K}\bm{\alpha}_2)$. Thus, by the strict convexity of the squared norm, $\|\gamma\bm{F}^{1/2}(\bm{Y}-\bm{X}\beta_1-\bm{K}\bm{\alpha}_1) +(1-\gamma)\bm{F}^{1/2}(\bm{Y}-\bm{X}\beta_2-\bm{K}\bm{\alpha}_2) \|^2$ $< \gamma \|\bm{F}^{1/2}(\bm{Y}-\bm{X}\beta_1-\bm{K}\bm{\alpha}_1) \|^2$ $+(1-\gamma)\|\bm{F}^{1/2}(\bm{Y}-\bm{X}\beta_2-\bm{K}\bm{\alpha}_2) \|^2$ and $\|\gamma\bm{K}^{1/2}\bm{\alpha}_1 + (1-\gamma)\bm{K}^{1/2}\bm{\alpha}_2\|^2= \gamma\|\bm{K}^{1/2}\bm{\alpha}_1\|^2 + (1-\gamma)\|\bm{K}^{1/2}\bm{\alpha}_2\|^2$. Consequently, \eqref{eq: strict convexity} holds. We have therefore established that \eqref{eq: strict convexity} holds for any $(\bm{\alpha}_1,\beta_1)\neq (\bm{\alpha}_2,\beta_2)$ and $\gamma\in(0,1)$. \\
	Since $f$ is strictly convex over $\mathcal{N}(\bm{K})^\perp\times\mathbb{R}^p$, the problem $\min_{\bm{\alpha}\in\mathcal{N}(\bm{K})^\perp,\beta\in\mathbb{R}^p}f(\bm{\alpha},\beta)$ admits a unique solution. We next use this to conclude that (i) Equation \eqref{eq: computation of alphahat} admits a unique solution over $\mathcal{N}(\bm{K})^\perp$, and (ii) this solution has the smallest norm among all the solutions of \eqref{eq: computation of alphahat} over $\mathbb{R}^n$. Let $\bm{\widehat \alpha}$ be a solution to \eqref{eq: computation of alphahat}, and let us define
	\begin{equation*}
		\widehat \beta_{\bm{\alpha}}=(\bm{X}^T\bm{F}\bm{X})^{-1}\bm{X}^T\bm{F}(\bm{Y} - \bm{K}\bm{\alpha})\, .
	\end{equation*}
	By the Direct Sum Theorem, there exists a unique $\bm{\widehat \alpha}_\perp\in\mathcal{N}(\bm{K})^\perp$ and a unique $\bm{\widehat \alpha}_N\in\mathcal{N}(\bm{K})$ such that $\bm{\widehat \alpha}=\bm{\widehat \alpha}_\perp+\bm{\widehat \alpha}_N$. Since $\bm{K}\bm{\widehat \alpha}=\bm{K}\bm{\widehat \alpha}_\perp$, $(\bm{\widehat \alpha}_\perp,\widehat \beta_{\bm{\widehat \alpha}_\perp})$ satisfies  \eqref{eq: computation of alphahat} and \eqref{eq: beta matrix expression}, and hence the first-order conditions in \eqref{eq: first order coniditions}. Thus, $(\bm{\widehat \alpha}_\perp,\widehat \beta_{\bm{\widehat \alpha}_\perp})$ is a global minimizer of $f$ and also solves the problem on the RHS of \eqref{eq: min problem over N(K) perp}. Since the solution to such a problem is unique, we find that there exists a unique element of $\mathcal{N}(\bm{K})^\perp$ satisfying \eqref{eq: computation of alphahat}. Moreover, denoting such a unique element by  $\bm{ \widehat \alpha}_\perp$, any solution $\bm{\widetilde \alpha}$ to \eqref{eq: computation of alphahat} can be written in the form $\bm{\widetilde \alpha}=\bm{\widehat \alpha}_\perp+\bm{\widetilde \alpha}_N$,  with $\bm{\widetilde \alpha}_N\in\mathcal{N}(\bm{K})$ and $\|\bm{\widetilde \alpha}\|^2=\|\bm{\widehat \alpha}_\perp\|^2+\|\bm{\widetilde \alpha}_N\|^2$ $\geq \|\bm{\widehat \alpha}_\perp\|^2$. Thus,  $\bm{\widehat \alpha}_\perp$ is also the minimum norm solution to \eqref{eq: computation of alphahat}. This concludes the proof. \\
	
	We finally note that all solutions to \eqref{eq: computation of alphahat} yield the same function $\widehat h$ in \eqref{eq: expresson of hhat and betahat}. In fact, if $\bm{\alpha}_1=(\alpha_{1,1},\ldots,\alpha_{1,n})^T$ and $\bm{\alpha}_2=(\alpha_{2,1},\ldots,\alpha_{2,n})^T$ are two solutions to \eqref{eq: computation of alphahat}, they can be written as $\bm{\alpha}_1=\bm{\widehat \alpha}_\perp + \bm{\alpha}_{1 N}$ and  $\bm{\alpha}_2=\bm{\widehat \alpha}_\perp + \bm{\alpha}_{2 N}$, with $\bm{\widehat \alpha}_\perp$ denoting the unique solution to \eqref{eq: computation of alphahat} over $\mathcal{N}(\bm{K})^\perp$ and $\bm{\alpha}_{1 N}, \bm{\alpha}_{2 N}\in \mathcal{N}(\bm{K})$. Hence,  $\bm{\alpha}_1-\bm{\alpha}_2=\bm{\alpha}_{1 N} - \bm{\alpha}_{2 N}$ and $\|\sum_{i=1}^n \alpha_{1,i} K(\cdot,Z_i) - \sum_{i=1}^n \alpha_{2,i} K(\cdot,Z_i)\|^2_{\mathcal{H}}$ $=(\bm{\alpha}_1-\bm{\alpha}_2)^T\bm{K} (\bm{\alpha}_1-\bm{\alpha}_2)$ $=(\bm{\alpha}_{1 N}-\bm{\alpha}_{2 N})^T\bm{K} (\bm{\alpha}_{1 N}-\bm{\alpha}_{2 N})$ $=0$. Thus, $\widehat h$ in \eqref{eq: expresson of hhat and betahat} does not depend on which solution of \eqref{eq: computation of alphahat} is used. 
	
	\begin{flushright}
		$\blacksquare$
	\end{flushright}
	
	\noindent \textbf{Proof of Proposition \ref{prop: ifr}}. We start with the following decomposition:
	\begin{align}\label{eq: decomposition of the empirical process}
		\sqrt{n}(\widehat \theta - \theta_0)=&\sqrt{n}\,\{\operatorname{E}_n[\widehat h'(Z_i)]-\operatorname{E}[h_0'(Z_i)]\,\}\nonumber\\
		=&\sqrt{n}(\E_n - \E)[\widehat h'(Z_i) - h_0'(Z_i)]\nonumber\\
		&+\sqrt{n}(\E_n-\E)[h'_0(Z_i)]\nonumber\\
		&+ \sqrt{n}\E [\widehat h'(Z_i) - h_0'(Z_i)]\, , 
	\end{align}
	where $\operatorname{E}[\widehat h(Z_i)]:=\int_{ }\widehat h(z) P(dz)$. 
	From Assumption \ref{as: belonging condition}(i)(ii) and Lemma \ref{lem: ase}\ref{lem: ase:: sample}, it follows that 
	\begin{equation*}
		\sqrt{n}(\E_n - \E)[\widehat h'(Z_i) - h_0'(Z_i)]=o_P(1)\, .  
	\end{equation*}
	Thus, to prove the desired result, we need to obtain an influence function representation for the third term on the RHS of \eqref{eq: decomposition of the empirical process}. We have (see the comments below):
	\begin{align}\label{eq: integration by parts}
		\sqrt{n}\E [\widehat h'(Z) - h_0'(Z)]=&\int_{\mathcal{Z}} [\widehat h'(z) - h_0'(z)] f_Z(z) dz\nonumber\\
		\overset{(a)}{=}& \left\{\lim_{z\rightarrow \infty}[\widehat h(z) - h_0(z) ] f_Z(z) - \lim_{z\rightarrow -\infty}[\widehat h(z) - h_0(z) ] f_Z(z) \right\} \nonumber\\
		&- \sqrt{n}\int_{ \mathcal{Z}} [\widehat h(z) - h_0(z)] f_Z'(z) dz \nonumber\\
		\overset{(b)}{=}& -\sqrt{n}\int_{\mathcal{Z}} [\widehat h(z) - h_0(z)] f_Z'(z) dz\, .
	\end{align}
	Equality $(a)$ follows from an integration by parts, while  equality $(b)$ is a direct consequence of Assumption \ref{as: belonging condition}(iii). Let us now deal with the RHS of the above display. By construction, $\widehat h \in\mathcal{H}$, and by Assumption \ref{as: rkhs}(ii), $h_0, f'_Z\in\mathcal{H}$. Thus, by Lemma \ref{lem: RKHS results}\ref{lem: RKHS results: H eig space}\ref{lem: RKHS results: equality of H and H eig space}, there exists square-summable sequences $(\widehat \gamma_j)_j\,,\,(\gamma_j)_j\,,\,(\pi_j)_j$ such that
	\begin{align}\label{eq: expression of hhat h0 and f' in terms of Heig}
		\widehat h(z)=\sum_j \widehat \gamma_j \widetilde \phi_j(z)\,,\, h_0(z)=\sum_j\gamma_j \widetilde \phi_j(z)\,,\text{ and}\, f_Z'(z)=\sum_j \pi_j \widetilde{\phi}_j(z)\,,
	\end{align}
	where $(\widetilde \phi_j)_j$ is a sequence of functions defined in Lemma \ref{lem: RKHS results}. 
	We therefore have (see the comments below):
	\begin{align}\label{eq: from L2 inner product to ell 2 inner product}
		\sqrt{n}\int_{\mathcal{Z}} [\widehat h(z) - h_0(z)] f_Z'(z) dz=& \sqrt{n} \int_{\mathcal{Z}} \sum_j (\widehat \gamma_j - \gamma_j) \widetilde \phi_j (z) \sum_j \pi_j \widetilde{\phi}_j (z) dz \nonumber \\
		\overset{(a)}{=}&  \sqrt{n} \lim_{J\rightarrow \infty} \int_{\mathcal{Z}} \sum_{j=1}^J (\widehat \gamma_j - \gamma_j) \widetilde{\phi}_j(z) \sum_{k=1}^J \pi_k \widetilde{\phi}_k(z) dz \nonumber \\
		\overset{(b)}{=}& \sqrt{n} \lim_{J\rightarrow \infty} \sum_{j=1}^J (\widehat \gamma_j - \gamma_j) \pi_j \int_{\mathcal{Z}} |\widetilde \phi_j(z)|^2 dz \nonumber \\
		\overset{(c)}{=}& \sqrt{n} \lim_{J\rightarrow \infty} \sum_{j=1}^J (\widehat \gamma_j - \gamma_j) \pi_j \eta_j\nonumber  \\
		\overset{(d)}{=}& \sqrt{n} \left<(\widehat \gamma_j - \gamma_j)_j , (\pi_j\eta_j)_j\right>_{\ell^2}\,.
	\end{align}
	Here $\ell^2$ denotes the space of square summable sequences, and for any $(a_j)_j,(b_j)_j\in\ell^2$, $\left\langle(a_j)_j,(b_j)_j\right\rangle_{\ell^2}=\sum_{j=1}^\infty a_j b_j$. 
	To obtain equality $(a)$, note that by the Cauchy-Schwarz inequality, $|\sum_{j=1}^J \pi_j \widetilde{\phi}_j(z)|^2$ $\leq \sum_{j=1}^J \pi_j^2$ $ \sum_{j=1}^J \widetilde{\phi}_j(z)^2$ $\leq \sum_{j=1}^\infty \pi_j^2$ $ \sum_{j=1}^\infty \widetilde{\phi}_j(z)^2$ , where $\sum_j \pi_j^2=\|f_Z'\|_{\mathcal{H}}^2$ by \eqref{eq: expression of hhat h0 and f' in terms of Heig} and Lemma \ref{lem: RKHS results}\ref{lem: RKHS results: H eig space}\ref{lem: RKHS results: equality of H and H eig space}. Similarly, $|\sum_{j=1}^J (\widehat \gamma_j - \gamma_j) \widetilde{\phi}_j(z)|^2\leq $ $\sum_{j=1}^\infty (\widehat \gamma_j - \gamma_j)^2 \sum_{j=1}^\infty \widetilde \phi_j(z)^2$, where $\sum_{j=1}^\infty (\widehat \gamma_j - \gamma_j)^2=\|\widehat h - h_0\|_{\mathcal{H}}^2$. From Lemma \ref{lem: RKHS results}\ref{lem: RKHS results: upper bound for l2 sum}, $\sum_{j=1}^\infty \widetilde{\phi}^2_j(z)\leq K(z,z)$. Gathering the above result, we have $|\sum_{j=1}^J (\widehat \gamma_j - \gamma_j)\widetilde \phi_j(z)$ $ \sum_{j=1}^J \pi_j \widetilde \phi_j(z)|\leq \|\widehat h - h_0\|_{\mathcal{H}} \|f_Z'\|_{\mathcal{H}} K(z,z)$, where $\int_{\mathcal{Z}} K(z,z) dz$ $ <\infty$ by the continuity of $K$ and the boundedness of $\mathcal{Z}$, see Assumptions \ref{as: iid and T}(iii) and \ref{as: rkhs}(i). Thus, we can apply the Lebesgue Dominated Convergence Theorem to get  equality $(a)$ of the above display. To show equality $(b)$, note that, from Lemma \ref{lem: RKHS results}\ref{lem: RKHS results: svd of K}\ref{lem: RKHS results: H eig space},   $\widetilde \phi_j=\sqrt{\eta_j} \phi_j$, where $(\eta_j)_j$ is a sequence of positive scalars and $\int_{\mathcal{Z}} \phi_j(z) \phi_k(z)=0$ for any $j\neq k$.
	Thus, equality $(b)$ follows. Equality $(c)$ follows from $\widetilde \phi_j=\sqrt{\eta_j} \phi_j$ and $\int_{\mathcal{Z}}|\phi_j(z)|^2 dz=1$, see Lemma \ref{lem: RKHS results}\ref{lem: RKHS results: svd of K}\ref{lem: RKHS results: H eig space}. 
	Finally, equality $(d)$ follows from the definition of the inner product on $\ell^2$. \\
	Next, let us show that $(\pi_j \eta_j)_j$ are the coefficients of $\int_{\mathcal{Z}} K(\cdot,z) f'_Z(z) dz$ in the space $\mathcal{H}_{EIG}$ from Lemma \ref{lem: RKHS results}\ref{lem: RKHS results: H eig space}. That is, $\int_{\mathcal{Z}} K(\cdot,z) f'_Z(z) dz=\sum_j \pi_j \eta_j \widetilde \phi_j(\cdot)$. First, by Mercer's Theorem (see \citet[Theorem 12.20]{wainwright2019high}),  $K(z,u)=\sum_j \eta_j \phi_j(z)\phi_j(u)=\sum_j \widetilde \phi_j(z) \widetilde \phi_j(u)$. Since, as noted earlier, $f_Z'(z)=\sum_j \pi_j \widetilde \phi_j(z)$, we obtain (see the comment below):
	\begin{align*}
		\int_{\mathcal{Z}} K(z,u) f_Z'(u) du =& \int_{\mathcal{Z}} \sum_j \widetilde \phi_j(z) \widetilde \phi_j(u) \sum_k \pi_k \widetilde \phi_k(u) du\\
		\overset{(a)}{=}& \lim_{J\rightarrow \infty} \int_{\mathcal{Z}} \sum_{j=1}^J \widetilde \phi_j(z) \widetilde \phi_j(u) \sum_{k=1}^J \pi_k \widetilde \phi_k(u) du\\
		\overset{(b)}{=}& \sum_j \widetilde \phi_j(z) \pi_j \int_{\mathcal{Z}} |\widetilde \phi_j(u)|^2 du \\
		\overset{(c)}{=}& \sum_j \widetilde \phi_j(z) \pi_j \eta_j\, .
	\end{align*}
	Equality $(a)$ follows from combining the bounds $|\sum_{j=1}^J \widetilde{\phi}_j(u) \widetilde{\phi}_j(z)|^2\leq \sum_j \widetilde{\phi}_j(u)^2 \sum_j \widetilde{\phi}_j(z)^2 \leq K(z,z) K(u,u)$, $|\sum_{k=1}^J \pi_k \widetilde{\phi}_k(u)|^2\leq \|f_Z'\|_{\mathcal{H}}^2 K(u,u)$, together with $\int_{\mathcal{Z}} K(u,u) du<\infty$ and the Lebesgue Dominated Convergence Theorem, as argued earlier in this proof. Equality $(b)$ follows from the orthogonality of $(\widetilde \phi_j)_j$, that is, $\int_{\mathcal{Z}}\widetilde \phi_j(z) \widetilde \phi_k(z) dz=0$ for $j\neq k$, as noted earlier. Similarly, equality $(c)$ follows from $\int_{\mathcal{Z}}|\widetilde \phi_j(z)|^2dz=\eta_j$, as noted earlier.  
	So, by Lemma \ref{lem: RKHS results}\ref{lem: RKHS results: H eig space},  $\int_{\mathcal{Z}} K(\cdot,z) f_Z'(z) dz\in\mathcal{H}_{EIG}$ as long as the coefficients $(\pi_j \eta_j)_j$ are square summable. To show this, note first that $\sum_j (\pi_j \eta_j)^2\leq [\sup_j \eta_j]^2 \sum_j \pi_j^2 =[\sup_j \eta_j]^2 \|f_Z'\|_{\mathcal{H}}^2$, where we have used  $\sum_j \pi_j^2=\|f_Z'\|_{\mathcal{H}}^2$, as  established earlier in this proof. Next, let us denote by  $L^2(\mathcal{Z})$ the space of square integrable functions defined on $\mathcal{Z}$. Then, the operator $\varphi\mapsto \int_{\mathcal{Z}} K(\cdot,z) \varphi(z) dz$ from $L^2(\mathcal{Z})$ to $L^2(\mathcal{Z})$ is self-adjoint and positive. 
	To see that it is self-adjoint, note that for any $\varphi,\psi\in L^2(\mathcal{Z})$,  
	\begin{align*}
		\left<\int_{\mathcal{Z}} K(\cdot,u)\varphi(u) du,\psi\right>_{L^2(\mathcal{Z})}
		& =
		\int_{\mathcal{Z}} \left[\int_{\mathcal{Z}} K(z,u)\varphi(u) du\right] \overline{\psi(z)} dz \\
		& =
		\int_{\mathcal{Z}} \varphi(u) \overline{\int_{\mathcal{Z}} K(u,z) \psi(z) }dz du \\
		& =
		\left< \varphi, \int_{\mathcal{Z}} K(\cdot,z)\psi(z) dz \right>_{L^2(\mathcal{Z})}.
	\end{align*}
	Thus the operator $\varphi\mapsto \int_{\mathcal{Z}} K(\cdot,z)\varphi(z)dz$ is self-adjoint. To see that such an operator is also positive, note that, since $K(z,u)=\sum_j \widetilde \phi_j(z) \widetilde\phi_j(u)$, as established earlier in this proof,  $\left<\int_{\mathcal{Z}} K(\cdot,u)\varphi(u) du,\varphi\right>_{L^2(\mathcal{Z})}=$ $\int_{\mathcal{Z}} [\int_{\mathcal{Z}} K(z,u)\varphi(u) du] \overline{\varphi(z)} dz=$ $\sum_j \int_{\mathcal{Z}} \widetilde \phi_j(u) \varphi(u) du \overline{\int_{\mathcal{Z}} \widetilde \phi_j(z) \varphi(z) dz }\geq 0 $. Thus, the operator $\varphi\mapsto \int_{\mathcal{Z}} K(\cdot,z)\varphi(z)dz$ is also positive. It follows that (see the comments below):
	\begin{align*}
		\sup_{\varphi\in L^2(\mathcal{Z}),\|\varphi\|_{L^2(\mathcal{Z})}=1 }\Big\|\int_{\mathcal{Z}} K(\cdot,u) \varphi(u) du\Big \|_{L^2(\mathcal{Z})}
		\overset{(a)}{=} & \sup_{\varphi\in L^2(\mathcal{Z}), \|\varphi\|_{L^2(\mathcal{Z})}=1} \left<\int_{\mathcal{Z}} K(\cdot,u) \varphi(u) du, \varphi\right>_{L^2(\mathcal{Z})} \\
		\overset{(b)}{\geq} &\sup_j \left<\int_{\mathcal{Z}} K(\cdot,u)\phi_j(u) du,  \phi_j\right>_{L^2(\mathcal{Z})}\\
		\overset{(c)}{=}&\sup_j \eta_j \left<\phi_j, \phi_j \right>_{L^2(\mathcal{Z})} \\
		\overset{(d)}{=}&\sup_j \eta_j  \, .
	\end{align*}
	Equality $(a)$ is a direct consequence of \citet[Theorem 15.9]{kress1999linear} and the fact the operator $\varphi\mapsto \int_{\mathcal{Z}} K(\cdot,z)\varphi(z)dz$ is positive and self adjoint. Inequality $(b)$ follows from $\int_{\mathcal{Z}}|\phi_j(z)|^2dz=1$, as noted earlier. Equality $(c)$ is a direct consequence of $\int_{\mathcal{Z}}K(\cdot,u)\phi_j(u)du=\eta_j \phi(\cdot)$ from Lemma \ref{lem: RKHS results}\ref{lem: RKHS results: svd of K}. Finally, equality $(d)$ follows from $\left<\phi_j,\phi_j\right>_{L^2(\mathcal{Z})}=\int_{\mathcal{Z}}|\phi_j(z)|^2 ds=1$. 
	Gathering these results shows that $\sum_j (\pi_j \eta_j)^2 < \infty$. Hence, $\int_{\mathcal{Z}} K(\cdot,z) f_Z'(z) dz=\sum_j \pi_j \eta_j \widetilde{\phi}_j(\cdot)\in \mathcal{H}_{EIG}$. \\
	Now, since $\int_{\mathcal{Z}} K(\cdot,z) f_Z'(z) dz=\sum_j \pi_j \eta_j \widetilde{\phi}_j(\cdot)\in \mathcal{H}_{EIG}$, using Lemma \ref{lem: RKHS results}\ref{lem: RKHS results: H eig space}\ref{lem: RKHS results: equality of H and H eig space}, we obtain
	\begin{align}\label{eq: from l2 inner product to H inner product}
		\sqrt{n}\left<(\widehat \gamma_j - \gamma_j), (\pi_j\eta_j)_j\right>_{\ell^2}=& \sqrt{n} \left<\widehat h - h_0, \int_{\mathcal{Z}} K(\cdot,z) f_Z'(z) dz \right>_{EIG}\nonumber\\
		=& \sqrt{n} \left< \widehat h - h_0, \int_{\mathcal{Z}} K(\cdot,z) f_Z'(z) dz \right>_{\mathcal{H}}\, .
	\end{align}
	To derive the influence function representation of the above inner product,  we first reformulate  $\widehat h$ in terms of the operators $\widehat{\A}$ and $\widehat{\B}$, and the estimator $\widehat s$ from Equation \eqref{eq: Ahat and shat}.
	As already noted in Equation \eqref{eq: hhat Tikhonov expression} (see the proof of Proposition \ref{prop:computation_solution_empirical_program}), $\widehat h$ can be expressed as  
	\begin{equation*}
		\widehat h=(\widehat{\T}^*\widehat{\T}+\lambda\I)^{-1}\widehat{\T}^*\widehat r\,,
	\end{equation*}
	where $\widehat \T=(\I-\widehat \P)\widehat \A$, $\widehat r=(\I-\widehat{\P})\widehat{s}$, and $\widehat \P=\widehat \B(\widehat \B^*\widehat \B)^{-1}\widehat \B^*$. Let 
	\begin{equation*}
		g:=\int_{\mathcal{Z}}K(\cdot,z) f'_Z(z) dz\, .
	\end{equation*}
	From \cite{babii2025unobservables},
	\begin{align}\label{eq: decomposition of hhat - h0}
		\widehat h - h_0=\sum_{j=1}^4\Theta_j + h_\lambda - h_0\,,
	\end{align}
	where 
	\begin{align*}
		\Theta_1:=&(\T^*\T+\lambda\I)^{-1}\T^*(\widehat r - \widehat \T h_0)\\
		\Theta_2:=& (\T^*\T+\lambda\I)^{-1}(\widehat \T^*-\T^*)(\widehat r - \widehat \T h_0)\\
		\Theta_3:=&[(\widehat \T^*\widehat \T+\lambda\I)^{-1}-(\T^*\T+\lambda\I)^{-1}]\widehat \T^*(\widehat r - \widehat \T h_0)\\
		\Theta_4:=& (\widehat \T^*\widehat \T+\lambda\I)^{-1} \widehat \T^*\widehat \T h_0 - h_\lambda\\
		h_\lambda:=& (\widehat \T^*\widehat \T+\lambda\I)^{-1} \T^*\T h_0\, .
	\end{align*}
	We now use the above decomposition to obtain an expansion of the RHS of \eqref{eq: from l2 inner product to H inner product}. 
	Let us show that $\sqrt{n}\left<\Theta_j,g\right>_{\mathcal{H}}=o_P(1)$ for $j=2,3,4$. From Lemma \ref{lem: bounds for decomposition of hhat-h0 for asynorm}\ref{lem: bounds for decomposition of hhat-h0 for asynorm:: Theta 2}, 
	\begin{align*}
		\sqrt{n}\left<\Theta_2,g\right>_{\mathcal{H}}=O_P\Big(\sqrt{n}\|\widehat \T-\T\|_{op}\,\|\widehat r - \widehat \T h_0\|_\mu\,\|(\lambda\I+\T^*\T)^{-1}g\|_\mathcal{H}\Big)\, .  
	\end{align*}
	Using Lemma \ref{lem: rates and compacness}\ref{lem: rates and compacness:: That rate}\ref{lem: rates and compacness:: (I - Phat) shat - That h rate}, we have  $\|\widehat \T - \T\|_{op}=O_P(n^{-1/2})$ and $\|\widehat r - \widehat \T h_0\|_\mu=O_P(n^{-1/2})$. Since $g\in\mathcal{R}(\T^*\T)$ by Assumption \ref{as: source conditions}, and $\T$ is a compact operator by Lemma \ref{lem: rates and compacness}\ref{lem: rates and compacness:: compactness of A and T}, we can apply 
	Lemma \ref{lem: Bounds on Compact Operator}\ref{lem: Bounds on Compact Operator: iv} with $\gamma=2$ to obtain $\|(\lambda\I+\T^*\T)^{-1}g\|_\mathcal{H}=O(1)$. Gathering these results shows that the RHS of the above display is $o_P(1)$. Thus, 
	\begin{equation}\label{eq: negligibility of Theta2 g}
		\sqrt{n}\left<\Theta_2,g\right>_{\mathcal{H}}=o_P(1)\, .
	\end{equation}
	Let us now show the negligibility of $\sqrt{n}\left<\Theta_3,g\right>_\mathcal{H}$. From Lemma \ref{lem: bounds for decomposition of hhat-h0 for asynorm}\ref{lem: bounds for decomposition of hhat-h0 for asynorm:: Theta 3}, 
	\begin{align*}
		\sqrt{n}\left<\Theta_3,g\right>_\mathcal{H}= O_P\Big( \sqrt{n}\|\widehat r - \widehat \T h_0\|_\mu\cdot\|\widehat \T-\T\|_{op}\cdot&\\
		&\Big[ \|(\widehat \T^* \widehat \T + \lambda \I)^{-1}\widehat \T^*\|_{op}\,\|\T( \T^*  \T + \lambda \I)^{-1}g\|_\mu\\
		&+ \|\widehat \T(\widehat \T^* \widehat \T + \lambda \I)^{-1}\widehat \T^*\|_{op}\,\|( \T^*  \T + \lambda \I)^{-1}g\|_\mathcal{H}\Big]\Big)\, .  
	\end{align*}
	As already noted above, $\|\widehat r - \widehat \T h_0\|_\mu=O_P(n^{-1/2})$ and $\|\widehat \T - \T\|_{op}=O_P(n^{-1/2})$. Since $\T$ and $\widehat \T$ are compact operators by Lemma \ref{lem: rates and compacness}\ref{lem: rates and compacness:: compactness of A and T}\ref{lem: rates and compacness:: compactness of That}, and $h_0,g\in\mathcal{R}(\T^*\T)$ by Assumption \ref{as: source conditions},  we can apply Lemma \ref{lem: Bounds on Compact Operator}\ref{lem: Bounds on Compact Operator: i}\ref{lem: Bounds on Compact Operator: iii}\ref{lem: Bounds on Compact Operator: iv}\ref{lem: Bounds on Compact Operator: v} with $\gamma=2$ to obtain $\|\widehat \T(\widehat \T^* \widehat \T + \lambda \I)^{-1}\widehat \T^*\|_{op}=O_P(1)$, $\|(\widehat \T^* \widehat \T + \lambda \I)^{-1}\widehat \T^*\|_{op}=O_P(\lambda^{-1/2})$, $\|( \T^*  \T + \lambda \I)^{-1}g\|_\mathcal{H}=O_P(1)$, and $\|\T( \T^*  \T + \lambda \I)^{-1}g\|_\mu=O(1)$. Gathering these results and using the fact that $n \lambda \rightarrow \infty$, we obtain 
	\begin{align}\label{eq: negligibility of Theta3 g}
		\sqrt{n}\left<\Theta_3,g\right>_\mathcal{H}= O_P(n^{-1/2}[\lambda^{-1/2}+1])=o_P(1)\, .
	\end{align}
	Let us now show the negligibility of $\sqrt{n}\left<\Theta_4,g\right>_\mathcal{H}$. From Assumption \ref{as: source conditions}, we have $g\in\mathcal{R}(\T^*\T)$. Hence, we can apply Lemma \ref{lem: bounds for decomposition of hhat-h0 for asynorm}\ref{lem: bounds for decomposition of hhat-h0 for asynorm:: Theta 4} to obtain 
	\begin{align}\label{eq: first bound for Theta4 g}
		\sqrt{n}\left<\Theta_4, g\right>_\mathcal{H}=O_P\Big(&\sqrt{n}\|(\widehat \T^* \widehat \T + \lambda \I)^{-1}\widehat \T^* \widehat \T h_0-h_\lambda\|_\mathcal{H}\,\|\widehat \T - \T\|_{op}\nonumber\\
		&+\sqrt{n}\|\widehat \T [(\widehat \T^* \widehat \T + \lambda \I)^{-1}\widehat \T^* \widehat \T h_0-h_\lambda]_\mu\|\Big)\,,
	\end{align}
	where
	\begin{align}\label{eq: 1st term of bound of Theta4 g}
		\|(\widehat \T^* \widehat \T + \lambda \I)^{-1}\widehat \T^* \widehat \T h_0 - h_\lambda\|_\mathcal{H}\leq \lambda \cdot\|\widehat \T-\T\|_{op}\cdot &\|(\T^*\T+\lambda\I)^{-1}h_0\|_\mathcal{H}\cdot\nonumber\\
		&\Big[\|(\widehat \T^* \widehat \T + \lambda \I)^{-1}\widehat \T^*\|_{op}+\|(\widehat \T^* \widehat \T + \lambda \I)^{-1}\|_{op} \|\T\|_{op}\Big]
	\end{align}
	and 
	\begin{align}\label{eq: 2nd term of bound of Theta4 g}
		\|\widehat \T [(\widehat \T^* \widehat \T + \lambda \I)^{-1}\widehat \T^* \widehat \T h_0 - h_\lambda] \|_\mu\leq \lambda &\cdot \|\widehat \T-\T\|_{op}\cdot \|(\T^*\T+\lambda\I)^{-1}h_0\|_{\mathcal{H}}\nonumber \\
		&\cdot\Big[\|\widehat \T(\widehat \T^* \widehat \T + \lambda \I)^{-1}\widehat \T^*\|_{op}+\|\T\|_{op}\|\widehat \T(\widehat \T^* \widehat \T + \lambda \I)^{-1}\|_{op}\Big]\, .
	\end{align}
	Now, as already noted earlier in this proof, $\|\widehat \T - \T\|_{op}=O_P(n^{-1/2})$, $\|(\widehat \T^*\widehat \T + \lambda \I)^{-1}\widehat \T^*\|_{op}=O_P(\lambda^{-1/2})$, and $\|\widehat \T(\widehat \T^*\widehat \T + \lambda \I)^{-1}\widehat \T^*\|_{op}=O_P(1)$. Also, since  the norm of an operator equals the norm of its adjoint, $\|\widehat \T(\widehat \T^*\widehat \T+\lambda \I)^{-1}\|_{op}=\|[\widehat \T (\widehat \T^*\widehat \T+\lambda \I)^{-1}]^*\|_{op}$ $=\| (\widehat \T^* \widehat \T + \lambda \I)^{-1}\widehat \T^*\|_{op}=O_P(\lambda^{-1/2})$. Since $\T$ is a compact operator and compact operators are bounded (see \citet[Theorem 2.14]{kress1999linear}), $\|\T\|_{op}<\infty$. Furthermore, because  $h_0\in\mathcal{R}(\T^*\T)$ by Assumption \ref{as: source conditions}, we can apply Lemma \ref{lem: Bounds on Compact Operator}\ref{lem: Bounds on Compact Operator: iv}  with $\gamma=2$ to obtain $\|(\T^*\T+\lambda \I)^{-1}h_0\|_\mathcal{H}=O(1)$. Also, from Lemma    \ref{lem: Bounds on Compact Operator}\ref{lem: Bounds on Compact Operator: ii}, we have $\|(\widehat \T^*\widehat\T+\lambda \I)^{-1}\|_{op}=O(\lambda^{-1})$. From these results, together with \eqref{eq: 1st term of bound of Theta4 g} and \eqref{eq: 2nd term of bound of Theta4 g}, it follows that
	\begin{gather*}
		\|(\widehat \T^*\widehat \T + \lambda \I)^{-1}\widehat \T^*\widehat \T h_0 - h_\lambda\|_\mathcal{H}=O_P\Big(\lambda n^{-1/2}[\lambda^{-1/2}+\lambda^{-1}] \Big)\,\nonumber\\
		\text{ and }\|\widehat \T[(\widehat \T^*\widehat \T+\lambda \I)^{-1}\widehat \T^*\widehat \T h_0 - h_\lambda]\|_\mu=O_P\Big(\lambda n^{-1/2}[1+\lambda^{-1/2}]\Big)\, .
	\end{gather*}
	The above display, together with $\|\widehat \T-\T\|_{op}=O_P(n^{-1/2})$ and \eqref{eq: first bound for Theta4 g}, yields
	\begin{align}\label{eq: negligibility of Theta4 g}
		\sqrt{n}\left<\Theta_4,g\right>_\mathcal{H}=O_P\Big(&\sqrt{n} \lambda n^{-1/2}[\lambda^{-1/2}+\lambda^{-1}] n^{-1/2}\nonumber\\
		&+ \sqrt{n}\lambda n^{-1/2}[1+\lambda^{-1/2}]\Big)=o_P(1)\, .
	\end{align}
	Let us finally show the negligibility of $\sqrt{n}\left<h_\lambda-h_0,g\right>_\mathcal{H}$. From Assumption \ref{as: source conditions}, we have $h_0\in\mathcal{R}(\T^*\T)$, and by definition $h_\lambda=(\T^*\T+\lambda\I)^{-1}\T^*\T h_0$.  Hence, applying Lemma \ref{lem: Bounds on Compact Operator}\ref{lem: Bounds on Compact Operator: vi} with $\gamma=2$ gives
	\begin{equation}\label{eq: negligibility of hlambda - h0}
		\sqrt{n}\left<h_\lambda-h_0,g\right>_\mathcal{H}=O(\sqrt{n}\lambda)=O(\sqrt{n \lambda^2})=o(1)\,.
	\end{equation}
	From the decomposition in \eqref{eq: decomposition of hhat - h0} and Equations \eqref{eq: negligibility of Theta2 g}, \eqref{eq: negligibility of Theta3 g}, \eqref{eq: negligibility of Theta4 g}, and \eqref{eq: negligibility of hlambda - h0}, we obtain 
	\begin{align}\label{eq: expansion for sqrt n hhat - h0 g}
		\sqrt{n}\left<\widehat h - h_0,g\right>_{\mathcal{H}}=&\sqrt{n}\left<(\T^*\T+\lambda\I)^{-1}\T^*(\widehat r - \widehat \T h_0),g\right>_{\mathcal{H}}+o_P(1)\nonumber\\
		=&\sqrt{n}\left<(\widehat r - \widehat \T h_0),\T(\T^*\T+\lambda\I)^{-1}g\right>_{\mu}+o_P(1)\nonumber\\
		=& \sqrt{n}\left<(\widehat r - \widehat \T h_0),\T(\T^*\T+\lambda\I)^{-1}\T^*\T f_0\right>_{\mu}+o_P(1)\,,
	\end{align}
	where the last equality follows from $g=(\T^*\T)f_0$ for some $f_0\in\mathcal{H}$. This holds by Assumption \ref{as: source conditions}, since $g\in\mathcal{R}(\T^*\T)$. To obtain the final result, it remains to derive an influence function representation for the leading term of the above display.\\
	Now, by Lemma \ref{lem: rates and compacness}\ref{lem: rates and compacness:: compactness of A and T}\ref{lem: rates and compacness:: injectivity of T}, $\T$ is compact and injective. Hence, from \citet[Definition 15.5 and Theorem 15.23]{kress1999linear},  we have $\|(\T^*\T+\lambda \I)^{-1}\T^*\T f_0 - f_0\|_{\mathcal{H}}=o(1)$ as $\lambda\rightarrow 0$. Also, from Lemma \ref{lem: rates and compacness}\ref{lem: rates and compacness:: (I - Phat) shat - That h rate} and the fact that  $\widehat r=(\I-\widehat \P)\widehat s$, we have $\sqrt{n}\|\widehat r - \widehat \T h_0\|_\mu=O_P(1)$. Combining these results with the Cauchy-Schwarz inequality gives
	\begin{align*}
		\sqrt{n}&\left<(\widehat r - \widehat \T h_0),\T[(\T^*\T+\lambda\I)^{-1}(\T^*\T)f_0-  f_0]\right>_{\mu}\\
		&=O_P(\sqrt{n}\|\widehat r - \widehat \T h_0\|_\mu\,\|\T\|_{op}\|(\T^*\T+\lambda\I)^{-1}\T^*\T f_0-  f_0\|_\mathcal{H})=o_P(1)\, .
	\end{align*}
	Accordingly, in the leading term on the RHS of \eqref{eq: expansion for sqrt n hhat - h0 g}, we may replace  $\T(\T^*\T+\lambda\I)^{-1}(\T^*\T)f_0$ with  $\T f_0$ at the cost of an $o_P(1)$ reminder. Using this and recalling that $\widehat r=(\I-\widehat \P)\widehat s$ and $\widehat \T=(\I-\widehat \P)\widehat \A$, we obtain 
	\begin{align}\label{eq: expansion for sqrt n hhat - h0 g :: ii}
		\sqrt{n}\left<(\widehat r - \widehat \T h_0),\T(\T^*\T+\lambda\I)^{-1}(\T^*\T)f_0\right>_{\mu}=& \sqrt{n}\left<(\widehat r - \widehat \T h_0),\T f_0\right>_{\mu} + o_P(1)\nonumber\\
		=& \sqrt{n}\left< (\I-\widehat \P)(\widehat s - \widehat \A h_0) , \T f_0\right>_{\mu}+o_P(1)\nonumber\\
		=& \sqrt{n}\left< (\I-\P)(\widehat s - \widehat \A h_0) , \T f_0\right>_{\mu}\nonumber\\
		&- \sqrt{n}\left< (\widehat \P - \P)(\widehat s - \widehat \A h_0) , \T f_0\right>_{\mu}\nonumber\\
		&+ o_P(1)\, .
	\end{align}
	Let us now derive an influence function representation for the first term on the RHS of \eqref{eq: expansion for sqrt n hhat - h0 g :: ii}. To this end, using $\widehat \A h_0=\operatorname{E}_n\{ h_0(Z_i) \exp(\textbf{i}(X_i^T,W_i)\cdot)\}$ and $\widehat s =\operatorname{E}_n\{ Y_i \exp(\textbf{i}(X_i^T,W_i)\cdot)\}$, we obtain 
	\begin{align}\label{eq: (I-P)(shat - Ahat h0)}
		\sqrt{n}\left<(\I-\P)(\widehat s - \widehat \A h_0), \T f_0\right>_\mu=&\sqrt{n}\left<(\I-\P)\operatorname{E}_n [Y_i-h_0(Z_i)]\exp(\textbf{i}(X_i^T,W_i)\cdot), \T f_0\right>_\mu\nonumber\\
		=& \sqrt{n}\operatorname{E}_n [Y_i-h_0(Z_i)] \left<(\I-\P) \exp(\textbf{i}(X_i^T,W_i)\cdot), \T f_0\right>_\mu\, .
	\end{align}
	Let us prove that the RHS has zero expectation. Since $\P$ is the projection operator onto the range of $\B$, we have $(\I-\P)\B \beta=0$ for any $\beta\in\mathbb{R}^p$. Applying $(\I-\P)$ to both sides of  $s=\B\beta_0+\A h_0$ gives $(\I-\P)(s-\A h_0) = 0$. Recalling that $\A h_0=\operatorname{E} \{h_0(Z_i) \exp(\textbf{i}(X_i^T,W_i)\cdot)\}$ and $s =\operatorname{E} \{Y_i \exp(\textbf{i}(X_i^T,W_i)\cdot)\}$, we obtain  
	\begin{align}\label{eq: zero xpectation for (I-P)E}
		(\I-\P) \operatorname{E}\{[Y_i- h_0(Z_i)] \exp(\textbf{i}(X_i^T,W_i)\cdot)\}=0\, .
	\end{align}
	To show that the RHS of \eqref{eq: (I-P)(shat - Ahat h0)} has zero expectation, let us first prove that we can interchange the operators $\E$ and $(\I-\P)$ in the above display. By Lemma \ref{lem: rates and compacness}\ref{lem: rates and compacness:: expression of P and Phat}, $\P= \B( \B^* \B)^{-1} \B^*$. Since $ \B\beta=\operatorname{E}\{X_i^T\beta \exp(\textbf{i}(X_i^T,W_i)\cdot)\}\in L^2_\mu$, for any $\beta\in\mathbb{R}^p$ and $\psi\in L^2_\mu$ we can compute the adjoint $\B^*:L^2_\mu\to \mathbb{R}^p$ as follows (see the comments below):
	\begin{align*}
		\left<\B \beta,\psi\right>_\mu=&\int_{ }\operatorname{E}\{X_i^T\beta \exp(\textbf{i}(X_i^T,W_i)t)\} \overline{\psi(t)} d\mu(t)\\
		=& \int_{ } \int_{ }\beta^T X_1 \exp(\textbf{i}(X_1^T,W_1)t) \overline{\psi(t)} d P(X_1,W_1) d\mu(t)\\
		=& \beta^T\overline{\int_{ }X_1\int_{ }\exp(-\textbf{i}(X_1^T,W_1)t) \;\psi(t)d\mu(t) d P(X_1,W_1)}\\
		=& \beta^T \B^*\psi\,.
	\end{align*}
	In the third line of the above display, to exchange $\operatorname{E}$ and $\int_{ }d\mu(t)$, we have used Tonelli-Fubini's Theorem. This applies because $\|\beta\|\,\|X_1\| \,|\psi(t)|$ bounds  the absolute value of the integrand and $\int_{ }\int_{ }\|X_1\|\, |\psi(t)|dP(X_1)d\mu(t)<\infty$, where $\|\cdot\|$ denotes the Euclidean norm on $\mathbb{R}^p$. The last line of the above display follows from the definition of the adjoint operator. Thus (see the comments below), 
	\begin{align*}
		\B^* \operatorname{E}\{[Y_i- h_0(Z_i)]&\exp(\textbf{i}(X_i^T,W_i)\cdot)\}\\
		\overset{(a)}{=}&\int_{ }X_1\int_{ }\exp(-\textbf{i}(X_1^T,W_1)t)\operatorname{E}\{[Y_i - h_0(Z_i)]\exp(\textbf{i}(X_i^T,W_i)t)\} d\mu(t) d P(X_1,W_1)\\
		\overset{(b)}{=}& \operatorname{E}\left\{[Y_i - h_0(Z_i)]\int_{ }X_1\int_{ }\exp(-\textbf{i}(X_1^T,W_1)t)\;\exp(\textbf{i}(X_i^T,W_i)t) d \mu(t) d P(X_1,W_1)\right\}\\
		\overset{(c)}{=}& \operatorname{E}\{[Y_i-h_0(Z_i)]\B^* \exp(\textbf{i}(X_i^T,W_i)\cdot)\}\, .
	\end{align*}
	Equality $(a)$ follows from the form of $\B^*$. Equality $(b)$ follows from Tonelli-Fubini's Theorem, which allows exchanging the integrals. This applies because the absolute value of the integrand is bounded by   $|Y_i-h_0(Z_i)|\cdot \|X_1\|$, and $\int_{ }\operatorname{E}\{|Y_i-h_0(Z_i)|\} \cdot \|X_1\| d P(X_1)<\infty$. Finally, equality $(c)$ follows again from the definition of $\B^*$. Next, since $\B^*\B$ maps $\mathbb{R}^p$ to $\mathbb{R}^p$, it is a $p\times p$ matrix (see \citet[pages 112-113]{kreyszig1991introductory}). Therefore, its inverse $(\B^*\B)^{-1}$  is also  $p\times p$ matrix. Consequently, from the previous display, we obtain
	\begin{align*}
		(\B^*\B)^{-1}\B^* \operatorname{E}\{[Y_i- h_0(Z_i)]\exp(\textbf{i}(X_i^T,W_i)\cdot)\}=& \operatorname{E}\{[Y_i-h_0(Z_i)](\B^*\B)^{-1}\B^* \exp(\textbf{i}(X_i^T,W_i)\cdot)\}\, .
	\end{align*}
	By Lemma \ref{lem: rates and compacness}\ref{lem: rates and compacness:: expression of P and Phat}, $\P=\B(\B^*\B)^{-1}\B^*$.
	Hence, applying $\B$ to both sides of the above equation yields (see the comments below)
	\begin{align}\label{eq: P applied to expectation}
		\P \operatorname{E}\{[Y_i- h_0(Z_i)]\exp(\textbf{i}(X_i^T,W_i)\cdot)\}=& \B\operatorname{E}\{[Y_i-h_0(Z_i)]\underbrace{(\B^*\B)^{-1}\B^* \exp(\textbf{i}(X_i^T,W_i)\cdot)}_{=: \delta(X_i,W_i)\in\mathbb{R}^p}\}\nonumber\\
		\overset{(a)}{=}& \int_{ }X_1^T \operatorname{E}\{[Y_i-h_0(Z_i)]\delta(X_i,W_i)\} \exp(\textbf{i}(X_1^T,W_1)\cdot) dP(X_1,W_1)\nonumber\\
		\overset{(b)}{=}& \operatorname{E}\left\{[Y_i-h_0(Z_i)]\int_{ }X_1^T\delta(X_i,W_i)\,\exp(\textbf{i}(X_1^T,W_1)\cdot) dP(X_1,W_1)\right\}\nonumber\\
		=& \operatorname{E}\left\{[Y_i-h_0(Z_i)]\B\delta(X_i,W_i)\right\}\nonumber\\
		=& \operatorname{E}\left\{[Y_i-h_0(Z_i)]\B(\B^*\B)^{-1}\B^*\exp(\textbf{i}(X_i^T,W_i)\cdot)\right\}\nonumber\\
		=& \operatorname{E}\left\{[Y_i-h_0(Z_i)]\P\exp(\textbf{i}(X_i^T,W_i)\cdot)\right\}\, .
	\end{align}
	Equality $(a)$ follows from the definition of $\B$. Equality $(b)$ follows from Tonelli-Fubini's Theorem. This applies because the integrand is bounded in absolute value by $|Y_i-h_0(Z_i)|\cdot |X_1^T\delta(X_i,W_i)|$
	$\leq |Y_i-h_0(Z_i)|\cdot\|X_1\|\cdot \|\delta(X_i,W_i)\|$ $\leq |Y_i-h_0(Z_i)|\cdot\|X_1\|\cdot \|(\B^*\B)^{-1}\|\cdot \|\B^* \exp(\textbf{i}(X_i^T,W_i)\cdot)\|$, with $\|\B^* \exp(\textbf{i}(X_i^T,W_i)\cdot)\|\leq \|\B\|_{op} $ and $\int_{ }\operatorname{E}\{|Y_i-h_0(Z_i)|\} \|X_1\| d P(X_1)<\infty$, where, with a slight abuse of notation, $\|\cdot\|$ denotes both the Euclidean and the Frobenius norm.  The remaining equalities follow from the definitions of $\B$, $\delta(X_i,W_i)$, and $\P=\B(\B^*\B)^{-1}\B^*$. \\
	Equation \eqref{eq: P applied to expectation} shows that we can interchange $\operatorname{E}$ and $\P$, so that
	\begin{equation*}
		(\I-\P) \operatorname{E}\{[Y_i- h_0(Z_i)]\exp(\textbf{i}(X_i^T,W_i)\cdot)\}=\operatorname{E}\{[Y_i- h_0(Z_i)]\;(\I-\P)\exp(\textbf{i}(X_i^T,W_i)\cdot)\}\, .
	\end{equation*}
	From the above display and \eqref{eq: zero xpectation for (I-P)E}, we obtain 
	\begin{equation*}
		\operatorname{E}\{[Y_i- h_0(Z_i)]\;(\I-\P)\exp(\textbf{i}(X_i^T,W_i)\cdot)\}=0\, .
	\end{equation*}
	Thus, by exchanging the integrals as done earlier, we obtain 
	\begin{align}\label{eq: null expectation}
		\operatorname{E}&\left\{[Y_i-h_0(Z_i)]\left<(\I-\P)\exp(\textbf{i}(X_i^T,W_i)\cdot),\T f_0\right>_{\mu}\right\}\nonumber\\
		&=\operatorname{E}\left\{\int_{ }[Y_i-h_0(Z_i)]\;[(\I-\P)\exp(\textbf{i}(X_i^T,W_i)\cdot)](t)\;\overline{(\T f_0)(t)} d\mu(t)\right\}\nonumber\\
		&= \int_{ }\operatorname{E}\left\{[Y_i-h_0(Z_i)]\;[(\I-\P)\exp(\textbf{i}(X_i^T,W_i)\cdot)](t)\right\}\overline{(\T f_0)(t)} d\mu(t)=0\, .
	\end{align}
	Accordingly,  \eqref{eq: (I-P)(shat - Ahat h0)} provides an influence function representation for the first term on the RHS of \eqref{eq: expansion for sqrt n hhat - h0 g :: ii}.  \\
	We now derive an influence function representation for the second term on the RHS of \eqref{eq: expansion for sqrt n hhat - h0 g :: ii}. By Lemma \ref{lem: rates and compacness}\ref{lem: rates and compacness:: expression of P and Phat}, $\widehat \P=\widehat \B(\widehat \B^*\widehat \B)^{-1}\widehat \B^*$ and $\P=\B(\B^*\B)^{-1}\B^*$. Also, by Lemma \ref{lem: rates and compacness}\ref{lem: rates and compacness:: Bhat rate}\ref{lem: rates and compacness:: Bhatstar Bhat inverse rate} $\|\widehat \B - \B\|_{op}=O_P(n^{-1/2})$,  $\|\widehat \B^*\widehat \B-\B^* \B\|_{op}=O_P(n^{-1/2})$, and $\|(\widehat \B^*\widehat \B)^{-1}- ( \B^* \B)^{-1}\|_{op}=O_P(n^{-1/2})$. Thus, we obtain
	\begin{align}\label{eq: decomposition of Phat - P}
		\widehat \P - \P=&   (\widehat \B - \B)(\widehat \B^*\widehat \B)^{-1}\widehat \B^*\nonumber\\
		&+\B(\widehat \B^*\widehat \B)^{-1}[ \B^*\B- \widehat \B^*\widehat \B ](\B^*\B)^{-1}\widehat \B^*\nonumber\\
		&+ \B(\B^*\B)^{-1}(\widehat \B^*-\B^*)\nonumber\\
		=& (\widehat \B - \B)(\B^* \B)^{-1} \B^*\nonumber\\
		&+ \B(\B^* \B)^{-1}[ \B^*\B- \widehat \B^*\widehat \B ](\B^*\B)^{-1} \B^*\nonumber\\
		&+ \B(\B^*\B)^{-1}(\widehat \B^*-\B^*)+o_P(n^{-1/2})\,,
	\end{align}
	where the remainder term is $o_P(n^{-1/2})$ in the sense that its operator norm is $o_P(n^{-1/2})$. From the above display, and since $\|\widehat \B - \B\|_{op}=O_P(n^{-1/2})$ and  $\|\widehat \B^*\widehat \B-\B^* \B\|_{op}=O_P(n^{-1/2})$,  as noted earlier, we obtain $\|\widehat \P - \P\|_{op}=O_P(n^{-1/2})$. Also, from Lemma \ref{lem: rates and compacness}\ref{lem: rates and compacness:: Ahat rate}\ref{lem: rates and compacness:: shat rate}, $\|\widehat\A h_0 -  \A h_0\|_\mu=O_P(n^{-1/2})$ and $\|\widehat s - s\|_\mu=O_P(n^{-1/2})$. Using these rates and the above display, for the second term on the RHS of \eqref{eq: expansion for sqrt n hhat - h0 g :: ii}, we have 
	\begin{align*}
		\sqrt{n}\left<(\widehat \P - \P)(\widehat s-\widehat \A h_0),\T h_0\right>_\mu=&\sqrt{n}\left<(\widehat \P - \P)( s- \A h_0),\T h_0\right>_\mu+o_P(1)\\
		=& \sqrt{n}\left< (\widehat \B - \B)( \B^* \B)^{-1} \B^*(s-\A h_0), \T h_0 \right>_\mu\\
		&+ \sqrt{n}\left< \underbrace{\B( \B^* \B)^{-1}[ \B^*\B- \widehat \B^*\widehat \B ](\B^*\B)^{-1} \B^*(s - \A h_0)}_{=: a_1} , \T h_0 \right>_\mu\\
		&+ \sqrt{n}\left< \underbrace{\B(\B^*\B)^{-1}(\widehat \B^*-\B^*)(s - \A h_0)}_{=:a_2} , \T h_0 \right>_\mu\\
		&+ o_P(1)\, .
	\end{align*}
	Now, let us recall that $\T=(\I-\P)\A$, where $(\I-\P)$ is the projection operator onto $\mathcal{R}(\B)^\perp$, the orthogonal complement of the range of $\B$. Then, $\T f_0=(\I-\P)\A f_0\in\mathcal{R}(\B)^\perp$. Since $a_1$ and $a_2$, defined above, belong to $\mathcal{R}(\B)$,  the last two leading terms on the RHS of the above display are exactly zero. By recalling that $\widehat \B \beta=\operatorname{E}_n[X_i^T\beta \exp(\textbf{i}(X_i^T,W_i)\cdot)]$ and $\B \beta=\operatorname{E}[X_i^T\beta \exp(\textbf{i}(X_i^T,W_i)\cdot)]$, the first leading term can be expressed as  
	\begin{align}\label{eq: expansion of leading term of Phat - P in innner product}
		\sqrt{n}\left<(\widehat \B - \B)\underbrace{(\B^*\B)^{-1}\B^*(s-\A h_0)}_{\in\mathbb{R}^p},\T h_0\right>_\mu=\sqrt{n}(\operatorname{E}_n - \operatorname{E})\Big\{[X_i^T &(\B^*\B)^{-1}\B^*(s-\A h_0) ] \nonumber    \\
		& \left<\exp(\textbf{i}(X_i^T,W_i)\cdot), \T f_0\right>_\mu\big\}\, .
	\end{align}
	
	Finally, combining the two equations above with \eqref{eq: expansion for sqrt n hhat - h0 g :: ii}, \eqref{eq: (I-P)(shat - Ahat h0)}, and \eqref{eq: null expectation}, we obtain the influence function representation for the leading term on the RHS of \eqref{eq: expansion for sqrt n hhat - h0 g}. This completes the proof of the desired result.\\
	
	$(ii)$ The result is a direct consequence of the influence function representation in Part $(i)$. 
	
	\begin{flushright}
		$\blacksquare$
	\end{flushright}
	
	\noindent \textbf{Proof of Proposition \ref{prop: bootstrap validity}}. \\
	
	$(i)$ In what follows, we consider statements relative to joint probability space of the bootstrap weights  and the sample data. We start with the following decomposition:
	\begin{align}\label{eq: decomposition for bootstrap i}
		\sqrt{n}(\widehat \theta_b - \widehat \theta)=&\sqrt{n}\left\{ \operatorname{E}_n \frac{\xi_i}{\overline{\xi}}\widehat h'_b(Z_i)  - \operatorname{E}_n[\widehat h'(Z_i)] \right\} \nonumber  \\
		=& \sqrt{n}\operatorname{E}_n \frac{\xi_i}{\overline{\xi}} [\widehat h _b'(Z_i) - \widehat h'(Z_i)]  \nonumber \\
		&+ \sqrt{n}\operatorname{E}_n  \left\{\frac{\xi_i}{\overline{\xi}}\widehat h'(Z_i) - \widehat h'(Z_i)\right\} \, .
	\end{align}
	Let $\overline{\widehat h '(Z)}:=\operatorname{E}_n h '(Z_i)$. Then, for the second term on the RHS of \eqref{eq: decomposition for bootstrap i},  we have
	\begin{align}\label{eq: 2nd term  of decomposition for bootstrap i}
		\sqrt{n}\operatorname{E}_n \left\{ \frac{\xi_i}{\overline{\xi}}\widehat h'(Z_i) - \widehat h'(Z_i) \right\}= & \sqrt{n}\operatorname{E}_n \frac{\xi_i}{\overline{\xi}}\Big[\widehat h'(Z_i)-\overline{\widehat h '(Z)}\Big]\nonumber\\
		=& \sqrt{n}\operatorname{E}_n \frac{(\xi_i-1)}{\overline{\xi}}\Big[\widehat h'(Z_i)-\overline{\widehat h '(Z)}\Big]\nonumber\\
		=& \sqrt{n}\operatorname{E}_n \frac{(\xi_i-1)}{\overline{\xi}} [\widehat h '(Z_i) - h_0'(Z_i)]\nonumber\\
		&+ \sqrt{n}\operatorname{E}_n \frac{(\xi_i-1)}{\overline{\xi}} \Big[h_0'(Z_i)-\overline{h_0'(Z)}\Big]\nonumber\\
		&- \sqrt{n}\operatorname{E}_n \frac{(\xi_i-1)}{\overline{\xi}}\overline{[\widehat h'(Z)-h_0'(Z)]}\, .
	\end{align}
	We now consider the first term on the RHS of the above display. By Assumption \ref{as: bootstrap weights}, the bootstrap weights $\{\xi_i\}_i$ are independent from the sample data and satisfy $\operatorname{E}\xi_i=1$. Hence, $\operatorname{E} (\xi_i-1) [\widehat h '(Z_i) - h_0'(Z_i)]=\int_{ }(\xi-1)[\widehat h'(z)-h_0'(z)] d P_Z(z) d P_\xi(\xi)=0$, where $P_Z$ and $P_\xi$ denote the probability distributions of $Z$ and $\xi$, respectively. It follows that the empirical process $\sqrt{n}\operatorname{E}_n (\xi_i-1) [\widehat h '(Z_i) - h_0'(Z_i)]$   is centered. By Assumption \ref{as: belonging condition}(i)(ii), $\sup_{z\in\mathcal{Z}}|\widehat h'(z)-h_0'(z)|=o_P(1)$ and, with probability approaching one, $\widehat h'\in\mathcal{G}$, with $\log N_{[\,]}(\epsilon,\mathcal{G},L^2(P))\leq C \epsilon^{-v}$ and $v\in(0,2)$.    Accordingly, by Lemma \ref{lem: ase}\ref{lem: ase:: bootstrap},  $\sqrt{n}\operatorname{E}_n (\xi_i-1) [\widehat h '(Z_i) - h_0'(Z_i)]=o_P(1)$. Since $\overline{\xi}=1+o_P(1)$, we obtain
	\begin{align*}
		\sqrt{n}\operatorname{E}_n\frac{(\xi_i-1)}{\overline{\xi}} [\widehat h '(Z_i) - h_0'(Z_i)]=o_P(1)\, .
	\end{align*}
	Next, for the second term on the RHS of \eqref{eq: 2nd term  of decomposition for bootstrap i}, we have 
	\begin{align*}
		\sqrt{n}\operatorname{E}_n \frac{(\xi_i-1)}{\overline{\xi}} \Big[h_0'(Z_i)-\overline{h_0'(Z)}\Big]=& \sqrt{n}\operatorname{E}_n \frac{\xi_i}{\overline{\xi}} \Big[h_0'(Z_i)-\overline{h_0'(Z)}\Big] \\
		=& \sqrt{n}\operatorname{E}_n \left(\frac{\xi_i}{\overline{\xi}}-1\right) h_0'(Z_i)\, .
	\end{align*}
	Let us now consider the last term on the RHS of \eqref{eq: 2nd term  of decomposition for bootstrap i}. Using $\sup_{z\in\mathcal{Z}}|\widehat h'(z)-h_0(z)|=o_P(1)$ from Assumption \ref{as: belonging condition}(i), we obtain  $\overline{\widehat h(Z)-h_0(Z)}=o_P(1)$. This, together with $\sqrt{n}\operatorname{E}_n (\xi_i-1)=O_P(1)$ and $\overline{\xi}=1+o_P(1)$, implies
	\begin{equation*}
		\sqrt{n}\operatorname{E}_n \frac{(\xi_i-1)}{\overline{\xi}}\overline{[\widehat h'(Z)-h_0'(Z)]}=o_p(1)\, .
	\end{equation*}
	Gathering the above results, we find that 
	\begin{equation}\label{eq: expansion of 2nd term  of decomposition for bootstrap i}
		\sqrt{n}\operatorname{E}_n \left\{ \frac{\xi_i}{\overline{\xi}}\widehat h'(Z_i) - \widehat h'(Z_i) \right\}=\sqrt{n}\operatorname{E}_n \left(\frac{\xi_i}{\overline{\xi}}-1\right) h_0'(Z_i)+o_P(1)\,
	\end{equation}
	which provides an expansion for the second term on the RHS of \eqref{eq: decomposition for bootstrap i}. \\
	To prove the desired result, we now obtain an expansion for the first term on the RHS of \eqref{eq: decomposition for bootstrap i}. We start by decomposing it as follows: 
	\begin{align}\label{eq: decomoposition to get rid of overline xi}
		\sqrt{n}\operatorname{E}_n \frac{\xi_i}{\overline{\xi}} [\widehat h _b'(Z_i) - \widehat h'(Z_i)] =&\sqrt{n}\operatorname{E}_n \xi_i [\widehat h _b'(Z_i) - \widehat h'(Z_i)] \nonumber\\
		&+ \sqrt{n}\operatorname{E}_n \xi_i\Big(\frac{1}{\overline{\xi}}-1\Big) [\widehat h _b'(Z_i) - \widehat h'(Z_i)] \, .
	\end{align}
	For the second term on the RHS of the above display, we have (see the comments below):
	\begin{align*}
		\sqrt{n}\operatorname{E}_n \xi_i\Big(\frac{1}{\overline{\xi}}-1\Big) [\widehat h _b'(Z_i) - \widehat h'(Z_i)] =&\frac{\overline{\xi[\widehat h'_b(Z)-\widehat h'(Z)]}}{\overline{\xi}}\sqrt{n}(1-\overline{\xi})\\
		&= o_P(1)\,.
	\end{align*}
	To obtain the second equality, note that Assumptions \ref{as: belonging condition}(i) and \ref{as: belonging condition for bootstrap}(i) imply $\sup_{z\in\mathcal{Z}}|\widehat h'_b(z)-\widehat h'(z)|=o_P(1)$, so that $\overline{\xi[\widehat h'_b(Z)-\widehat h'(Z)]}=o_P(1)$. Moreover, as established earlier in this proof, $\sqrt{n}(\overline{\xi}-1)=O_P(1)$. Hence, the last equality in the above display follows. \\
	Next, we focus on the first term on the RHS of \eqref{eq: decomoposition to get rid of overline xi}. We have
	\begin{align}\label{eq: numerator}
		\sqrt{n}\operatorname{E}_n \xi_i[\widehat h_b'(Z_i) - \widehat h'(Z_i)]  =& \sqrt{n}(\operatorname{E}_n-\operatorname{E})\; \xi_i[\widehat h_b'(Z_i) - \widehat h'(Z_i)]\nonumber\\
		&+ \sqrt{n}\operatorname{E}\; \xi_i[\widehat h_b'(Z_i) - \widehat h'(Z_i)]\, ,
	\end{align}
	where $\operatorname{E}\xi_i[\widehat h_b'(Z_i)-\widehat h'(Z_i)]=\int_{ }\xi[\widehat h'_b(z)-\widehat h'(z)] d P_Z(z) d P_\xi(\xi)$. 
	Let us show the negligibility of the first term on the RHS of \eqref{eq: numerator}. By Assumptions \ref{as: belonging condition} and \ref{as: belonging condition for bootstrap}, $\sup_{z\in\mathcal{Z}}|\widehat h'_b(z)-\widehat h'(z)|=o_P(1)$ and, with probability approaching one, $\widehat h'_b\,,\,\widehat h'\in\mathcal{G}$, where $\log N_{[\,]}(\epsilon,\mathcal{G},L^2(P))\leq C \epsilon^{-\upsilon}$ and $\upsilon\in(0,2)$. This implies that, with probability approaching one, $\widehat h'_b-\widehat h'\in\widetilde{\mathcal{G}}:=\{f_1-f_2:f_1,f_2\in\mathcal{G}\}$ with $\log N_{[\,]}(\epsilon,\widetilde{\mathcal{G}},L^2(P))\leq C_2 \epsilon^{-\upsilon}$, for a constant $C_2$.  Consequently, Lemma \ref{lem: ase}\ref{lem: ase:: bootstrap} yields 
	\begin{equation*}
		\sqrt{n}(\operatorname{E}_n-\operatorname{E})\; \xi_i[\widehat h_b'(Z_i) - \widehat h'(Z_i)]=o_P(1)\, .
	\end{equation*}
	Next, consider the second term on the RHS of \eqref{eq: numerator}. Since the bootstrap weights $\{\xi_i:i=1,\ldots,n\}$ are independent from the sample data and satisfy $\operatorname{E}\xi_i=1$, we have 
	\begin{align}\label{eq: decomposition of expectation for boostrap}
		\sqrt{n}\operatorname{E}\xi_i [\widehat h_b'(Z_i)-\widehat h'(Z_i)]=& \sqrt{n}\operatorname{E} [\widehat h_b'(Z_i)-\widehat h'(Z_i)]\nonumber\\
		&= \sqrt{n}\operatorname{E} [\widehat h_b'(Z_i)- h'_0(Z_i)]\nonumber\\
		&- \sqrt{n}\operatorname{E} [\widehat h'(Z_i)- h'_0(Z_i)]\, .
	\end{align}
	The expansion of the second term on the RHS of this equation was already obtained in the proof of Proposition \ref{prop: ifr}. For  clarity, we report it below:
	\begin{align}\label{eq: expansion of expectation from asymptotic ifr}
		\sqrt{n}\operatorname{E} [\widehat h'(Z_i)- h'_0(Z_i)]=&  -\sqrt{n}\operatorname{E}_n[Y_i-h_0(Z_i)]\left<(\I-\P)\exp(\textbf{i}(X_i^T,W_i)\cdot),\T f_0\right>_\mu \nonumber\\
		&+ \sqrt{n}(\operatorname{E}_n - \operatorname{E})[X_i^T (\B^*\B)^{-1}\B^*(s-\A h_0) ]     \left<\exp(\textbf{i}(X_i^T,W_i)\cdot), \T f_0\right>_\mu\nonumber\\
		&+ o_P(1)\,. 
	\end{align}
	Let us obtain an expansion for the first term on the RHS  of \eqref{eq: decomposition of expectation for boostrap}. By Assumption \ref{as: belonging condition for bootstrap}(iii) and using the same arguments as in Equations \eqref{eq: integration by parts}, \eqref{eq: from L2 inner product to ell 2 inner product},  and \eqref{eq: from l2 inner product to H inner product}, we obtain 
	\begin{equation}\label{eq: from expectation to inner product in bootstrap world}
		\sqrt{n}\operatorname{E} [\widehat h_b'(Z_i)-h_0'(Z_i)]=-\sqrt{n}\left<\widehat h_b - h_0, g\right>_\mathcal{H}\, ,
	\end{equation}
	where $g=\int_{\mathcal{Z}}K(\cdot,z) f_Z'(z) dz$. 
	Next, define  the bootstrap counterparts of $\widehat \A$, $\widehat \B$, and $\widehat s$ as follows:
	\begin{align}\label{eq: definition of Ahat b, Bhat b, ans shat b}
		\widehat{\A}_bh:=\E_n\!\left\{h(Z_i) \frac{\xi_i}{\overline{\xi}}\exp(\textbf{i}\,(X^T_i,W_i)\, \cdot)\right\}\,,&\, \widehat{\A}_b:\mathcal{H}\to L^2_\mu\,,\nonumber \\
		\widehat{\B}_b\beta=\operatorname{E}_n \left\{X_i^T\beta \frac{\xi_i}{\overline{\xi}} \exp(\textbf{i}\,(X_i^T,W_i)\,\cdot)\right\}\,,&\; \widehat{\operatorname{\B}}_b:\mathbb{R}^p\to L^2_\mu\,,\nonumber \nonumber \\
		\widehat s_b:=\E_n \left\{Y_i \frac{\xi_i}{\overline{\xi}}\exp(\textbf{i}\,(X^T_i,W_i)\, \cdot)\right\}\,,&\, \widehat s_b\in L^2_\mu\, .
	\end{align}
	Then, the problem in \eqref{eq: bootstrap version of hhat} is equivalent to 
	\begin{equation*}
		(\widehat \beta_b,\widehat h _b)=\arg \min_{\beta\in\mathbb{R}^p,h\in\mathcal{H}} \|\widehat s_b - \widehat{\operatorname{B}}_b\beta - \widehat{\operatorname{A}}_b h\|^2_\mu + \lambda \|h\|_\mathcal{H}^2\, .
	\end{equation*}
	Let us denote by $\widehat \P_b$ the projection operator onto $\mathcal{R}(\widehat \B _b)$, the range of $\widehat \B_b$. By defining $\widehat \T_b:=(\I - \widehat \P _b)\widehat \A_b$, we can use Lemma \ref{lem: compactness, and expressions of h0, hhat}\ref{lem: compactness, and expressions of h0, hhat:: hhatb and betahatb} to express $\widehat h_b$ as 
	\begin{equation*}
		\widehat h_b=(\widehat \T^*_b\widehat \T_b + \lambda \I )^{-1}\widehat\T^*_b(\operatorname{I} - \widehat{\operatorname{P}}_b)\widehat s_b \, .
	\end{equation*}
	Then, to obtain an expansion of $\sqrt{n}\left<\widehat h_b - h_0, g\right>_\mathcal{H}$, we apply arguments analogous to those used  in the proof of Proposition~\ref{prop: ifr}, from Equation~\eqref{eq: decomposition of hhat - h0} up to Equation \eqref{eq: (I-P)(shat - Ahat h0)}, and from Equation \eqref{eq: decomposition of Phat - P} up to Equation \eqref{eq: expansion of leading term of Phat - P in innner product}. Specifically, we substitute $(\widehat{\T},\widehat{\B}, \widehat{s})$ with $(\widehat{\T}_b, \widehat{\B}_b, \widehat{s}_b)$ and replace Lemma~\ref{lem: rates and compacness} with Lemma~\ref{lem: rates and compacness for bootstrap}. This yields
	\begin{align}\label{eq: expansion of inner product in bootstrap world}
		\sqrt{n}\left<\widehat h_b - h_0, g\right>_\mathcal{H}=&\sqrt{n}\operatorname{E}_n [Y_i-h_0(Z_i)] \frac{\xi_i}{\overline{\xi}} \left<(\I-\P) \exp(\textbf{i}(X_i^T,W_i)\cdot), \T f_0\right>_\mu\nonumber\\
		&- \sqrt{n}\operatorname{E}_n\Big\{\frac{\xi_i}{\overline{\xi}}[X_i^T (\B^*\B)^{-1}\B^*(s-\A h_0) ]   \nonumber \\
		&\hspace{4cm} \left<\exp(\textbf{i}(X_i^T,W_i)\cdot), \T f_0\right>_\mu\big\}\nonumber\\
		&+\sqrt{n}\operatorname{E}\Big\{[X_i^T (\B^*\B)^{-1}\B^*(s-\A h_0) ]   \nonumber \\
		&\hspace{4cm} \left<\exp(\textbf{i}(X_i^T,W_i)\cdot), \T f_0\right>_\mu\Big\}\, .
	\end{align}
	Plugging Equations  \eqref{eq: expansion of expectation from asymptotic ifr}, \eqref{eq: from expectation to inner product in bootstrap world}, and \eqref{eq: expansion of inner product in bootstrap world} into \eqref{eq: decomposition of expectation for boostrap} gives
	\begin{align*}
		\sqrt{n}\operatorname{E}\xi_i [\widehat h_b'(Z_i)-\widehat h'(Z_i)]=& - \sqrt{n}\operatorname{E}_n [Y_i-h_0(Z_i)] \Big(\frac{\xi_i}{\overline{\xi}}-1\Big) \left<(\I-\P) \exp(\textbf{i}(X_i^T,W_i)\cdot), \T f_0\right>_\mu\\
		&+ \sqrt{n}\operatorname{E}_n\Big\{\Big(\frac{\xi_i}{\overline{\xi}}-1\Big)[X_i^T (\B^*\B)^{-1}\B^*(s-A h_0) ]   \nonumber \\
		&\hspace{4cm}\cdot \left<\exp(\textbf{i}(X_i^T,W_i)\cdot), \T f_0\right>_\mu\Big\}\\
		&+ o_P(1)\, .
	\end{align*}
	Gathering the above results, we obtain an expansion for the first term on the RHS of \eqref{eq: decomposition for bootstrap i}. This concludes the proof of $(i)$. \\
	
	$(ii)$ 
	By Lemma \ref{lem: conditional multiplier clt for Bayesian bootsrap}, we obtain that, conditionally on the sample data $\{(Y_i,Z_i,X_i,W_i)\}_{i\geq 1}$, $\sqrt{n}\operatorname{E}_n[(\xi_i/\overline{\xi}-1)\psi(Y_i,Z_i,X_i,W_i)]\leadsto T$, where $T\sim\mathcal{N}(0,\operatorname{Var}[\psi(Y_i,Z_i,X_i,W_i)]\,)$, for almost all realizations $\{(Y_i,Z_i,X_i,W_i)\}_{i\geq 1}$. That is, 
	\begin{equation*}
		\sup_x|{\Pr}_\xi(\sqrt{n}\operatorname{E}_n[(\xi_i/\overline{\xi}-1)\psi(Y_i,Z_i,X_i,W_i)] \leq x)-F_{T}(x)|=o(1)\text{ almost surely},    
	\end{equation*}
	where $\Pr_{\xi}$ is the probability distribution treating the bootstrap weights $\{\xi_i:i=1,2,\ldots\}$ as random  and the sample data as fixed, and $F_T$ is the distribution function of $T$. Thus, by the expansion in Part $(i)$, we obtain that $\sqrt{n}(\widehat \theta_b - \widehat \theta)\leadsto T$ in probability. That is, $\sup_x|{\Pr}_\xi(\sqrt{n}(\widehat \theta_b - \widehat \theta) \leq x)$ $- $ $F_T(x)|=o_P(1)$. Consequently, by the Continuous Mapping Theorem, 
	\begin{equation*}
		\sup_x|{\Pr}_\xi(|\sqrt{n}(\widehat \theta_b - \widehat \theta)| \leq x)-F_{|T|}(x)|=o_P(1)\,,    
	\end{equation*}
	where $F_{|T|}$ denotes the distribution function of $|T|$, with $T$ defined above.  Since $F_{|T|}$ is continuous and strictly increasing, its quantile function is continuous. Hence, by the above display and \citet[Lemma 21.2]{van2000asymptotic}, we obtain
	\begin{equation}\label{eq: convergence of bootstrap quantile}
		\widehat c_{1-\alpha}=c_{1-\alpha}+o_P(1)\,,
	\end{equation}
	where $c_{1-\alpha}$ denotes the $(1-\alpha)$ quantile of $F_{|T|}$, and $\widehat c_{1-\alpha}$ is the bootstrap $(1-\alpha)$ quantile of $|\sqrt{n}(\widehat \theta_b-\widehat \theta)|$, as defined in Section \ref{sec:bootstrap_test_for_theta0}. Next, combining Proposition \ref{prop: ifr}(ii) with the Continuous Mapping Theorem, we obtain $|\sqrt{n}(\widehat \theta-\theta_0)|\leadsto |T|$. Thus, by \eqref{eq: convergence of bootstrap quantile} and \citet[Theorem 2.7(v)]{van2000asymptotic}, $(|\sqrt{n}(\widehat \theta-\theta_0)|,\widehat c_{1-\alpha})\leadsto (|T|,c_{1-\alpha})$. By the Continuous Mapping Theorem, this implies that  $|\sqrt{n}(\widehat \theta-\theta_0)|-\widehat c_{1-\alpha}\leadsto |T|-c_{1-\alpha}$. Since the distribution function of $|T|$ is continuous at any point, the distribution function of $|T|-c_{1-\alpha}$ is continuous at 0. Hence, $\Pr(|\sqrt{n}(\widehat \theta - \theta_{0})|-\widehat c_{1-\alpha}\leq 0)$ $\rightarrow \Pr(|T|-c_{1-\alpha}\leq 0)$. Consequently,
	\begin{equation*}
		\Pr(|\sqrt{n}(\widehat \theta - \theta_{0})|>\widehat c_{1-\alpha})=1-\Pr(|\sqrt{n}(\widehat \theta - \theta_{0})|\leq\widehat c_{1-\alpha})\rightarrow 1-\Pr(|T|\leq c_{1-\alpha})=\alpha\, .
	\end{equation*}
	Finally, under $\mathcal{H}_0$, we have $\theta_0=\theta_{\mathcal{H}_0}$. Hence, by the above display, 
	\begin{equation*}
		\Pr(|\sqrt{n}(\widehat \theta - \theta_{\mathcal{H}_0})|>\widehat c_{1-\alpha})\rightarrow \alpha \text{ under }\mathcal{H}_0\, .
	\end{equation*}
	
	$(iii)$ From Equation \eqref{eq: convergence of bootstrap quantile}, $\widehat c_{1-\alpha}=c_{1-\alpha}+o_P(1)$. Also, by Proposition \ref{prop: ifr}, $\widehat \theta=\theta_0+o_P(1)$. Hence, under $\mathcal{H}_1:\theta_0\neq \theta_{\mathcal{H}_0}$, we have  $\widehat \theta- \theta_{\mathcal{H}_0}=(\widehat \theta- \theta_0)+(\theta_0- \theta_{\mathcal{H}_0})=o_P(1)+c$, with $c:=\theta_0-\theta_{\mathcal{H}_0}\neq 0$. Accordingly, $|\widehat \theta - \theta_{\mathcal{H}_0}|-\widehat c_{1-\alpha}/\sqrt{n}=|c|+o_P(1)$. This implies that 
	\begin{equation*}
		\Pr(|\sqrt{n}(\widehat \theta -  \theta_{\mathcal{H}_0})|>\widehat c_{1-\alpha})=\Pr(|\widehat \theta -  \theta_{\mathcal{H}_0}|-\widehat c_{1-\alpha}/\sqrt{n}>0)\rightarrow 1\,, 
	\end{equation*}
	which delivers the desired result.  
	
	\begin{flushright}
		$\blacksquare$
	\end{flushright}
	
	\section{Auxiliary lemmas}
	\label{sec: appendix auxiliary lemmas}
	
	\begin{lemma}\label{lem: positive definiteness of F}
		Assume that $\mathcal{F}_\mu$ is symmetric and that the \(\{(X_i^T,W_i):1,\ldots,n\}\) are all different. Then, $\bm{F}$ is a positive definite matrix, that is, $\bm{\gamma}^T\bm{F}\bm{\gamma}>0$ for any $\gamma\in\mathbb{R}^n$ with $\bm{\gamma}\neq 0$.    
	\end{lemma}
	\begin{proof}
		The proof of this result is given as an intermediate step in the proof of Proposition 2.1 in \cite{beyhum2023one}.
	\end{proof}
	Before presenting the following lemmas, we recall that $\A$, $\B$, and $s$ are defined in Equation \eqref{eq: definitions of A, B, and s}; $\widehat{\A}$, $\widehat{\B}$, $\widehat s$ are defined in Equation \eqref{eq: Ahat and shat}; and $\widehat{\A}_b$, $\widehat{\B}_b$, and $\widehat s_b$ are defined in Equation \eqref{eq: definition of Ahat b, Bhat b, ans shat b}. The projection operators $\P$, $\widehat{\P}$, and $\widehat{\P}_b$ onto $\mathcal{R}(\B)$, $\mathcal{R}(\widehat{\B})$, and $\mathcal{R}(\widehat{\B}_b)$ are well defined and can be expressed as in Lemma \ref{lem: rates and compacness}\ref{lem: rates and compacness:: expression of P and Phat} and Lemma \ref{lem: rates and compacness for bootstrap}\ref{lem: rates and compacness for bootstrap:: expression of Phat b}.  
	\begin{lemma}\label{lem: compactness, and expressions of h0, hhat}
		Let Assumptions \ref{as: iid and T}, \ref{as: rkhs}(i)(ii) , and \ref{as: completeness} hold. Then,
		\begin{enumerate}[label=(\roman*)]
			\item \label{lem: compactness, and expressions of h0, hhat:: h0 and beta0} $h_0=\T^{-1}(\I-\P)s$ and $\beta_0=(\B^*\B)^{-1}\B^*(s-\A h_0)$, where $\operatorname{T}=(\operatorname{I} - \operatorname{P})\operatorname{A}$.
			\item \label{lem: compactness, and expressions of h0, hhat:: h regularized and beta} The problem $\min_{h\in\mathcal{H},\beta\in\mathbb{R}^p}\|s-\A h - \B \beta\|^2_\mu+\lambda\|h\|_{\mathcal{H}}^2$ admits a unique solution given by 
			\begin{equation*}
				h_\lambda=(\T^*\T+\lambda\I)^{-1}\T^*(\I-\P)s\;,\; \beta_\lambda=(\B^*\B)^{-1}\B^*(s-\A h_\lambda)\, .
			\end{equation*}
			\item \label{lem: compactness, and expressions of h0, hhat:: hhat and betahat} If, moreover, Assumption \ref{as: bootstrap weights} holds, with probability approaching one, the problem $\min_{h\in\mathcal{H},\beta\in\mathbb{R}^p}\|\widehat{s}-\widehat{\A} h - \widehat{\B} \beta\|^2_\mu+\lambda\|h\|_{\mathcal{H}}^2$ admits a unique solution given by 
			\begin{equation*}
				\widehat h=(\widehat{\T}^*\widehat{\T}+\lambda\I)^{-1}\widehat{\T}^*(\I-\widehat{\P})\widehat{s}\;,\; \widehat{\beta}=(\widehat{\B}^*\widehat{\B})^{-1}\widehat{\B}^*(\widehat{s}-\widehat{\A} \widehat{h})\, ,
			\end{equation*}
			where $\widehat{\T}:=(\I-\widehat{\P})\widehat{\A}$. 
			\item \label{lem: compactness, and expressions of h0, hhat:: hhatb and betahatb}  With probability approaching one, the problem $\min_{h\in\mathcal{H},\beta\in\mathbb{R}^p}\|\widehat{s}_b-\widehat{\A}_b h - \widehat{\B}_b \beta\|^2_\mu+\lambda\|h\|_{\mathcal{H}}^2$ admits a unique solution given by 
			\begin{equation*}
				\widehat h_b=(\widehat{\T}^*_b\widehat{\T}_b+\lambda\I)^{-1}\widehat{\T}^*_b(\I-\widehat{\P}_b)\widehat{s}_b\;,\; \widehat{\beta}_b=(\widehat{\B}^*_b\widehat{\B}_b)^{-1}\widehat{\B}^*_b(\widehat{s}_b-\widehat{\A}_b \widehat{h}_b)\, ,
			\end{equation*}
			where $\widehat{\T}_b:=(\I-\widehat{\P}_b)\widehat{\A}_b$.
		\end{enumerate}
	\end{lemma}
	
	\begin{proof}
		$(i)$ From the definitions of $s$, $\A$, and $\B$, Equation \eqref{eq: Bierens reformulation of integral equation} can be equivalently rewritten as 
		\begin{equation}\label{eq: reformulation of the dentifying equation}
			s=\A h_0 + \B \beta_0\, .
		\end{equation}
		As noted just above Assumption \ref{as: source conditions}, the projection operator $\P$ onto $\mathcal{R}(\B)$, the range of $\B$, is well defined. Since $(\I - \P)\B=0$, applying $(\I-\P)$ to both sides of \eqref{eq: reformulation of the dentifying equation} gives $(\I - \P)s=\T h_0$, where $\T=(\I - \P)\A$. By Lemma \ref{lem: rates and compacness}\ref{lem: rates and compacness:: injectivity of T}, $\T$ is injective, and hence $h_0=\T^{-1}(\I-\P)s$. To obtain an expression for $\beta_0$, we apply $\B^*$ to both sides of \eqref{eq: reformulation of the dentifying equation} and rearrange terms to get $(\B^*\B)\beta_0=\B^*(s-\A h_0)$. By Assumption \ref{as: completeness}(i), $\B$ is injective, so that $\mathcal{N}(\B)=\{0\}$, where $\mathcal{N}(\B)$ denotes the null space of $\B$. Since $\mathcal{N}(\B)=\mathcal{N}(\B^*\B)$, it follows that $\B^*\B$ is injective. Hence, we obtain $\beta_0=(\B^*\B)^{-1}\B^*(s-\A h_0)$. \\
		
		$(ii)$ Note first that 
		\begin{align}\label{eq: profiling the PL problem}
			\min_{h\in\mathcal{H},\beta\in\mathbb{R}^p}\|s-\A h - \B \beta\|^2_\mu+\lambda \|h\|^2_\mathcal{H}=&\min_{h\in\mathcal{H}}\Big\{\min_{\beta\in\mathbb{R}^p}\|s-\A h - \B \beta\|^2_\mu\Big\}+\lambda \|h\|^2_\mathcal{H}\nonumber\\
			&= \min_{h\in\mathcal{H}}\Big\{\min_{a\in\mathcal{R}(\B)}\|s-\A h - a\|^2_\mu\Big\}+\lambda \|h\|^2_\mathcal{H}\, .
		\end{align}
		As established just above Assumption \ref{as: source conditions}, $\mathcal{R}(\B)$ is a closed linear subspace. Thus, by \citet[Theorems 1.26 and 13.3]{kress1999linear}, for any $h\in\mathcal{H}$, the problem $\min_{b\in\mathcal{R}(\B)}\|s-\A h - b\|^2_\mu$ admits a unique solution given by $b_h=\P(s-\A h)$. Accordingly, 
		\begin{align}\label{eq: profiled PL problem}
			\min_{h\in\mathcal{H}}\Big\{\min_{b\in\mathcal{R}(\B)}\|s-\A h - b\|^2_\mu\Big\}+\lambda \|h\|^2_\mathcal{H}=\min_{h\in\mathcal{H}}\|(\I-\P)s - (\I-\P)\A h\|^2_\mu+\lambda \|h\|^2_\mathcal{H}\, .
		\end{align}
		By Lemma \ref{lem: rates and compacness}\ref{lem: rates and compacness:: compactness of A and T}, $\T=(\I-\P)\A$ is a compact operator. Because compact operators are bounded (see \citet[Theorem 2.14]{kress1999linear}), $\T$ is also bounded. Hence, by \citet[Theorem 16.1]{kress1999linear}, the problem on the RHS of \eqref{eq: profiled PL problem} admits a unique solution given by
		\begin{equation*}
			(\T^*\T+\lambda \I)^{-1}\T^*(\I-\P)s=:h_\lambda.  
		\end{equation*}
		Plugging $h_\lambda$ into the expression for $b_h$ obtained earlier yields $b_{h_\lambda}=\P(s-\A h_\lambda)$. Since $b_{h_\lambda}\in\mathcal{R}(\B)$, there exists a $\beta_\lambda\in\mathbb{R}^p$ such that $\B \beta_\lambda=b_{h_\lambda}$. Thus, $\B \beta_\lambda=\P(s-\A h_\lambda)$ and   $\B^*\B \beta_\lambda=\B^*\P(s-\A h_\lambda)$. As  noted in the proof of Part $(i)$, $\B^*\B$ is injective, and hence $\beta_\lambda=(\B^*\B)^{-1}\B^*\P(s-\A h_\lambda)$. Since, by Lemma \ref{lem: rates and compacness}\ref{lem: rates and compacness:: expression of P and Phat}, $\P=\B(\B^*\B)^{-1}\B^*$, we have $\B^*\P=\B^*$. Thus,  $\beta_\lambda=(\B^*\B)^{-1}\B^*(s-\A h_\lambda)$. Now, recall that $h_\lambda$ is a solution to the problem on the RHS of \eqref{eq: profiled PL problem}. Consequently, by the two displays above, $(h_\lambda,\beta_\lambda)$ is a solution to the problem on the left-hand side of \eqref{eq: profiling the PL problem}. By the strict convexity of the squared norm, this solution must be unique.\\
		
		$(iii)$ By Lemma \ref{lem: rates and compacness}\ref{lem: rates and compacness:: compactness of That}, $\widehat{\T}$ is a compact operator. Thus, as long as $\widehat{\B}^*\widehat{\B}$ is injective with probability approaching one, the proof proceeds analogously to that of Part~$(ii)$. We first note that, since both $\widehat{\B}^*\widehat{\B}$ and $\B^*\B$ are linear operators from $\mathbb{R}^p$ to $\mathbb{R}^p$, they can be represented as $p\times p$ matrices (see \citet[Section 2.9]{kreyszig1991introductory}). By Lemma \ref{lem: rates and compacness}\ref{lem: rates and compacness:: Bhatstar Bhat inverse rate}, $\|\widehat{\B}^*\widehat{\B} - \B^*\B\|_{op}=o_P(1)$. Hence, the injectivity of $\B^*\B$ implies that $\widehat{\B}^*\widehat{\B}$ is injective with probability approaching one.    \\
		
		$(iv)$ The proof follows by the same arguments as in the proof of Part~$(iii)$, replacing Lemma \ref{lem: rates and compacness}\ref{lem: rates and compacness:: Bhatstar Bhat inverse rate}\ref{lem: rates and compacness:: compactness of That} with Lemma \ref{lem: rates and compacness for bootstrap}\ref{lem: rates and compacness for bootstrap:: Bhatstar b Bhat b inverse rate}\ref{lem: rates and compacness for bootstrap:: compactness of That b}. 
		
	\end{proof}
	
	\begin{lemma}\label{lem: rate for hhat}
		Let Assumptions \ref{as: iid and T}, \ref{as: rkhs}(i)(ii),  and \ref{as: completeness} hold. If $n \lambda \rightarrow \infty$ and $\lambda\rightarrow 0$ , then
		\begin{equation*}
			\|\widehat h - h_0\|_{\mathcal{H}}=o_P(1)\,. 
		\end{equation*}
		If moreover $h_0\in\mathcal{R}[(\T^*\T)^{b/2}]$, then
		\begin{equation*}
			\|\widehat h - h_0\|_{\mathcal{H}}=O_P\left(\frac{1}{\sqrt{n \lambda}}+\lambda^{(b\wedge 2)/2}\right)\, .
		\end{equation*}
	\end{lemma}
	\begin{proof}
		The proof follows by arguments analogous to those in the proof of \citet[Theorem 3.1]{beyhum2023one}. For completeness, we include the details here. As obtained in Equation  \eqref{eq: hhat Tikhonov expression} (see the proof of Proposition \ref{prop:computation_solution_empirical_program}), 
		\begin{equation*}
			\widehat h =(\widehat{\T}^* \widehat{\T} + \lambda \I)^{-1}\widehat{\T}^*(\I-\widehat{\P})\widehat s\,,
		\end{equation*}
		where $\widehat{\T}=(\I - \widehat{\P})\widehat{\A}$,  $\widehat{\P}$ is the projection operator onto $\mathcal{R}(\widehat{\B})$, the range of $\widehat{\B}$, and $(\widehat{\A}, \widehat{\B},\widehat s)$ are defined in Equation \eqref{eq: Ahat and shat}. Next, notice that
		\begin{equation*}
			\widehat h- h_0= \widehat{S}_1+\widehat{S}_2+\widehat{S}_3+\widehat{S}_4+h_\lambda - h_0\,,
		\end{equation*}
		where
		\begin{gather*}
			\widehat{S}_1=(\T^*\T+\lambda \I)^{-1}\T^*[(\I - \widehat{\P})\widehat s - \widehat{\T} h_0]\,,\, \widehat{S}_2=(\T^*\T+\lambda \I)^{-1}(\widehat{\T}^*- \T^*)[(\I - \widehat{\P})\widehat s - \widehat{\T} h_0]\,,\\
			\widehat{S}_3=[(\widehat{\T}^*\widehat{\T}+\lambda \I)^{-1}-(\T^*\T+\lambda \I)^{-1}]\widehat{\T}^*[(\I - \widehat{\P})\widehat s - \widehat{\T}h_0]\,,\, \widehat S_4=(\widehat{\T}^*\widehat{\T}+\lambda \I)^{-1}\widehat{\T}^*\widehat{\T} h_0 - h_\lambda\,,\\
			h_\lambda = (\T^*\T+\lambda \I)^{-1}\T^* \T h_0\, .
		\end{gather*}
		Since $\lambda \rightarrow 0$ and $n \lambda \rightarrow \infty$, Lemma \ref{lem: uniform in lambda}(i)-(iv) implies $\|\widehat S_j\|_{\mathcal{H}}=O_P(1/\sqrt{n \lambda})$ for $j$=1,2,3,4. Moreover, since $\T$ is injective and compact as stated in Lemma \ref{lem: rates and compacness}\ref{lem: rates and compacness:: compactness of A and T}\ref{lem: rates and compacness:: injectivity of T}, we can apply \citet[Definition 15.5 and Theorem 15.23]{kress1999linear} to conclude that $\|h_\lambda - h_0\|_{\mathcal{H}}=o(1)$. Thus, the first part of the lemma follows. To obtain the second part of the lemma, recall that $\T$ is compact. Hence, combining $h_0\in\mathcal{R}[(\T^*\T)^{b/2}]$ and Lemma \ref{lem: Bounds on Compact Operator}\ref{lem: Bounds on Compact Operator: vi} yields  $\|h_\lambda - h_0\|_{\mathcal{H}}=O(\lambda^{(b\wedge 2)/2})$. This establishes the second part of the lemma and concludes the proof.  
	\end{proof}
	
	\begin{lemma}\label{lem: bounds for decomposition of hhat-h0 for asynorm}
		Let Assumption \ref{as: iid and T} and \ref{as: rkhs}(i) hold. Let  $\T$ be as defined in Section \ref{sec:asymptotic_analysis}, $\widehat \T$ as in Lemma  \ref{lem: compactness, and expressions of h0, hhat}\ref{lem: compactness, and expressions of h0, hhat:: hhat and betahat}, and $\widehat \P$ the projection operator onto $\mathcal{R}(\widehat \B)$, the range of $\widehat \B$. Let $(\widehat \B,\widehat s)$ be as in Equation \eqref{eq: Ahat and shat}, 
		$\widehat r:=(\I-\widehat \P)\widehat s$, and $h_\lambda=(\T^*\T+\lambda \I)^{-1}\T^*\T h_0$. Then,
		\begin{enumerate}[label=(\roman*)]
			\item \label{lem: bounds for decomposition of hhat-h0 for asynorm:: Theta 2} 
			\begin{align*}
				\!\!\!\!\!\!\!
				\left<(\T^*\T+\lambda \I)^{-1}(\widehat \T^* -\T^*)(\widehat r-\widehat \T h_0),g\right>_\mathcal{H}=O_P\Big(\|\widehat \T^*-\T^*\|_{op}\,\|\widehat r - \widehat \T h_0\|_\mu\,\|(\lambda\I+\T^*\T)^{-1}g\|_\mathcal{H}\Big)\, .
			\end{align*}
			\item \label{lem: bounds for decomposition of hhat-h0 for asynorm:: Theta 3} 
			\begin{align*}
				\!\!\!\!\!\!\!
				\left<[(\widehat \T^* \widehat \T + \lambda \I)^{-1}-( \T^*  \T + \lambda \I)^{-1}]\widehat \T^*(\widehat r - \widehat \T h_0), g\right>_\mathcal{H}=&\;O_P\Big( \|\widehat r - \widehat \T h_0\|_\mu\cdot\|\widehat \T-\T\|_{op}\cdot\\
				&\Big[ \|(\widehat \T^* \widehat \T + \lambda \I)^{-1}\widehat \T^*\|_{op}\,\|\T( \T^*  \T + \lambda \I)^{-1}g\|_\mu\\
				&+ \|\widehat \T(\widehat \T^* \widehat \T + \lambda \I)^{-1}\widehat \T^*\|_{op}\,\|( \T^*  \T + \lambda \I)^{-1}g\|_\mathcal{H}\Big]\Big)\, .
			\end{align*}
			\item \label{lem: bounds for decomposition of hhat-h0 for asynorm:: Theta 4} If $g\in\mathcal{R}(\T^*\T)$, 
			\begin{align*}
				\!\!\!\!\!\!\!
				\left<(\widehat \T^* \widehat \T + \lambda \I)^{-1}\widehat \T^* \widehat \T h_0 - h_\lambda, g\right>_\mathcal{H}=O_P\Big(&\|(\widehat \T^* \widehat \T + \lambda \I)^{-1}\widehat \T^* \widehat \T h_0-h_\lambda\|_\mathcal{H}\,\|\widehat \T - \T\|_{op}\\
				&+\|\widehat \T [(\widehat \T^* \widehat \T + \lambda \I)^{-1}\widehat \T^* \widehat \T h_0-h_\lambda]\|_\mu\Big)
			\end{align*}
			with 
			\begin{align*}
				\|(\widehat \T^* \widehat \T + \lambda \I)^{-1}\widehat \T^* \widehat \T h_0 - h_\lambda\|_\mathcal{H}\leq \lambda \cdot\|\widehat \T-\T\|_{op}\cdot &\|(\T^*\T+\lambda\I)^{-1}h_0\|_\mathcal{H}\cdot\\
				&\Big[\|(\widehat \T^* \widehat \T + \lambda \I)^{-1}\widehat \T^*\|_{op}+\|(\widehat \T^* \widehat \T + \lambda \I)^{-1}\|_{op} \|\T\|_{op}\Big]
			\end{align*}
			and 
			\begin{align*}
				\|\widehat \T [(\widehat \T^* \widehat \T + \lambda \I)^{-1}\widehat \T^* \widehat \T h_0 - h_\lambda] \|_\mu\leq \lambda &\cdot \|\widehat \T-\T\|_{op}\cdot \|(\T^*\T+\lambda\I)^{-1}h_0\|_\mathcal{H}\\
				&\cdot\Big[\|\widehat \T(\widehat \T^* \widehat \T + \lambda \I)^{-1}\widehat \T^*\|_{op}+\|\T\|_{op}\|\widehat \T(\widehat \T^* \widehat \T + \lambda \I)^{-1}\|_{op}\Big]\, .
			\end{align*}
		\end{enumerate}
	\end{lemma}
	\begin{proof}
		The proofs of these results are contained in the proof of \citet[Lemma 2.2]{beyhum2024testing}.
	\end{proof}
	
	\begin{lemma}\label{lem: ase}
		Let $\widehat f:\mathcal{Z}\to \mathbb{R}$ be a data-dependent function and $f_0:\mathcal{Z}\to \mathbb{R}$ be a deterministic function. Assume that $\sup_{z\in\mathcal{Z}}|\widehat f(z)- f_0(z)|=o_P(1)$ and $\Pr(\widehat f\in\mathcal{G})\rightarrow 1$ for a deterministic class of functions $\mathcal{G}$ such that $\log N_{[\,]}(\epsilon,\mathcal{G},L^2(P))\leq C \epsilon ^{-v}$ with $v\in(0,2)$. Then,
		\begin{enumerate}[label=(\roman*)]
			\item \label{lem: ase:: sample}
			\begin{equation*}
				\sqrt{n}(\E_n -\E)[\widehat f(Z_i) - f_0(Z_i)]=o_P(1)\, ,
			\end{equation*}
			where $\operatorname{E}\widehat f(Z_i)=\int_{ }\widehat f(z) d P_Z(z)$ and $P_Z$ is the probability distribution of $Z$.
			\item \label{lem: ase:: bootstrap} If moreover $\{\zeta_i:i=1,\ldots,n\}$ are random variables independent from $\{Z_i:i=1,\ldots,n\}$ with $\operatorname{E}\zeta^2<\infty$, then 
			\begin{equation*}
				\sqrt{n}(\E_n -\E)\zeta_i[\widehat f(Z_i) - f_0(Z_i)]=o_P(1)\, .
			\end{equation*}
		\end{enumerate}
	\end{lemma}
	\begin{proof}
		The proofs follow from a simple application of \citet[Lemma 19.34]{van2000asymptotic}.
	\end{proof}
	
	\begin{lemma}\label{lem: Bounds on Compact Operator}
		Let $\mathcal{X}$ and $\mathcal{Y}$ be two Hilbert spaces with norms $\|\cdot\|_\mathcal{X}$ and $\|\cdot\|_\mathcal{Y}$, respectively. Let $\operatorname{C}:\mathcal{X}\to \mathcal{Y} $ be a linear compact operator with singular system $(\widetilde{\mu}_j,\widetilde{\varphi}_j,\widetilde{\psi}_j)_j$, where $(\widetilde{\mu}_j)_j$ is a sequence of positive values in $\mathbb{R}$, $(\widetilde{\varphi}_j)_j$ is a sequence of orthonormal elements in $\mathcal{X}$, and $(\widetilde{\psi}_j)_j$ is a sequence of orthonormal elements in $\mathcal{Y}$. Let $\lambda\downarrow 0$. Then, \\
		\begin{enumerate}[label=(\roman*)]
			\item \label{lem: Bounds on Compact Operator: i} $$||\operatorname{C}(\lambda \I + \operatorname{C}^*\operatorname{C})^{-1}\operatorname{C}^*||_{op}\leq 1 \,.$$
			\item \label{lem: Bounds on Compact Operator: ii}$$||\lambda(\lambda \I + \operatorname{C}^* \operatorname{C})^{-1}||_{op}\leq 2 \, .$$ 
			\item \label{lem: Bounds on Compact Operator: iii} $$||(\lambda \I + \operatorname{C}^* \operatorname{C} )^{-1}\operatorname{C}^* ||_{op}\leq \frac{1}{2\sqrt{\lambda}}\, .$$ 
		\end{enumerate}
		\noindent If, moreover, $\phi\in\mathcal{R}[(\operatorname{C}^*\operatorname{C})^{\gamma/2}]$, that is,  $||\phi||_\gamma^2:=\sum_j\widetilde{\mu}_j^{-2\gamma}|\left<\phi,\widetilde{\varphi}_j\right> |^2<\infty$, then
		\begin{enumerate}[label=(\roman*), start=4]
			\item \label{lem: Bounds on Compact Operator: iv}  
			$$||\lambda(\lambda \I + \operatorname{C}^*\operatorname{C})^{-1}\phi||_\mathcal{X}
			=O(\lambda^{\frac{\gamma\wedge 2}{2}})\, .$$ 
			\item \label{lem: Bounds on Compact Operator: v}   $$||\lambda \operatorname{C} (\lambda \I + \operatorname{C}^*\operatorname{C})^{-1} \phi ||_\mathcal{Y}
			=O(\lambda^{\frac{\gamma+1}{2}\wedge 1 }) \, . $$
			\item \label{lem: Bounds on Compact Operator: vi} 
			$$||(\lambda \I + \operatorname{C}^*\operatorname{C})^{-1}\operatorname{C}^*\operatorname{C}\phi - \phi ||_\mathcal{X}=O(\lambda^{\frac{\gamma \wedge 2}{2}})\, .$$
			
		\end{enumerate} 
	\end{lemma}
	\begin{proof}
		These inequalities are well known in the inverse problems literature. The proofs of parts $(i)$, $(ii)$, and $(iii)$ can be found in \citet[Lemma A.1]{florens2011identification}. The proofs of parts $(iv)$ and $(v)$ follow from arguments similar to those used in the proofs of \citet[Lemma A.1]{florens2011identification} and \citet[Theorem 3.1]{babii2022high}. Finally, the proof of part $(vi)$ can be found in \citet[Proposition 3.11]{carrasco2007linear}.
	\end{proof}
	
	\begin{lemma}
		\label{lem: rates and compacness}
		Under Assumptions  \ref{as: iid and T}, \ref{as: rkhs}(i), \ref{as: completeness}, and $h_0\in\mathcal{H}$,
		\begin{enumerate}[label=(\roman*)]
			\item \label{lem: rates and compacness:: compactness of A and T} the operators $\A$ and $\T$ are compact;
			\item \label{lem: rates and compacness:: injectivity of T} $\T$ is injective;
			\item \label{lem: rates and compacness:: Bhat rate} $\|\widehat{\B}-\B\|_{op}=O_p(n^{-1/2})$;
			\item \label{lem: rates and compacness:: Bhatstar Bhat inverse rate} $\|\widehat{\B}^*\widehat{\B}- \B^*
			\B\|_{op} = O_p(n^{-1/2})$ and  $\|(\widehat{\B}^*\widehat{\B})^{-1}- (\B^*
			\B)^{-1}\|_{op} = O_p(n^{-1/2})$; 
			\item \label{lem: rates and compacness:: expression of P and Phat} $\P=\B(\B^*\B)^{-1}\B^*$ and $\widehat \P=\widehat\B(\widehat\B^*\widehat\B)^{-1}\widehat\B^*$ with probability approaching one;
			\item \label{lem: rates and compacness:: Ahat rate}  $\|\widehat{\A}-\A\|_{op}=O_p(n^{-1/2})$;
			\item \label{lem: rates and compacness:: compactness of That} the operator $\widehat{\T}$ is compact;
			\item \label{lem: rates and compacness:: That rate} $\|\widehat{\T}-\T\|_{op}=O_p(n^{-1/2})$;
			\item \label{lem: rates and compacness:: shat rate}  $\left\|\widehat{s}-s\right\|_\mu=O_p(n^{-1/2})$;
			\item  \label{lem: rates and compacness:: (I - Phat) shat - That h rate}  $\|(\I-\widehat{\P})\widehat{s}-\widehat{\T}h_0\|_\mu=O_p(n^{-1/2})$.
		\end{enumerate}
	\end{lemma}
	
	\begin{proof}
		$(i)$ Let us first show that $\A$ is a compact operator. Pick any $h\in\mathcal{H}$, and recall that $\ell^2$ is the space of square-summable sequences. By Lemma \ref{lem: RKHS results}\ref{lem: RKHS results: equality of H and H eig space},  there exists a sequence $(\beta_j)_j\in\ell^2$ such that $h(z)=\sum_{j=1}^\infty \beta_j \widetilde \phi_j(z)$ pointwise in $z\in\mathcal{Z}$, where $\widetilde \phi_j(z)=\sqrt{\eta_j}\phi_j(z)$, $(\eta_j)_j$ is a sequence of positive scalars, and $(\phi_j)_j$ is a sequence of functions defined on $\mathcal{Z}$. Note that $\sum_{j=1}^\infty \beta_j \widetilde \phi_j(z)$ is a well-defined limit. In fact, by Lemma \ref{lem: RKHS results}\ref{lem: RKHS results: upper bound for l2 sum}, we have $\sum_{j=1}^\infty|\widetilde \phi_j(z)|^2=\sum_{j=1}^\infty\eta_j| \phi_j(z)|^2\leq K(z,z)$. Thus, both sequences $(\widetilde \phi_j(z))_j$ and $(\beta_j)_j$ belong to $\ell^2$, and, denoting by $\left<\cdot,\cdot,\right>_{\ell^2}$ the inner product on $\ell^2$, we obtain that $\sum_{j=1}^\infty \beta_j \widetilde \phi_j(z)=\left<(\beta_j)_j,(\widetilde \phi_j(z))_j\right>_{\ell^2}$ is well defined. Then (see the comments below),    
		\begin{align}\label{eq: from A to Atilde}
			(\A h)(t)=&\operatorname{E}h(Z_i) \exp(\textbf{i}(X_i^T,W_i)t)=\operatorname{E}\lim_{N\rightarrow \infty}\sum_{j=1}^N \beta_j\widetilde \phi_j(Z_i) \exp(\textbf{i}(X_i^T,W_i)t)\nonumber\\
			&\overset{(a)}{=}\lim_{N\rightarrow \infty} \sum_{j=1}^N \beta_j \operatorname{E}\widetilde \phi_j(Z_i) \exp(\textbf{i}(X_i^T,W_i)t)=: \widetilde{\A}[(\beta_j)_j](t)\, .
		\end{align}
		To justify the interchange of $\lim_{N\rightarrow \infty}$ and $\operatorname{E}$ in equality $(a)$, observe that by $|\exp(\textbf{i}(X_i^T,W_i)t)|\leq 1$ and the Cauchy-Schwarz inequality, we have $|\sum_{j=1}^N \beta_j \widetilde \phi_j(Z_i) \exp(\textbf{i}(X_i^T,W_i)t)|$ $\leq (\sum_{j=1}^N \beta_j^2)^{1/2}$ $ [\sum_{j=1}^N \widetilde \phi_j(Z_i)^2]^{1/2}$ $ \leq (\sum_{j=1}^\infty \beta_j^2)^{1/2} [\sum_{j=1}^\infty \widetilde \phi_j(Z_i)^2]^{1/2}$. As shown earlier in this proof,  \([\sum_{j=1}^\infty \widetilde \phi_j(Z_i)^2]^{1/2}\)  \(\leq K(Z_i,Z_i)^{1/2}\), where $\operatorname{E}K(Z_i,Z_i)^{1/2}<\infty$ by the continuity of $K$ and the compactness of $\mathcal{Z}$ (see Assumptions \ref{as: iid and T}(iii) and  \ref{as: rkhs}(i)). Thus, by the Lebesgue Dominated Convergence Theorem, we can interchange  $\lim_{N\rightarrow \infty}$ and $\operatorname{E}$ and obtain equality $(a)$ of the above display.\\
		Next, note that $\widetilde{\A}:\ell^2\to L^2_\mu$. Indeed, since $\sum_{j=1}^\infty |\operatorname{E}\widetilde{\phi}_j(Z_i)\exp(\textbf{i}(X_i^T,W_i)t)|^2$ $\leq \sum_{j=1}^\infty \operatorname{E}|\widetilde \phi_j(Z_i)|^2$ $\leq \operatorname{E}K(Z_i,Z_i)<\infty$, we have
		\begin{align*}
			\|\widetilde{\A}[(\beta_j)_j]\|^2_\mu=&=\int_{ }|\sum_{j=1}^\infty \beta_j \operatorname{E}\widetilde{\phi}_j(Z_i)\exp(\textbf{i}(X_i^T,W_i)t)|^2 d\mu(t) \\
			\leq &\int_{ }\sum_{j=1}^\infty \beta_j^2 \sum_{j=1}^\infty |\operatorname{E}\widetilde \phi_j(Z_i) \exp(\textbf{i}(X_i^T,W_i)t)|^2 d\mu(t)<\infty \, .
		\end{align*}
		Let us now show that $\widetilde{\A}:\ell^2\to L^2_\mu$ is a compact operator.  Since every Hilbert-Schmidt operator is  compact (see \citet[Theorem 2.32]{carrasco2007linear}), it suffices to prove that $\widetilde{\A}$ is Hilbert-Schmidt. By \citet[Definition 2.30]{carrasco2007linear}, $\widetilde{\A}$ is Hilbert-Schmidt if for a complete orthonormal system $\{(\gamma_{s,j})_j:s=1,2,\ldots\}\subset\ell^2$, we have $\sum_{s=1}^\infty \|\widetilde{\A}[(\gamma_{s,j})_j]\|^2_\mu<\infty$. By Perseval's formula (see \citet[Theorem 2.8]{carrasco2007linear}), $\{(\gamma_{s,j})_j:s=1,2,\ldots\}$ is a complete orthonormal set in $\ell^2$ if $\sum_{s=1}^\infty (\sum_{j=1}^\infty \gamma_{s,j} \beta_j)^2=\sum_{j=1}^\infty \beta_j^2$ for any $(\beta_j)_j\in\ell^2$. Now, setting $\gamma_{s,j}=1$ if $s=j$ and $\gamma_{s,j}=0$ if $s\neq j$, we see that $\{(\gamma_{s,j})_j:s=1,2,\ldots\}$ is indeed a complete orthonormal set in $\ell^2$. Then, 
		\begin{align*}
			\sum_{s=1}^\infty \|\widetilde{\A}[(\gamma_{s,j})_j]\|^2_\mu=&\sum_{s=1}^\infty \int_{ } \Big| \sum_{j=1}^\infty \gamma_{s,j} \operatorname{E}\widetilde \phi_j(Z_i)\exp(\textbf{i}(X_i^T,W_i)t)\Big|^2 d \mu(t)\\
			=&\sum_{s=1}^\infty \int_{ } \Big| \operatorname{E}\widetilde \phi_s(Z_i)\exp(\textbf{i}(X_i^T,W_i)t) \Big|^2 d \mu(t)\\
			\leq& \sum_{s=1}^\infty   \operatorname{E}|\widetilde \phi_s(Z_i)|^2 \\
			\leq & \operatorname{E}K(Z_i,Z_i)<\infty\, . 
		\end{align*}
		This shows that $\widetilde{\A}$ is Hilbert- Schmidt and hence compact. We now use this result to prove that $\A$ is also compact. By \citet[Theorem 2.13]{kress1999linear}, the operator $\A:\mathcal{H}\to L^2_\mu$  is compact if and only if, for every bounded sequence $(h_s)_s$ in $\mathcal{H}$, the sequence $(\A h_s)_s$ admits a convergent subsequence in $L^2_\mu$. Using Lemma \ref{lem: RKHS results}\ref{lem: RKHS results: H eig space}\ref{lem: RKHS results: equality of H and H eig space}, we have $h_s(z)=\sum_{j=1}^\infty \beta_{s,j} \widetilde \phi_j(z)$ pointwise in $z\in\mathcal{Z}$, for some coefficients $(\beta_{s,j})_j\in \ell^2$, and $\|h_s\|^2_\mathcal{H}=\left<h_s,h_s\right>_{\mathcal{H}}$ $=\left<h_s,h_s\right>_{EIG}$ $=\sum_{j=1}^\infty (\beta_{s,j})^2$. Thus, the bounded sequence  $(h_s)_s$ in $\mathcal{H}$ corresponds to a bounded sequence $\{(\beta_{s,j})_j:s=1,2,\ldots\}$ in $\ell^2$. Since $\widetilde{\A}:\ell^2\to L^2_\mu$ is compact,  the sequence $\{\widetilde{\A}[(\beta_{s,j})_j]:s=1,2,\ldots\}$ has a convergent subsequence in $L^2_\mu$. Therefore, by \eqref{eq: from A to Atilde}, $(\A h_s)_s$ also admits a convergent subsequence. This shows that $\A$ is compact.\\
		We now establish the compactness of $\T=(\I-\P)\A$.
		By \citet[Theorem 13.3]{kress1999linear}, projection operators are bounded. Hence $\P$ is bounded, and so is $\I-\P$. Moreover, by \citet[Theorems 2.14 and 2.16]{kress1999linear}, the composition of a bounded and a compact operator is compact. Therefore, since $\A$ is compact, it follows that $(\I-\P)\A=\T$ is compact.  \\
		
		$(ii)$ To show that $\T$ is injective, it suffices to prove that $\T h=0 \Rightarrow h=0$. Since $\T =(\I - \P)\A $, the condition $\T h=0$ implies that $\A h=\P\A h$. As $\P$ is the projection operator onto $\mathcal{R}(\B)$, we have $\P \A h \in \mathcal{R}(\B)$. Hence, $\A h \in \mathcal{R}(\A) \cap\mathcal{R}(\B)$. By Assumption \ref{as: completeness}(ii), $\mathcal{R}(\A) \cap\mathcal{R}(\B)=\{0\}$, and therefore $\A h=0$. Finally, since $\A$ is injective by Assumption \ref{as: completeness}(i), it follows that $h=0$. \\
		
		$(iii)$ Note first that $\|(\widehat{\B} - \B)\beta\|^2_\mu=\int_{ }|(\operatorname{E}_n - \operatorname{E}) X_i^T\beta \exp(\textbf{i}(X_i^T,W_i)t)|^2 d\mu(t) $ $\leq \|\beta\|^2 \int_{ }\|(\operatorname{E}_n - \operatorname{E}) X_i $ $ \exp(\textbf{i}(X_i^T,W_i)t)\|^2$ $d \mu(t)$, where $\|\cdot\|$ denotes the Euclidean norm on $\mathbb{R}^p$. Thus, 
		\begin{align*}
			\operatorname{E}\|\widehat{\B} - \B\|_{op}^2=\operatorname{E}\sup_{\beta\in\mathbb{R}^p,\|\beta\|=1}\|(\widehat{\B} - \B)\beta\|^2_\mu\leq \operatorname{E} \int_{ }\|(\operatorname{E}_n - \operatorname{E}) X_i  \exp(\textbf{i}(X_i^T,W_i)t)\|^2 d \mu(t)\, .
		\end{align*}
		Without loss of generality, suppose that $X\in\mathbb{R}$. Then, the RHS of the above display equals
		\begin{align*}
			\int_{ }\operatorname{E}|(\operatorname{E}_n - \operatorname{E}) X_i \exp(\textbf{i}(X_i^T,W_i)t)|^2 d \mu(t)=&\int_{ }\operatorname{Var}\Big[\frac{1}{n}\sum_{i=1}^n X_i \exp(\textbf{i}(X_i^T,W_i)t) \Big] d \mu(t)\\
			=&\int_{ }\frac{1}{n}\operatorname{Var}\Big[ X_i \exp(\textbf{i}(X_i^T,W_i)t) \Big] d \mu(t)\\
			\leq&  \int_{ }  \frac{\operatorname{E}|X_i \exp(\textbf{i}(X_i^T,W_i)t)|^2}{n}d\mu(t)\\
			\leq &\frac{\operatorname{E}X_i^2}{n}\,,
		\end{align*}
		where $\operatorname{E}X_i^2<\infty$ by Assumption \ref{as: iid and T}(i). Hence, $\operatorname{E}\|\widehat{\B} - \B\|_{op}^2=O_p(n^{-1})$, which gives the desired result.\\
		
		$(iv)$ Since $\|\widehat{\B}-\B\|_{op}=O_P(n^{-1/2})$ by Part $(iii)$, the proof proceeds in the same manner as that of \citet[Lemma S3.2(c)]{beyhum2023one}. \\
		
		$(v)$ The proof that $\P=\B(\B^*\B)^{-1}\B^*$ is given in \citet[pages 194-195]{beyhum2023one}. The proof that $\widehat{\P}=\widehat{\B}(\widehat{\B}^*\widehat{\B})^{-1}\widehat{\B}^*$ with probability approaching one (wpa1) follows by analogous arguments, provided that $\widehat{\B}$ is injective wpa1. As established in the proof of Lemma \ref{lem: compactness, and expressions of h0, hhat}\ref{lem: compactness, and expressions of h0, hhat:: hhat and betahat}, $\widehat{\B}^*\widehat{\B}$ is injective wpa1, which implies $\mathcal{N}(\widehat{\B}\widehat{\B}^*)=\{0\}$ wpa1. Since $\mathcal{N}(\widehat{\B})=\mathcal{N}(\widehat{\B}\widehat{\B}^*)$, we conclude that $\mathcal{N}(\widehat{\B})=\{0\}$ wpa1. Hence, $\widehat{\B}$ is injective wpa1. \\
		
		$(vi)$ As noted in the proof of Part $(i)$, for any $h\in\mathcal{H}$ there exists a sequence $(\beta_j)_j\in\ell^2$ such that $h(z)=\sum_{j=1}^\infty \beta_j \widetilde{\phi}_j(z)$ pointwise in $Z\in\mathcal{Z}$ and $\|h\|^2_{\mathcal{H}}=\sum_{j=1}^\infty \beta_j^2$. Hence (see the comments below),
		\begin{align*}
			\|\widehat{\A}h - \A h\|^2_\mu=&\int_{ } \Big| \operatorname{E}_n\Big\{ \sum_{j=1}^\infty   \beta_j\widetilde \phi_j(Z_i) \exp(\textbf{i}(X_i^T,W_i)t)\Big\} - \operatorname{E}\Big\{\sum_{j=1}^\infty \beta_j \widetilde \phi_j(Z_i)\exp(\textbf{i}(X_i^T,W_i)t)\Big\} \Big|^2 d \mu(t)\\
			\overset{(a)}{=}&\int_{ } \Big|\sum_{j=1}^\infty \beta_j \operatorname{E}_n\big\{\widetilde \phi_j(Z_i) \exp(\textbf{i}(X_i^T,W_i)t)\big\} - \sum_{j=1}^\infty \beta_j \operatorname{E}\big\{\widetilde \phi_j(Z_i)\exp(\textbf{i}(X_i^T,W_i)t)\big\} \Big|^2 d \mu(t)\\
			=& \int_{ } \Big|\sum_{j=1}^\infty \beta_j (\operatorname{E}_n-\operatorname{E})\widetilde \phi_j(Z_i) \exp(\textbf{i}(X_i^T,W_i)t)\Big|^2 d \mu(t)\\
			\overset{(b)}{\leq} & \sum_{j=1}^\infty \beta_j^2 \int_{ }\sum_{j=1}^\infty\Big|(\operatorname{E}_n-\operatorname{E})\widetilde \phi_j(Z_i) \exp(\textbf{i}(X_i^T,W_i)t)\Big|^2 d \mu(t)\,.
		\end{align*}
		Here, equality $(a)$ follows from exchanging $\operatorname{E}_n$ with $\sum_{j=1}^\infty$, and $\operatorname{E}$ with $\sum_{j=1}^\infty$, as in \eqref{eq: from A to Atilde}. Inequality $(b)$ follows from Cauchy-Schwarz. 
		Then, using $\|h\|^2_{\mathcal{H}}=\sum_{j=1}^\infty \beta_j^2$ yields  
		\begin{align*}
			\operatorname{E}\|\widehat{\A} -&\A \|^2_{op}=\operatorname{E}\sup_{h\in\mathcal{H},\|h\|_{\mathcal{H}}=1}\|\widehat{\A} h -\A h \|^2_\mu\leq  \operatorname{E}\int_{ }\sum_{j=1}^\infty\Big|(\operatorname{E}_n-\operatorname{E})\widetilde \phi_j(Z_i) \exp(\textbf{i}(X_i^T,W_i)t)\Big|^2 d \mu(t)\\
			=& \int_{ }\sum_{j=1}^\infty \operatorname{E}\Big|(\operatorname{E}_n-\operatorname{E})\widetilde \phi_j(Z_i) \exp(\textbf{i}(X_i^T,W_i)t)\Big|^2 d \mu(t)= \int_{ }\sum_{j=1}^\infty \operatorname{Var}\Big[\operatorname{E}_n\widetilde \phi_j(Z_i) \exp(\textbf{i}(X_i^T,W_i)t)\Big] d \mu(t)\\
			=& \int_{ }\sum_{j=1}^\infty \frac{1}{n} \operatorname{Var}\Big[\widetilde \phi_j(Z_i) \exp(\textbf{i}(X_i^T,W_i)t)\Big] d \mu(t)\leq \int_{ }\sum_{j=1}^\infty \frac{1}{n} \operatorname{E}\Big|\widetilde \phi_j(Z_i) \exp(\textbf{i}(X_i^T,W_i)t)\Big|^2 d \mu(t)\\
			\leq& \frac{1}{n}  \sum_{j=1}^\infty \operatorname{E}|\widetilde \phi_j(Z_i)|^2\overset{(a)}{\leq} \frac{1}{n}\operatorname{E}K(Z_i,Z_i)\overset{(b)}{=}O(n^{-1}) \, ,
		\end{align*}
		where inequality $(a)$ follows from $\sum_{j=1}^\infty |\widetilde \phi_j(Z_i)|^2\leq K(Z_i,Z_i)$, as noted earlier, and equality $(b)$ follows from the compactness of $\mathcal{Z}$ and the continuity of $K$ (see Assumptions \ref{as: iid and T}(iii) and \ref{as: rkhs}(i)).
		Hence, the desired result follows.\\
		
		$(vii)$ The arguments used in the proof of Part $(i)$ remain valid  with $\widehat{\A}$ replacing $\A$, that is, with $\operatorname{E}_n$ replacing $\operatorname{E}$. Hence, $\widehat{\A}$ is a compact operator. Since $\widehat{\P}$ is the projection operator onto $\mathcal{R}(\widehat{\B})$, it is bounded, and so is $\I-\widehat{\P}$. By \citet[Theorem 2.14 and 2.16]{kress1999linear}, the composition of a bounded and a compact operator is compact. Therefore, $\widehat{\T}=(\I-\widehat{\P})\widehat{\A}$ is compact.\\
		
		$(viii)$ By Part $(v)$, $\widehat{\P}=\widehat{\B}(\widehat{\B}^*\widehat{\B})^{-1}\widehat{\B}^*$ and $\P=\B(\B^*\B)^{-1}\B^*$. Since by Parts $(iii)$ and $(iv)$ $\|\widehat{\B}-\B\|_{op}=O_P(n^{-1/2})$ and $\|(\widehat{\B}^*\widehat{\B})^{-1}- (\B^*\B)^{-1}\|_{op}=O_P(n^{-1/2})$, it follows that $\|(\widehat{\B}^*\widehat{\B})^{-1}\|_{op}=O_P(1)$ and $\|\widehat{\B}\|_{op}=O_P(1)$. Hence,
		\begin{equation*}
			\|\widehat{\P} - \P\|_{op}=\|(\widehat{\B}-\B)(\widehat{\B}^*\widehat{\B})^{-1}\widehat{\B}^* + \B [(\widehat{\B}^*\widehat{\B})^{-1} - (\B^*\B)^{-1}]\widehat{\B}^* + \B (\B^*\B)^{-1}[\widehat{\B}^* - \B^*]\|_{op}=O_p(n^{-1/2})\, .
		\end{equation*}
		By Part $(vi)$, $\|\widehat{\A}-\A\|_{op}=O_P(n^{-1/2})$ and $\|\widehat{\A}\|_{op}=O_P(1)$. Therefore,
		\begin{align*}
			\|\widehat{\T} - \T\|_{op}=\|(\I-\widehat{\P})\widehat{\A} - (\I-\P)\A\|_{op}=\|(\P-\widehat{\P})\widehat{\A} + (\I-\P)(\widehat{\A} - \A)\|_{op}=O_P(n^{-1/2})  \, .
		\end{align*}
		
		$(ix)$ The proof follows by arguments analogous to those used in the proof of Part $(iii)$. \\
		
		$(x)$ As already established in the proof of Part $(viii)$, $\|\widehat{\P} - \P\|_{op}=O_P(n^{-1/2})$. By Parts $(viii)$ and $(ix)$, $\|\widehat{\T}-\T\|_{op}=O_P(n^{-1/2})$ and $\|\widehat s - s\|_\mu=O_p(n^{-1/2})$. Moreover, by Lemma \ref{lem: compactness, and expressions of h0, hhat}\ref{lem: compactness, and expressions of h0, hhat:: h0 and beta0}, $(\I-\P)s=\T h_0$. 
		Hence,
		\begin{equation*}
			\|(\I - \widehat{\P})\widehat s - \widehat{\T} h_0\|_\mu=\|(\P - \widehat{\P})\widehat s +(\I - \P)(\widehat s - s) + (\I - \P)s - \T h_0 - (\widehat{\T} - \T)h_0\|_\mu=O_P(n^{-1/2})  \, . 
		\end{equation*}
	\end{proof}
	
	In the following lemma, we consider statements regarding the joint probability space of  the bootstrap weights  and the sample data. Recall that $\widehat \A_b$, $\widehat \B_b$, and $\widehat s_b$ are defined in \eqref{eq: definition of Ahat b, Bhat b, ans shat b}. 
	
	\begin{lemma}
		\label{lem: rates and compacness for bootstrap}
		Under Assumptions  \ref{as: iid and T}, \ref{as: rkhs}(i), \ref{as: completeness}(i), \ref{as: bootstrap weights}, and $h_0\in\mathcal{H}$,
		\begin{enumerate}[label=(\roman*)]
			\item \label{lem: rates and compacness for bootstrap:: Bhat rate} $\|\widehat{\B}_b-\B\|_{op}=O_p(n^{-1/2})$;
			\item \label{lem: rates and compacness for bootstrap:: Bhatstar b Bhat b inverse rate} $\|\widehat{\B}^*_b\widehat{\B}_b- \B^*
			\B\|_{op} = O_p(n^{-1/2})$ and $\|(\widehat{\B}^*_b\widehat{\B}_b)^{-1}- (\B^*
			\B)^{-1}\|_{op} = O_p(n^{-1/2})$;
			\item \label{lem: rates and compacness for bootstrap:: expression of Phat b}  $\widehat \P_b=\widehat\B_b(\widehat\B^*_b\widehat\B_b)^{-1}\widehat\B^*_b$ with probability approaching one;
			\item \label{lem: rates and compacness for bootstrap:: Ahat b rate}  $\|\widehat{\A}_b-\A\|_{op}=O_p(n^{-1/2})$;
			\item \label{lem: rates and compacness for bootstrap:: compactness of That b} the operator $\widehat{\T}_b$ is compact;
			\item \label{lem: rates and compacness for bootstrap:: That b rate} $\|\widehat{\T}_b-\T\|_{op}=O_p(n^{-1/2})$;
			\item \label{lem: rates and compacness for bootstrap:: shat b rate}  $\left\|\widehat{s}_b-s\right\|_\mu=O_p(n^{-1/2})$;
			\item  \label{lem: rates and compacness for bootstrap:: (I - Phat b) shat b - That b h rate}  $\|(\widehat{\I}-\widehat{\P}_b)\widehat{s}_b-\widehat{\T}_bh_0\|_\mu=O_p(n^{-1/2})$.
		\end{enumerate}
	\end{lemma}
	
	\begin{proof}
		$(i)$ First note that, for any $\beta\in\mathbb{R}^p$,
		\begin{align*}
			(\widehat{\B}_b\beta)(t) - (\B \beta)(t)=&\operatorname{E}_n \frac{\xi_i}{\overline{\xi}}X_i^T \beta \exp(\textbf{i}(X_i^T,W_i)t) - (\B\beta)(t)\\
			=&\operatorname{E}_n\frac{\xi_i -1}{\overline{\xi}}X_i^T\beta \exp(\textbf{i}(X_i^T,W_i)t)\\
			&+ \operatorname{E}_n \frac{1-\overline{\xi}}{\overline{\xi}} X_i^T\beta \exp(\textbf{i}(X_i^T,W_i)t)\\
			&+ (\widehat{\B}\beta)(t) - (\B \beta)(t)\\
			=:& (\widetilde{\B}_1 \beta)(t)+ (\widetilde{\B}_2 \beta)(t) + (\widehat{\B}\beta)(t) - (\B \beta)(t)\, .
		\end{align*}
		By Lemma \ref{lem: rates and compacness}\ref{lem: rates and compacness:: Bhat rate}, $\|\widehat{\B} - \B\|_{op}=O_P(n^{-1/2})$. Moreover, $(\widetilde{\B}_2 \beta)(t)$ $=(1-\overline{\xi})(\overline{\xi})^{-1} (\widehat{\B}\beta)(t)$ and $(1-\overline{\xi})(\overline{\xi})^{-1}=O_P(n^{-1/2})$. Thus,
		\begin{align}\label{eq: B tilde 2}
			\|\widetilde{\B}_2\|_{op}^2=\sup_{\beta\in\mathbb{R}^p,\|\beta\|=1}\|\widetilde{\B}_2\beta\|^2_\mu=\frac{(1-\overline{\xi})^2}{\overline{\xi}^2}\sup_{\beta\in\mathbb{R}^p,\|\beta\|=1}\|\widehat{\B}\beta\|^2_\mu=O_P(n^{-1}) \|\widehat{\B}\|_{op}^2=O_P(n^{-1})\,,
		\end{align}
		where in the last equality we have used $\|\widehat{\B}\|_{op}=O_P(1)$, which follows from  Lemma \ref{lem: rates and compacness}\ref{lem: rates and compacness:: Bhat rate} and $\|\B\|_{op}<\infty$. Hence, to show the desired result, it suffices to prove that $\|\widetilde{\B}_1\|_{op}=O_P(n^{-1/2})$. To this end, note that (see the comments below) 
		\begin{align*}
			\|\widetilde{\B}_1\|_{op}^2=&\sup_{\beta\in\mathbb{R}^p,\|\beta\|=1} \frac{1}{\overline{\xi}^2}\int_{ }|\beta^T\operatorname{E}_n (\xi_i-1)X_i \exp(\textbf{i}(X_i^T,W_i)t)|^2 d \mu(t)\\
			\overset{(a)}{\leq} & O_P(1) \sup_{\beta\in\mathbb{R}^p,\|\beta\|=1} \|\beta\|^2\int_{ }\|\operatorname{E}_n (\xi_i-1)X_i \exp(\textbf{i}(X_i^T,W_i)t)\|^2 d \mu(t)\\
			=& O_P(1) \int_{ }\|\operatorname{E}_n (\xi_i-1)X_i \exp(\textbf{i}(X_i^T,W_i)t)\|^2 d \mu(t) \, ,
		\end{align*}
		where $\|\cdot\|$ denotes the Euclidean norm on $\mathbb{R}^p$.
		Inequality $(a)$ follows from $1/\overline{\xi}=O_P(1)$ and Cauchy-Schwarz. Let us assume, without loss of generality, that $X\in\mathbb{R}$. Since the bootstrap weights $\{\xi_i:i=1\ldots,n\}$ are independent from the sample data and satisfy $\operatorname{E}\xi_i=1$ (Assumption \ref{as: bootstrap weights}), the empirical mean  $\operatorname{E}_n (\xi_i-1)X_i \exp(\textbf{i}(X_i^T,W_i)t)$ has zero expectation. Using this, together with $\operatorname{Var}\{\xi\}=1$ (Assumption \ref{as: bootstrap weights}) and $\operatorname{E}X^2<\infty$ (Assumption \ref{as: iid and T}(i)), we obtain 
		\begin{align*}
			\operatorname{E}\int_{ }|\operatorname{E}_n (\xi_i-1)X_i \exp(\textbf{i}(X_i^T,W_i)t)|^2 d \mu(t)  =&\int_{ }\operatorname{Var}[\operatorname{E}_n (\xi_i-1)X_i \exp(\textbf{i}(X_i^T,W_i)t)]d \mu(t)\\
			=& \int_{ }\frac{1}{n} \operatorname{Var}[(\xi_i-1)X_i \exp(\textbf{i}(X_i^T,W_i)t)]d \mu(t)\\
			\leq &\frac{1}{n}\int_{ }\operatorname{E}|(\xi_i-1)X_i \exp(\textbf{i}(X_i^T,W_i)t)|^2d \mu(t)\\
			\leq & \frac{1}{n} \operatorname{E}|(\xi_i-1)X_i|^2= \frac{1}{n} [\operatorname{E}(\xi_i-1)^2] [\operatorname{E}X_i^2]=O(n^{-1})\, . 
		\end{align*}
		Gathering the above results, we conclude that $\|\widetilde{\B}_1\|_{op}=O_P(n^{-1/2})$, which completes the proof of $(i)$.\\
		
		$(ii)-(iii)$ The proofs follow by arguments analogous to those used in the proof of Lemma  \ref{lem: rates and compacness}\ref{lem: rates and compacness:: Bhatstar Bhat inverse rate}\ref{lem: rates and compacness:: expression of P and Phat}. \\
		
		$(iv)$ Similarly to the proof of $(i)$, for any $h\in\mathcal{H}$, 
		\begin{align*}
			(\widehat{\A}h)(t) - (\A h)(t)=& \operatorname{E}_n\frac{\xi_i}{\overline{\xi}}h(Z_i) \exp(\textbf{i}(X_i^T,W_i)t)- (\A h)(t)\\
			=& \operatorname{E}_n \frac{\xi_i-1}{\overline{\xi}} h(Z_i)  \exp(\textbf{i}(X_i^T,W_i)t)\\
			&+ \operatorname{E}_n\frac{1-\overline{\xi}}{\overline{\xi}}h(Z_i) \exp(\textbf{i}(X_i^T,W_i)t)\\
			&+ (\widehat{\A}h)(t) - (\A h)(t)\\
			=:& (\widetilde{\A}_1 h)(t) + (\widetilde{\A}_2 h)(t)  +(\widehat{\A}h)(t) - (\A h)(t)\, .
		\end{align*}
		By Lemma  \ref{lem: rates and compacness} \ref{lem: rates and compacness:: Ahat rate}, $\|\widehat{\A} - \A\|_{op}=O_P(n^{-1/2})$. Proceeding as in Equation \eqref{eq: B tilde 2}, we obtain $\|\widetilde{\A}_2\|_{op} = O_P(n^{-1/2})$. Thus, to conclude the proof, it remains to show that  $\|\widetilde{\A}_1\|_{op}=O_P(n^{-1/2})$. As argued in the proof of  Lemma \ref{lem: rates and compacness}\ref{lem: rates and compacness:: compactness of A and T}, for any $h\in\mathcal{H}$ there exists a sequence $(\beta_j)_j\in\ell^2$ such that $h(z)=\sum_{j=1}^\infty\beta_j \widetilde \phi_j(z)$, with $\sum_{j=1}^\infty |\widetilde \phi_j(z)|^2\leq K(z,z)$ and $\|h\|_{\mathcal{H}}^2=\sum_{j=1}^\infty\beta_j^2$. Thus,   
		\begin{align*}
			\|\widetilde{\A}_1h\|^2_\mu=&\frac{1}{(\overline{\xi})^2}\int_{ }\Big|\operatorname{E}_n (\xi_i-1)\sum_{j=1}^\infty\beta_j \widetilde \phi_j(Z_i)\exp(\textbf{i}(X_i^T,W_i)t)\Big|^2 d \mu(t)\\
			=& \frac{1}{(\overline{\xi})^2}\int_{ }\Big|\sum_{j=1}^\infty\beta_j \operatorname{E}_n (\xi_i-1)\widetilde \phi_j(Z_i)\exp(\textbf{i}(X_i^T,W_i)t)\Big|^2 d \mu(t)\\
			\leq & \frac{1}{(\overline{\xi})^2}\int_{ } \sum_{j=1}^\infty\beta_j^2 \sum_{j=1}^\infty\Big|\operatorname{E}_n (\xi_i-1)\widetilde \phi_j(Z_i)\exp(\textbf{i}(X_i^T,W_i)t)\Big|^2 d \mu(t)\\
			=& \frac{1}{(\overline{\xi})^2} \|h\|^2_{\mathcal{H}} \int_{ }\sum_{j=1}^\infty \Big|\operatorname{E}_n (\xi_i-1)\widetilde \phi_j(Z_i)\exp(\textbf{i}(X_i^T,W_i)t)\Big|^2 d \mu(t)\, .
		\end{align*}
		Accordingly,
		\begin{align*}
			\|\widetilde{\A}_1\|^2_{op}=\sup_{h\in\mathcal{H},\|h\|_{\mathcal{H}}=1}\|\widetilde{\A}_1h\|^2_\mu\leq& \frac{1}{(\overline{\xi})^2}  \int_{ }\sum_{j=1}^\infty \Big|\operatorname{E}_n (\xi_i-1)\widetilde \phi_j(Z_i)\exp(\textbf{i}(X_i^T,W_i)t)\Big|^2 d\mu(t)\\
			=& O_P(1)\int_{ }\sum_{j=1}^\infty \Big|\operatorname{E}_n (\xi_i-1)\widetilde \phi_j(Z_i)\exp(\textbf{i}(X_i^T,W_i)t)\Big|^2 d\mu(t)\,,
		\end{align*}
		where the last equality follows from $\overline \xi=1+o_P(1)$, which is implied by $\operatorname{E}\xi=1$ (see Assumption \ref{as: bootstrap weights}). Since the bootstrap weights $\{\xi_i:i=1,\ldots,n\}$ are independent from the sample data,  
		$\operatorname{E}_n (\xi_i-1)\widetilde \phi_j(Z_i)\exp(\textbf{i}(X_i^T,W_i)t)$ has zero expectation. Using this, together with $\operatorname{Var}\{\xi\}=1$ (Assumption \ref{as: bootstrap weights}) and $\sum_{j=1}^\infty \operatorname{E}|\widetilde \phi_j(Z_i)|^2\leq \operatorname{E}K(Z_i,Z_i)<\infty$, we obtain 
		\begin{align*}
			\!\!\!\!\!
			\operatorname{E}\int_{ }\sum_{j=1}^\infty \Big|\operatorname{E}_n (\xi_i-1)\widetilde \phi_j(Z_i)\exp(\textbf{i}(X_i^T,W_i)t)\Big|^2 d\mu(t)
			& =
			\int_{ }\sum_{j=1}^\infty \operatorname{E}\Big|\operatorname{E}_n (\xi_i-1)\widetilde \phi_j(Z_i)\exp(\textbf{i}(X_i^T,W_i)t)\Big|^2 d\mu(t)\\
			& =
			\int_{ }\sum_{j=1}^\infty \operatorname{Var}\Big[\operatorname{E}_n (\xi_i-1)\widetilde \phi_j(Z_i)\exp(\textbf{i}(X_i^T,W_i)t)\Big] d\mu(t)\\
			& =
			\int_{ }\sum_{j=1}^\infty \frac{1}{n}\operatorname{Var}\Big[(\xi_i-1)\widetilde \phi_j(Z_i)\exp(\textbf{i}(X_i^T,W_i)t)\Big] d\mu(t)\\
			& \leq 
			\int_{ }\sum_{j=1}^\infty \frac{1}{n}\operatorname{E}\Big[(\xi_i-1)^2\widetilde \phi_j(Z_i)^2\exp(\textbf{i}(X_i^T,W_i)t)^2\Big] d\mu(t) \\
			& \leq
			\sum_{j=1}^\infty \frac{1}{n}\operatorname{E}\Big[(\xi_i-1)^2\widetilde \phi_j(Z_i)^2\Big] \\
			& =
			\sum_{j=1}^\infty \frac{1}{n}\operatorname{E}[(\xi_i-1)^2]\operatorname{E}[\widetilde \phi_j(Z_i)^2]\\
			& =
			\frac{1}{n}\operatorname{Var}\{\xi_i\} \sum_{j=1}^\infty \operatorname{E}[\widetilde \phi_j(Z_i)^2]\\
			& \leq
			\frac{1}{n} \operatorname{E}K(Z_i,Z_i)\\
			& = O(n^{-1})\, .
		\end{align*}
		Gathering the above results, we obtain $\|\widetilde{\A}_1\|_{op}=O_P(n^{-1/2})$. Hence, Part $(iv)$ is proved.\\
		
		$(v)$ By arguments analogous to those used in the proof of Lemma \ref{lem: rates and compacness}\ref{lem: rates and compacness:: compactness of A and T},  $\widehat{\A}_b$ is compact. Since $\widehat{\operatorname{\P}}_b$ is the projection operator onto $\mathcal{R}(\widehat{\B}_b)$, it is bounded. Thus, $(\I - \widehat{\P}_b)$ is also bounded. Consequently, $\widehat{\T}_b=(\I-\widehat{\operatorname{\P}}_b)\widehat{\A}_b$, being the composition of a bounded and a compact operator, is compact.\\
		
		$(vi)$ The proof proceeds similarly to that  of Lemma \ref{lem: rates and compacness}\ref{lem: rates and compacness:: That rate}, by using Part $(iii)$ 
		of the present lemma, $\|\widehat{\B}_b - \B\|_{op}=O_P(n^{-1/2})$, $\|(\widehat{\B}_b^* \widehat{\B}_b)^{-1} - (\B^*\B)^{-1}\|_{op}=O_P(n^{-1/2})$, and $\|\widehat{\A}_b - \A\|_{op}=O_P(n^{-1/2})$ (from Parts $(i)$, $(ii)$, and $(iv)$ of the present lemma).  \\
		
		$(vii)-(viii)$ The proof of Part $(vii)$ follows similarly to that of Part $(i)$ of the present lemma. The proof of Part $(viii)$  follows by arguments analogous to those used in the proof of Lemma \ref{lem: rates and compacness}\ref{lem: rates and compacness:: (I - Phat) shat - That h rate}.
		
	\end{proof}
	
	\begin{lemma}\label{lem: conditional multiplier clt for Bayesian bootsrap}
		Let $\{U_i:i=1,2,\ldots\}$ be an i.i.d. sequence of random variables with $\operatorname{E}U_i^2<\infty$. Let $\{\xi_i:i=1,2,\ldots\}$ be an i.i.d. sequence of random variables independent from $\{U_i:i=1,2,\ldots\}$ satisfying $\operatorname{E}\{\xi_i\}=\operatorname{Var}\{\xi_i\}=1$. Then, conditionally on $\{U_i:i=1,2,\ldots\}$,
		\begin{equation*}
			\sqrt{n}\operatorname{E}_n\left(\frac{\xi_i}{\overline{\xi}}-1\right)U_i\leadsto \mathcal{N}(0,\operatorname{Var}(U))
		\end{equation*}
		for almost all sequences $\{U_i:i=1,2,\ldots\}$.
	\end{lemma}
	\begin{proof}
		This is a standard result in bootstrap theory. A proof can be found, for example, in \citet[Theorem 2.6]{kosorok2008introduction}.
	\end{proof}
	
	\begin{lemma}\label{lem: uniform in lambda}
		Let $h_0\in\mathcal{H}$, and let Assumptions \ref{as: iid and T}, \ref{as: rkhs}(i), and \ref{as: completeness} hold. Assume that $\lambda \rightarrow 0$ and $n \lambda \rightarrow \infty$. 
		Then, (i) $(\T^* \T + \lambda \I )^{-1}\T^* [(\I-\widehat{\P})\widehat s-
		\widehat{\T} h_0]\, $,  (ii) $
		(\T^* \T + \lambda \I )^{-1}(\widehat{\T}^*-
		\T^*)[(\I-\widehat{\P})\widehat s - \widehat{\T}h_0)\,$, (iii) $
		[(\widehat{\T}^* \widehat{\T}+\lambda \I)^{-1} $ $- (\T^* \T
		+ \lambda \I )^{-1} ]\widehat{\T}^* [(\I-\widehat{\P})\widehat s
		- \widehat{\T}h_0]\, $, and (iv) $
		(\widehat{\T}^* \widehat{\T}+\lambda \I)^{-1}\widehat{\T}^*
		\widehat{\T}h_0 - $ $(\T^* \T+\lambda \I)^{-1}\T^*
		\T h_0$ are all $O_P(1/\sqrt{n \lambda})$.
		(v) If moreover $h_0\in\mathcal{R}[(\T^*\T)^{b/2}]$, then $\|(\T^* \T+\lambda \I)^{-1}\T^*
		\T h_0 - h_0\|_{\mathcal{H}}=O(\lambda^{(b\wedge 2)/2})$.
	\end{lemma}
	\begin{proof}
		The proof proceeds along lines similar to those of \citet[Lemma S3.3]{beyhum2023one}.
	\end{proof}
	
	\begin{lemma}\label{lem: RKHS results}
		Let $\mathcal{Z}$ be a compact set, and $K:\mathcal{Z}\times \mathcal{Z}\to \mathbb{C}$ be a continuous, symmetric, positive semidefinite kernel. Let $(\mathcal{H},\left<\cdot,\cdot\right>_\mathcal{H})$ denote the Reproducing Kernel Hilbert Space (RKHS) with reproducing kernel $K$. Then, there exist a sequence $(\eta_j)_j\subset\mathbb{R}_+$ and a sequence of continuous functions $(\phi_j)_j$, with $\phi_j : \mathcal{Z}\to \mathbb C$, such that: 
		\begin{enumerate}[label=(\roman*)]
			\item \label{lem: RKHS results: svd of K}
			\begin{equation*}
				\int_{ \mathcal{Z}}K(z,u)\, \phi_j(u)\, du=\eta_j \phi_j(z)\,\text{ for all }z\in\mathcal{Z}\text{ and }j\in\mathbb{N}
			\end{equation*}
			with 
			\begin{equation*}
				\int_{\mathcal{Z}} \phi_j(z)\, \overline{\phi_s(z)}\, d z =\begin{cases} 1\text{ if }j=s\\
					0 \text{ if }j\neq s
				\end{cases}\, ;
			\end{equation*}
			\item \label{lem: RKHS results: upper bound for l2 sum}
			\begin{equation*}
				K(z,z)\geq \sum_{j=1}^N \eta_j |\phi_j(z)|^2\text{ for all }z\in\mathcal{Z} \text{ and }N\in\mathbb{N}\,;
			\end{equation*}
			\item \label{lem: RKHS results: H eig space}
			for $\widetilde{\phi}_j(z)=\sqrt{\eta}_j \phi_j(z)$, the space 
			\begin{align*}
				\mathcal{H}_{EIG}:=\Big\{f:\mathcal{Z}\to \mathbb C & \text{ s.t. for a sequence } (\beta_j)_j\text{ in }\ell^2\\
				&\, f(z)=\sum_{j=1}^\infty \beta_j \widetilde{\phi}_j(z)\text{ pointwise in }z\in\mathcal{Z}\Big\}
			\end{align*}
			endowed with the inner product 
			\begin{equation*}
				\left<f,g\right>_{EIG}=\sum_{j=1}^\infty \beta_j \gamma_j \text{ (where }g(z)=\sum_{j=1}^\infty \gamma_j \widetilde{\phi}_j(z)\,) 
			\end{equation*}
			is an RKHS with reproducing kernel $K$;
			\item \label{lem: RKHS results: equality of H and H eig space}
			\begin{equation*}
				\left(\mathcal{H}\,,\, \left<\cdot,\cdot\right>_{\mathcal{H}}\right)=\left(\mathcal{H}_{EIG}\,,\, \left<\cdot,\cdot\right>_{EIG}\right)\, .
			\end{equation*}
		\end{enumerate}
	\end{lemma}
	\noindent \textbf{Proof}. The results in $(i)$ and $(ii)$ are known as Mercer's Theorem. The proofs of $(iii)$-$(iv)$ can be found, for example, in \citet[Chapter 12]{wainwright2019high}.
	
	\section{Low level conditions}\label{sec: low level conditions}
	In this Appendix, we provide primitive conditions that imply the high-level assumptions in the main text, namely Assumptions 
	\ref{as: belonging condition}
	and \ref{as: belonging condition for bootstrap}. Let
	\begin{equation*}
		\partial ^{\kappa \kappa}K(z_1,z_2):=\frac{\partial^{2\kappa}}{\partial^\kappa z_1 \partial^\kappa z_2}K(z_1,z_2)\, .
	\end{equation*}
	We introduce the following primitive conditions. 
	\begin{assumption}\label{as: smoothness of K}
		$\partial ^{\kappa \kappa}K$ is well-defined and continuous on $\operatorname{Int}(\mathcal{Z})\times \operatorname{Int}(\mathcal{Z})$ for some $\kappa\geq 2$. 
	\end{assumption}
	
	Let $\partial \mathcal{Z}$ denote the boundary of $\mathcal{Z}$. 
	\begin{assumption}\label{as: fZ null on the boundary of Supp(Z)}
		$f_Z=0$ on $\partial \mathcal{Z}$. 
	\end{assumption}
	
	The following lemma shows that, under these two low-level conditions, the high-level assumptions stated in the main text are satisfied. 
	
	\begin{lemma}\label{lem: high-level assumptions}
		Let $h_0\in\mathcal{H}$, and let Assumptions \ref{as: iid and T}, \ref{as: rkhs}(i), and \ref{as: completeness} hold.  
		Suppose that Conditions \ref{as: smoothness of K} and \ref{as: fZ null on the boundary of Supp(Z)} are satisfied, and that $\lambda \rightarrow 0$ and $n \lambda \rightarrow \infty$. Then, Assumptions 
		\ref{as: belonging condition}
		and \ref{as: belonging condition for bootstrap} hold.  
	\end{lemma}
	\begin{proof}
		We first establish Part (i) of Assumption \ref{as: belonging condition}. By Condition \ref{as: smoothness of K}, $\partial^{\kappa\kappa}K$ is well-defined and continuous for some $\kappa\geq 2$. It then follows from \citet[Corollary 4.36]{steinwart2008support} that every $h\in\mathcal{H}$ is $\kappa$ times continuously differentiable and satisfies
		\begin{equation}\label{eq: bound on derivative of h}
			|h^{(\kappa)}(z)|\leq \|h\|_\mathcal{H}[\partial^{\kappa \kappa}K(z,z)]^{1/2}\text{ for }\kappa=0,1,2\, . 
		\end{equation}
		Since $\widehat h,h_0\in\mathcal{H}$, we obtain
		\begin{equation}\label{eq: hhat kappa - h0 kappa}
			|\widehat h^{(\kappa)}(z) - h^{(\kappa)}_0(z)|\leq \|\widehat h-h_0\|_\mathcal{H}[\partial^{\kappa \kappa}K(z,z)]^{1/2}\, \text{ for }\kappa=0,1,2\, .
		\end{equation}
		Accordingly, 
		\begin{equation*}
			\sup_{z\in\mathcal{Z}}|\widehat h^{(\kappa)}(z) - h^{(\kappa)}_0(z)|\leq \|\widehat h-h_0\|_\mathcal{H}\sup_{z\in\mathcal{Z}}[\partial^{\kappa \kappa}K(z,z)]^{1/2}\, \text{ for }\kappa=0,1,2\, .
		\end{equation*}
		Since by Assumption \ref{as: iid and T}(iii) $\mathcal{Z}$ is compact, and by Condition \ref{as: smoothness of K} $\partial^{\kappa\kappa}K$ is continuous, $\sup_{z\in\mathcal{Z}}[\partial^{\kappa \kappa}K(z,z)]^{1/2}$ is finite for $\kappa=0,1,2$. Moreover, by Lemma \ref{lem: rate for hhat}, $\|\widehat h - h_0\|_{\mathcal{H}}=o_P(1)$. Gathering these results gives  
		\begin{equation*}
			\sup_{z\in\mathcal{Z}}|\widehat h^{(\kappa)}(z) - h^{(\kappa)}_0(z)|=o_P(1)\text{ for }\kappa=0,1,2\, .
		\end{equation*}
		Hence, Part (i) of Assumption \ref{as: belonging condition} is established. Part (i) of Assumption \ref{as: belonging condition for bootstrap} follows by analogous arguments. \\
		We now establish Part (ii) of Assumption \ref{as: belonging condition}. From Equation \eqref{eq: bound on derivative of h},
		\begin{equation}\label{eq: boundendess of h0 kappa}
			\sup_{z\in\mathcal{Z}}|h_0^{(\kappa)}(z)|\leq \|h_0\|_{\mathcal{H}} \sup_{z\in\mathcal{Z}}[\partial^{\kappa \kappa}K(z,z)]^{1/2}\,\text{ for }\kappa=0,1,2\, .
		\end{equation}
		Using \eqref{eq: bound on derivative of h} and $\|\widehat h - h_0\|_{\mathcal{H}}=o_P(1)$ gives 
		\begin{equation}\label{eq: boundendess of hhat kappa}
			\sup_{z\in\mathcal{Z}}|\widehat h^{(\kappa)}(z)|\leq [\|\widehat h-h_0\|_{\mathcal{H}}+\|h_0\|_{\mathcal{H}}] \sup_{z\in\mathcal{Z}}[\partial^{\kappa \kappa}K(z,z)]^{1/2}\overset{(a)}{\leq} 2\|h_0\|_{\mathcal{H}}\sup_{z\in\mathcal{Z}}[\partial^{\kappa \kappa}K(z,z)]^{1/2}]=: C_{\mathcal{H}}\,,
		\end{equation}
		where the inequality $(a)$ holds with probability approaching one. Let us now define the class of functions 
		\begin{equation*}
			\mathcal{G}:=\{f:\mathcal{Z}\to \mathbb{R}\text{ such that }|f^{(\kappa)}|\leq C_{\mathcal{H}}\text{ for }\kappa=0,1\}\, . 
		\end{equation*}
		Then, by \eqref{eq: boundendess of h0 kappa} and \eqref{eq: boundendess of hhat kappa}, $h_0'\in\mathcal{G}$ and $\Pr(\widehat h'\in\mathcal{G})\rightarrow 1$. 
		Applying \citet[Corollary 2.7.2]{van1996weak} with $\alpha=d=1$ yields 
		\begin{equation*}
			\log N_{[\,]}(\epsilon,\mathcal{G},L^2(P))\leq C \epsilon^{-1}\,,
		\end{equation*}
		for a fixed constant $C$ independent of  $\epsilon$. Thus, Part (ii) of Assumption \ref{as: belonging condition} is established. Assumption  \ref{as: belonging condition for bootstrap}(ii) can be obtained by analogous arguments. \\
		Finally,  since $\widehat h$, $\widehat h_b$, and $h_0$ are continuous on $\mathcal{Z}$, Assumptions \ref{as: belonging condition}(iii) and \ref{as: belonging condition for bootstrap}(iii) 
		follow directly from Condition \ref{as: fZ null on the boundary of Supp(Z)}. 
		
	\end{proof}
	
	\newpage
	
	\bibliographystyle{apalike}
	\bibliography{biblio}
	
\end{document}